\documentclass[11pt,english]{article}
\usepackage[T1]{fontenc}
\usepackage[latin9]{inputenc}
\usepackage{xcolor}
\usepackage{units}
\usepackage{amsmath}
\usepackage{amsthm}
\usepackage{amssymb}
\usepackage{graphicx}
\usepackage{setspace}
\usepackage{esint}
\PassOptionsToPackage{normalem}{ulem}
\usepackage{ulem}
\makeatletter

\providecolor{lyxadded}{rgb}{0,0,1}
\providecolor{lyxdeleted}{rgb}{1,0,0}

\DeclareRobustCommand{\lyxsout}[1]{\ifx\\#1\else\sout{#1}\fi}

\theoremstyle{definition}
 \newtheorem{example}{\protect\examplename}
\theoremstyle{definition}
\newtheorem{defn}{\protect\definitionname}
\theoremstyle{plain}
\newtheorem{thm}{\protect\theoremname}
\theoremstyle{plain}
\newtheorem{lem}{\protect\lemmaname}

\usepackage{amsfonts}
\usepackage{graphicx}
\usepackage{color}
\usepackage{comment}
\usepackage{hyperref}
\usepackage{soul}
\usepackage{bbm} 
\usepackage{etoolbox} 
\AtBeginEnvironment{example}{\itshape}
\AtEndEnvironment{example}{\normalfont}

\DeclareMathOperator{\Var}{Var} 
\DeclareMathOperator{\Vol}{Vol}
\DeclareMathOperator{\supp}{supp} 

\global\long\def\dd{\mathrm{d}}
\global\long\def\E{\mathbb{E}}
\global\long\def\trre[#1,#2]{\overset{{\scriptstyle (#2)}}{#1}} 
\global\long\def\I{\mathbbm{1}}
\newcommand{\reals}{\mathbbm{R}}

\newcommand{\real}{\mathrm{Re}}
\newcommand{\imag}{\mathrm{Im}}
\DeclareMathOperator*{\argmax}{arg\,max}

\allowdisplaybreaks
\makeatother

\usepackage{babel}
\providecommand{\definitionname}{Definition}
\providecommand{\examplename}{Example}
\providecommand{\lemmaname}{Lemma}
\providecommand{\theoremname}{Theorem}

\begin{document}
\title{A Toolbox for Refined Information-Theoretic Analyses with Applications}
\author{Neri Merhav and Nir Weinberger}

\maketitle
\newpage{}
\begin{abstract}
This monograph offers a toolbox of mathematical techniques, which
have been effective and widely applicable in information-theoretic
analysis. The first tool is a generalization of the method of types
to Gaussian settings, and then to general exponential families. The
second tool is Laplace and saddle-point integration, which allow to
refine the results of the method of types, and are capable of obtaining
more precise results. The third is the type class enumeration method,
a principled method to evaluate the exact random-coding exponent of
coded systems, which results in the best known exponent in various
problem settings. The fourth subset of tools aimed at evaluating the
expectation of non-linear functions of random variables, either via
integral representations, or by a refinement of Jensen's inequality
via change-of-measure, by complementing Jensen's inequality with a
reversed inequality, or by a class of generalized Jensen's inequalities
that are applicable for functions beyond convex/concave. Various application
examples of all these tools are provided along this monograph.
\end{abstract}
\newpage{}

\tableofcontents{}

\newpage{}

\section{Introduction \label{sec:Introduction}}

This monograph is concerned with a set of analytical tools for information-theoretic
analysis. The use of analytical methods to address challenging combinatorial
problems is a classic idea in math, and includes various widely-used
techniques such as Stirling's approximation, Chernoff's bound, transform
methods (with interchanging summation or integration order), and so
on. Analytical techniques also formed the basis to the inception of
information-theory by Shannon \cite{shannon1948mathematical}: On
the face of it, and even at a deeper look, efficient coding for noisy
channels is a formidable combinatorial problem, in a high dimensional
space. However, Shannon addressed that challenge using an analytical
techniques: 
\begin{enumerate}
\item The asymptotic equipartition property, and the estimation of volumes
in high dimensional spaces, which allows to evaluate the size of high-probability
sets. Then, in the proof of the \emph{noisy channel coding theorem}
for DMCs it is shown that a $n$-dimensional codeword is transmitted,
the set of likely outputs has size roughly given by $e^{nH(Y|X)}$,
where $H(Y|X)$ is the conditional entropy of the channel output $Y$
conditioned on the input $X$, and the total set of likely outputs
has roughly size of $e^{nH(Y)}$ (where $H(Y)$ is the entropy of
$Y$). 
\item The random-coding argument, that establishes the existence of optimal
codes by evaluating the ensemble-average of randomly chosen code,
and forms the basis for achievability (direct) results. 
\item Convexity of information-measures, which is used to establish data-processing
theorems, and consequently forms the basis for impossibility (converse)
results. 
\end{enumerate}
Combining these ideas directly lead, among other results, to the analytical
formula for the capacity of DMCs, given by $C=\max_{P_{X}}I(X;Y)$
(where $I(X;Y)=H(Y)-H(Y|X)$ is the mutual information). These basic
ideas were continuously generalized and refined by numerous authors
along the development of information theory. 

The goal of this manuscript is to follow this path, and propose a
set of advanced analytical tools, which have been affirmed to be efficient
and widely-applicable for information-theoretic problems, allowing
to obtain accurate and refined performance measure characterizations.
Chapters \ref{sec: CMoT} and \ref{sec: laplacesaddlepoint} to follow
address the problem of estimating volumes in high dimensions, first,
via a generalized method of types and, second, via the more advanced
saddle-point method; Chapter \ref{sec:TCE} describes the \emph{type
class enumeration method} (TCEM), a method to tightly analyze the
performance of random-coding ensembles, and Chapter \ref{sec:Manipulating-Expectations-of}
considers various aspects of convexity and Jensen's inequality, mostly
related to the computation of the expected values of non-linear functionals.
We next describe each one of these with more detail.

In Chapter \ref{sec: CMoT}, we describe a generalization of the method
of types \cite{Csiszar98,csiszar2011information}, which was originally
developed for finite alphabets, to Gaussian distributions, which are
distributions over a continuous alphabets, and more generally, to
distributions from exponential families. We introduce the notion of
a typical set with respect to (WRT) a given parametric family of probability
distributions. Such typical sets are defined in a way that the probability
of each vector in the set is roughly the same for all possible distributions
in the defined parametric family. This generalizes both the notion
of weak typicality (a family consisting a single distribution), and
the usual notion of strong typicality for finite alphabets (the family
is the set of all possible PMFs). Moreover, it allows to consider,
e.g., typical sets for the Gaussian distribution. A key property of
typical sets is their \emph{volume}, because if an event of interest
can be represented as the union over typical sets, then its probability
can be accurately determined on the exponential scale using the volume
of these sets. We thus develop a general method to evaluate the volume
of typical sets, and demonstrate its use on memoryless Gaussian sources,
on Gaussian sources conditioned on other vectors, and on Gaussian
sources with memory. We then generalize this method to distributions
from an exponential family. 

While the method of types is a general and widely applicable approach
that leads to useful exponential bounds, there are settings which
require more delicate analysis, and thus, more advanced tools. In
Chapter  \ref{sec: laplacesaddlepoint}, we begin by describing the
Laplace method for integration, and exemplify its use in the problems
of universal coding and extreme-value statistics. We then advance
to the closely-related saddle-point method for integration in the
complex plain, and show how it allows to accurately evaluate the size
of type classes, volumes of hyper-spheres, and large-deviations probabilities,
not only in the exact exponential rate, but also with the exact pre-exponential
factor. We show that this method can be applied beyond parametric
models. We further demonstrate its use for the evaluation of the number
of lattice points in an $L_{1}$ ball, and the evaluation of the volume
of an intersection of a hyper-sphere and hyperplane, refining the
analysis of a Chapter \ref{sec: CMoT}.

In Chapter \ref{sec:TCE}, we proceed to consider random codes. We
introduce the\emph{ }TCEM, which is a principled method for deriving
the error exponent of random codes.\emph{ }We first describe the standard
techniques commonly used to derive bounds on the error exponent, such
as Jensen's inequality and its implications, and various types of
union bounds. While these methods turned out to be effective in the
error-exponent analysis of basic settings such as point-to-point channels
and standard decoding rules, there is no guarantee that they are accurate
in more advanced scenarios. Indeed, we survey various settings in
which these methods are sub-optimal, and do not provide the exact
random-coding error exponent. As an alternative, we show that ensemble-average
error probabilities (and other related performance measures) may be
expressed via \emph{type class enumerators} (TCEs), and specifically
via their (non-integer) moments and tail probabilities. We demonstrate
this both on basic settings as well as more involved ones. We explore
the probabilistic and statistical properties of TCEs, and then survey
a multitude of settings in multi-user information theory, in distributed
compression and in hypothesis testing, for generalized decoding rules
such as those allowing erasures and list outputs, and for the analysis
of the typical random code. We outline how the TCEM is used in each
of these settings, and how it allows to obtain, among other things,
exact error-exponents for optimal decoding rules. In Appendix \ref{sec:Computation-of-the}
we show that the exponents obtained by the TCEM can also be computed
effectively. 

In Chapter \ref{sec:Manipulating-Expectations-of}, we address the
problem of evaluating the expectation of a non-linear function $f(\cdot)$
of a random variable (RV) $X$. In many cases, this function is either
convex or concave, and so a natural course of action is to bound it
using Jensen's inequality. However, there is no guarantee that the
resulting bound is tight enough for the intended application. We present
two general and useful strategies that can be employed in such cases.
The first one is based on finding an \emph{integral representation}
of the function. Then, we interchange the expectation and integral
order, and obtain an alternative expression for $\E\{f(X)\}$. The
technique is useful if computing the inner expectation is simpler
than the original expectation, or if it can be evaluated more accurately.
After evaluating the inner expectation, the expectation $\E\{f(X)\}$
of interest can be computed by solving a one-dimensional integral.
For example, when $f(t)=\ln(t)$, this allows to replace the evaluation
of the expected logarithm with its moment-generating function (MGF).
This is especially appealing since if $X=\sum_{i=1}^{n}X_{i}$ is
the sum of $n$ independent and identically distributed (IID) RVs,
then its MGF is the $n$-th power of the MGF of just one of them.
In accordance, this transforms the original expectation, which is
an integral in $\reals^{n}$, to a one-dimensional integral. We focus
on the logarithmic function $f(t)=\ln(t)$ (and its integer powers),
as well as the power function $f(t)=t^{\rho}$ for some $\rho>0$
(even non-integer), and exemplify the use of this technique in an
multitude of problems such as differential entropy for generalized
multivariate Cauchy densities, ergodic capacity of the Rayleigh single-input
multiple-output (SIMO) channel, moments of guesswork, and moments
of estimation error.

The second strategy preserves the use of Jensen's inequality and thus
exploits convexity or concavity properties, however, it goes beyond
the vanilla Jensen's inequality. This strategy may come in various
flavors. First, a change of measure can be performed before deploying
Jensen's inequality, and then the alternative measure can be optimized
over a given class to improve the bound. As a notable example, when
$f(t)=\ln(t)$, this reproduces the Donsker--Varadhan variational
characterization of the Kullback--Leibler (KL) divergence. Second,
one may use Jensen's inequality, but accompany it with an inequality
in the opposite direction, \emph{i.e.}, a reverse Jensen's inequality
(RJI), in order to evaluate its tightness. We provide a few techniques,
all which rely on a general form of such a RJI. Third, the ``supporting-line''
approach used to prove Jensen's inequality may be generalized to cases
in which the the function whose expected value is sought of is not
convex/concave, but takes a more complicated form, such as the composition
or a multiplication of a different function with a convex/concave
function. A generalized version of Jensen's inequality can still be
derived, by properly optimizing the supporting line. We exemplify
the use of these technique in various problems involving evaluation
of data compression performance and capacities.

Overall, we present a diverse toolbox of analytical techniques, indispensable
to any information-theorist aiming to obtain tight and accurate results.
This monograph was invited and written following a plenary talk, by
the first author, at the 2023 IEEE International Symposium on Information
Theory (ISIT 2023), Taipei, Taiwan, June 25-30, 2023. The title of
the talk was ``My little hammers and screwdrivers for analyzing code
ensemble performance''. It should be pointed out that some of the
proposed techniques (like Chapters \ref{sec: CMoT}, \ref{sec:TCE},
and many parts of Chapter \ref{sec:Manipulating-Expectations-of})
are original, while others are not new (like Chapter \ref{sec: laplacesaddlepoint}). 

\newpage{}

\section{Extension of the Method of Types to Continuous Alphabets \label{sec: CMoT}}

\subsection{Introduction \label{subsec: CMoTintro}}

In their renowned 1981 book \cite{csiszar2011information}, Csisz{\'a}r
and K{\"o}rner introduced the groundbreaking concept of the method
of types. This method has since emerged as a cornerstone within classical
Shannon theory, offering a remarkably potent and versatile mathematical
analytical tool-set. Its primary application lies in establishing
coding theorems -- predominantly their achievability parts, while
occasionally encompassing converse parts as well. Additionally, this
method's utility extends to the evaluation of error probability exponential
decay rates (referred to as error exponents) and the exponential growth
rates of subsets of sequences as functions of the block length (or
the dimension).

The method of types serves as a fundamental combinatorial approach,
originally crafted for memoryless sources and channels with finite
alphabets. In essence, this method involves partitioning the space
of all $q^{n}$ $q$-ary sequences of length $n$ into distinct equivalence
classes termed\emph{ type classes}. Each type class encompasses sequences
that share an identical empirical distribution, characterized by a
specific array of relative frequencies pertaining to the $q$ alphabet
letters. An alternative perspective on type classes is that within
each such class, any sequence can be derived through permutations
of other sequences. The strength of the method of types emanates from
a concept of elegant simplicity: Despite the exponential growth of
each type class's size with $n$ (its exponential rate being determined
by the entropy of the corresponding empirical distribution), the diversity
of distinct type classes experiences only a polynomial growth with
$n$. This interplay of growth dynamics yields a crucial outcome:
The likelihood of any event expressed as a union of type classes is
dominated by the exponential behavior driven by the most probable
type class contained within the event. Similarly, when dealing with
the size of a set defined as a union of type classes, this size experiences
an exponential dominance dictated by the largest type class within
that set.

In \cite{Csiszar98}, Csisz{\'a}r provides an extensively comprehensive
exploration of the method of types. This scholarly work not only encompasses
foundational principles, but also delves into numerous applications.
These applications span a wide spectrum, including the derivation
of error exponents for source coding, channel coding, source-channel
coding, hypothesis testing, the type covering lemma, the packing lemma,
the capacity evaluation for arbitrarily varying channels, rate-distortion
coding, as well as multi-terminal source and channel coding theorems.
Within the same publication, Csisz{\'a}r undertakes a meticulous
survey of several notable extensions to the method of types. Foremost
among these are second-order and higher-order types, with recognition
attributed to prior work by Billingsley \cite{Billingsley61}, Boza
\cite{Boza71}, Whittle \cite{Whittle55}, Davisson, Longo, and Sgarro
\cite{DLS81}, as well as Natarajan \cite{Natarajan85}. Furthermore,
the exploration extends to finite-state types, with acknowledgment
directed to Weinberger, Merhav, and Feder \cite{WMF94}. Csisz{\'a}r's
comprehensive survey \cite{Csiszar98} ends with a section addressing
continuous alphabets. This section's outset acknowledges that extensions
of the type concept to continuous alphabets remain largely uncharted.
It proceeds to navigate this challenge by adopting a discretization
strategy through fine quantization. Nonetheless, this approach reveals
vulnerabilities when grappling with probability density functions
(PDFs) that are supported by the entirety of the real line or half
of it. In such cases, achieving arbitrarily high resolution quantization,
a requisite of the traditional method of types, becomes unattainable.
While acknowledging that coarsely quantizing the tails of distributions
generally entails minimal impact due to their low probabilities, certain
technical intricacies arise, particularly concerning the uniformity
of convergence across a class of distributions. This concern becomes
particularly salient when confronted with the need to interchange
limit operations, such as the limits as $n$ grows large and the quantization
resolution increases concurrently. Furthermore, the cost associated
with achieving high-resolution quantization manifests as an escalated
computational workload in the calculation of the desired exponential
rate. This is due to the fact that the number of free parameters to
optimize is equal to the number of quantization levels minus one.

Within this chapter, our central proposition emerges: The extension
of type classes and the critical components of the method of types
into the realm of continuous alphabets is not only viable but also
remarkably intuitive. This assertion is particularly pertinent when
considering PDFs originating from the broader exponential family \cite{me89},
\cite{MKLS94}, and especially when dealing with the Gaussian PDF,
as expounded in references such as \cite{AM98,HM15,HSMM19,me93,me13,me19,TM23}.
Notably, our approach circumvents the need for the discrete approximations
proposed in \cite{Csiszar98}.

Our methodology revolves around the partitioning of the space of $n$-sequences
into equivalence classes, referred to as type classes, in analogy
to their finite-alphabet counterparts. This construction retains two
pivotal attributes, analogous to their roles in the customary finite-alphabet
context: 
\begin{enumerate}
\item It is possible to devise a computable expression that characterizes
the exponential growth rate of each type class's size or volume as
a function of $n$. This expression, which is always a certain form
of entropy or differential entropy, remains amenable to calculation
independently of $n$ and aligns with the concept of\emph{ single-letter
expression} in the jargon of information theorists. 
\item The array of distinct type classes relevant to the problem at hand
exhibits sub-exponential growth WRT $n$. This assures that the quantity
of distinct types relevant to our problem expands in a manner manageable
for analysis. 
\end{enumerate}
By ``computable expression'' in the first point, we refer to an
expression whose computational complexity remains fixed as $n$ varies.
In relation to the second point, when we mention types ``relevant
to the problem at hand,'' we imply scenarios where the aggregate
number of distinct type classes might conceivably be boundless, yet
the vast majority beyond a sub-exponential subset hold minimal importance
and can be disregarded, given their inconsequential collective impact
on the quantity of interest. This might be due to their associated
probabilities being negligibly small.

In the upcoming sections, we embark on a concise exploration of the
fundamental concepts that underlie the extension of the method of
types to encompass continuous alphabets. Our journey begins with the
Gaussian scenario before encompassing the broader domain of exponential
families. Throughout these discussions, we will interweave illustrative
examples to provide practical context for the concepts being elucidated.

\subsection{Various Definitions of Type Classes \label{subsec: typedefs}}

As mentioned earlier, in the memoryless, finite-alphabet case, we
define a type class as the set of all sequences that share the same
empirical distribution. More precisely, given a $q$-ary sequence
$\boldsymbol{x}=(x_{1},x_{2},\ldots,x_{n})$, with $x_{i}\in{\cal X}$,
$i=1,2,\ldots,n$, ${\cal X}$ being a finite alphabet of size $q$,
the empirical distribution, $\hat{P}_{\boldsymbol{x}}$, associated
with $\boldsymbol{x}$ is the vector $\{\hat{P}_{\boldsymbol{x}}(x),~x\in{\cal X}\}$,
where $\hat{P}_{\boldsymbol{x}}(x)=n_{\boldsymbol{x}}(x)/n$, $n_{\boldsymbol{x}}(x)$
being the number of occurrences of the letter $x\in{\cal X}$ in $\boldsymbol{x}$.
Thus, the type of $\boldsymbol{x}$, ${\cal T}_{n}(\boldsymbol{x})$
is defined by 
\begin{equation}
{\cal T}_{n}(\boldsymbol{x})\triangleq\left\{ \boldsymbol{x}'\in{\cal X}^{n}\colon\hat{P}_{\boldsymbol{x}'}=\hat{P}_{\boldsymbol{x}}\right\} .
\end{equation}
An alternative, equivalent definition of ${\cal T}_{n}(\boldsymbol{x})$
is as the set of all $\boldsymbol{x}'\in{\cal X}^{n}$ that can be
obtained as permutations of $\boldsymbol{x}$. Since ${\cal T}_{n}(\boldsymbol{x})$
corresponds to a particular empirical probability distribution, say,
$\hat{P}$, it would be sometimes convenient to denote it by ${\cal T}_{n}(\hat{P})$.
Similar notation applies to type classes of pairs of $n$-vectors,
$(\boldsymbol{x},\boldsymbol{y})$ (and triples, etc.), where in the
alternative notation, $\hat{P}$ is understood to be the joint empirical
distribution.

Clearly, the two previously provided definitions hold specifically
within the realm of finite-alphabet memoryless systems. However, when
considering the more general scenario, a broader definition becomes
necessary. The essential requirement for formulating a comprehensive
method of types is that sequences falling within the same type class
exhibit matching probabilities, particularly in the exponential sense.
In cases where the data may be governed by a single probability distribution
(or PDF, in continuous scenarios) denoted as $P$, the following definition
holds: 
\begin{equation}
{\cal T}_{n}(P)\triangleq\left\{ \boldsymbol{x}\in{\cal X}^{n}\colon-\frac{\ln P(\boldsymbol{x})}{n}=H\right\} ,
\end{equation}
Here, $H$ represents a constant, which in discrete scenarios, signifies
the entropy rate of distribution $P$, while in continuous contexts,
it signifies the differential entropy rate. This definition underscores
the fundamental property that all sequences within a given type class
share a consistent probabilistic behavior. It encapsulates the notion
that their probabilities, when viewed through the lens of logarithmic
scaling, converge to a common value, thereby enabling a more inclusive
approach in diverse scenarios. In certain instances, intricate technical
nuances necessitate the incorporation of a certain tolerance factor,
denoted as $\epsilon>0$. This becomes particularly pertinent in continuous
scenarios, as we will soon delve into.\footnote{In the next chapter, when we explore the saddle-point method, we will
see how to circumvent the need for this tolerance factor.} This leads us to introduce the notion of an \emph{$\epsilon$-inflated
type class}, represented as follows: 
\begin{equation}
{\cal T}_{n,\epsilon}(P)\triangleq\left\{ \boldsymbol{x}\in{\cal X}^{n}\colon\bigg|-\frac{\ln P(\boldsymbol{x})}{n}-H\bigg|\le\epsilon\right\} .
\end{equation}
These definitions of type classes are aligned with the concept of\emph{
weak typicality}. However, there are instances where we require this
property of almost equal log-probabilities (or log-densities) not
solely for a one specific source $P$, but concurrently for all sources
within a given class. Consider a parametric family of sources, $\{P_{\theta}\colon\theta\in\Theta\}$,
where $\theta$ is the parameter and $\Theta$ is the parameter space.
We define the type class of $\boldsymbol{x}$ WRT the class $\{P_{\theta}\colon\theta\in\Theta\}$
(see also \cite{MW04}) as 
\begin{equation}
{\cal T}_{n}(\boldsymbol{x})\triangleq\left\{ \boldsymbol{x}'\in{\cal X}^{n}\colon P_{\theta}(\boldsymbol{x}')=P_{\theta}(\boldsymbol{x}),\;\forall\theta\in\Theta\right\} .
\end{equation}
Indeed, when the set $\{P_{\theta}\colon\theta\in\Theta\}$ encompasses
all memoryless sources with a given finite alphabet ${\cal X}$ of
size $q$, the parameter vector $\theta$ can be construed as the
vector comprising $q-1$ letter probabilities, with the $q$-th probability
completing their sum to unity. This alignment of definitions corresponds
precisely with the conventional characterization of a type class for
memoryless sources. The rationale underlying this correspondence stems
from the fact that the probability of a sequence $\boldsymbol{x}$
under any memoryless source depends on $\boldsymbol{x}$ solely via
the empirical distribution $\hat{P}$ associated with $\boldsymbol{x}$.
As a result, any two sequences sharing the same empirical distribution
must invariably possess identical probabilities across all memoryless
sources indexed by distinct $\theta$ values. In essence, this expansive
definition of a type class seamlessly envelops the well-established
definition applicable to memoryless sources, effectively encompassing
it as a specific case. More generally, the $\epsilon$-inflated type
class of $\boldsymbol{x}$ is defined as 
\begin{equation}
{\cal T}_{n,\epsilon}(\boldsymbol{x})\triangleq\left\{ \boldsymbol{x}'\in{\cal X}^{n}\colon\bigg|\frac{\ln P_{\theta}(\boldsymbol{x}')}{n}-\frac{\ln P_{\theta}(\boldsymbol{x})}{n}\bigg|\le\epsilon,~\forall\theta\in\Theta\right\} .
\end{equation}
These definitions align with the concept of \emph{strong typicality}.
It is evident that broadening the scope of reference sources, achieved
by expanding the parametric family, causes the type classes to contract.
Conversely, focusing on a subset of $\{P_{\theta}\colon\theta\in\Theta\}$
results in the expansion of type classes. At the far end of this spectrum,
when the subclass of sources becomes a singleton, we encounter the
concept of weak typicality.

As a pertinent example that illustrates this, consider the class of
memoryless, zero-mean Gaussian sources parameterized by the variance,
denoted as $\theta=\sigma^{2}$. The PDF for this class is expressed
as follows: 
\begin{equation}
P_{\sigma^{2}}(\boldsymbol{x})=\frac{\exp\left\{ -\sum_{i=1}^{n}x_{i}^{2}/(2\sigma^{2})\right\} }{(2\pi\sigma^{2})^{n/2}}.
\end{equation}
Since $P_{\sigma^{2}}(\boldsymbol{x})$ depends on $\boldsymbol{x}$
only via $\sum_{i=1}^{n}x_{i}^{2}$, it is clear that all sequences
$\{\boldsymbol{x}\}$ with a given norm (\emph{i.e.}, all sequences
pertaining to points on the surface of a given Euclidean hyper-sphere
centered at the origin) have the same PDF, $P_{\sigma^{2}}(\boldsymbol{x})$.
Thus, a natural definition of type classes WRT the class of zero-mean,
memoryless Gaussian sources parameterized by their variance, also
known as\emph{ Gaussian types} (or ``power types''), are surfaces
of $n$-dimensional hyper-spheres centered at the origin. Expanding
this parametric class by introducing a mean parameter $\mu$ leads
us to consider $\theta=(\sigma^{2},\mu)$ ($\sigma^{2}>0$, $\mu\in\reals$).
Consequently, the PDF becomes: 
\begin{equation}
P_{\sigma^{2},\mu}(\boldsymbol{x})=\frac{\exp\left\{ -\sum_{i=1}^{n}(x_{i}-\mu)^{2}/(2\sigma^{2})\right\} }{(2\pi\sigma^{2})^{n/2}}.
\end{equation}
In this context, $P_{\sigma^{2},\mu}(\boldsymbol{x})$ depends on
$\boldsymbol{x}$ exclusively through $\sum_{i=1}^{n}x_{i}^{2}$ and
$\sum_{i=1}^{n}x_{i}$. Accordingly, the definition of a type class
involves the intersection of a hyper-sphere surface defined by a particular
radius and a hyper-plane defined by a specific value of $\sum_{i=1}^{n}x_{i}$.
This type class is notably smaller compared to the type class relative
to $\{P_{\sigma^{2}}(\boldsymbol{x}),~\sigma^{2}>0\}$, which was
solely defined by the hyper-sphere surface without any additional
intersection with a hyper-plane.

More generally, consider a parametric class of memoryless sources
that form an \emph{exponential family} (see, e.g., Lehmann \cite[Section 1.4]{Lehmann83}).
This means that the single-letter marginal is of the form, 
\begin{equation}
P_{\theta}(x)=\frac{\exp\left\{ \sum_{j=1}^{k}\theta_{j}\phi_{j}(x)\right\} }{Z(\theta)},
\end{equation}
where $\theta=(\theta_{1},\ldots,\theta_{k})$ is the parameter vector,
$\phi_{i}:{\cal X}\to\reals$ are given functions, and $Z(\theta)$
is a normalization constant, given by 
\begin{equation}
Z(\theta)\triangleq\sum_{x\in{\cal X}}\exp\left\{ \sum_{j=1}^{k}\theta_{j}\phi_{j}(x)\right\} ,
\end{equation}
in the discrete case, or 
\begin{equation}
Z(\theta)\triangleq\int_{{\cal X}}\exp\left\{ \sum_{j=1}^{k}\theta_{j}\phi_{j}(x)\right\} \dd x,
\end{equation}
in the continuous case. Moving on to $n$-sequences, by considering
the product form, 
\begin{equation}
P_{\theta}(\boldsymbol{x})=\prod_{i=1}^{n}P_{\theta}(x_{i})=\frac{\exp\left\{ \sum_{j=1}^{k}\theta_{j}\sum_{i=1}^{n}\phi_{j}(x_{i})\right\} }{[Z(\theta)]^{n}}.
\end{equation}
Here, type classes are defined by a given combination of values of
the statistics, $\sum_{i=1}^{n}\phi_{j}(x_{i})$, for $j=1,\ldots,k$.
The class of memoryless Gaussian sources parameterized by $\sigma^{2}$
only, is an exponential family with $k=1$, a transformed parameter,
$\theta=\theta_{1}=-\frac{1}{2\sigma^{2}}$, $\phi_{1}(x)=x^{2}$
and accordingly, $Z(\theta)=\sqrt{2\pi\sigma^{2}}=\sqrt{-\pi/\theta}$.
The broader class, parameterized by $(\sigma^{2},\mu)$, is also an
exponential family with $k=2$, $\theta_{1}=-\frac{1}{2\sigma^{2}}$,
$\theta_{2}=\frac{\mu}{\sigma^{2}}$, $\phi_{1}(x)=x^{2}$, $\phi_{2}(x)=x$,
and 
\begin{equation}
Z(\theta)=\sqrt{2\pi\sigma^{2}}\exp\left\{ \frac{\mu^{2}}{2\sigma^{2}}\right\} =\sqrt{-\frac{\pi}{\theta_{1}}}\exp\left\{ -\frac{\theta_{2}^{2}}{4\theta_{1}}\right\} .
\end{equation}
The class of $q$-ary memoryless sources with ${\cal X}=\{1,2,\ldots,q\}$
is yet another example of an exponential family with $k=q-1$, a parameter
transformation, 
\begin{equation}
\theta_{j}=\ln\left(\frac{p_{j}}{1-\sum_{l=1}^{q-1}p_{l}}\right),~~~~j=1,2,\ldots,q-1,
\end{equation}
\begin{equation}
\phi_{j}(x)=\begin{cases}
1, & x=j\\
0, & x\neq j
\end{cases}
\end{equation}
and 
\begin{equation}
Z(\theta)=\frac{1}{1-\sum_{j=1}^{q-1}p_{j}}.
\end{equation}
In summary, we observe that exponential families are general enough
to include at least two important special cases of memoryless sources:
Finite-alphabet memoryless sources and Gaussian memoryless sources,
but of course they include much more \cite[Section 1.4]{Lehmann83}.

The method of types for exponential families is useful for assessing
the exponential order of certain sums or integrals (depending on whether
the alphabet is discrete or continuous) of functions that depend on
$\boldsymbol{x}$ only via the set of statistics $\{\phi_{j},j=1,2,\ldots,k\}$,
\emph{i.e.}, 
\begin{equation}
\sum_{\boldsymbol{x}}f\left(\sum_{i=1}^{n}\phi_{1}(x_{i}),\sum_{i=1}^{n}\phi_{2}(x_{i}),\ldots,\sum_{i=1}^{n}\phi_{k}(x_{i})\right)
\end{equation}
in the discrete alphabet case, or 
\begin{equation}
\int_{\reals^{n}}f\left(\sum_{i=1}^{n}\phi_{1}(x_{i}),\sum_{i=1}^{n}\phi_{2}(x_{i}),\ldots,\sum_{i=1}^{n}\phi_{k}(x_{i})\right)\dd\boldsymbol{x}
\end{equation}
in the continuous alphabet case. Most notably, the method is useful
when $f$ is an exponential function of $\{\phi_{j},~j=1,2,\ldots,k\}$,
for example, an exponential function of a linear combination of $\{\phi_{j}\}$,
possibly multiplied by an indicator function for the event that the
vector $\{\phi_{j},~j=1,\ldots,k\}$ lies in a certain region in $\reals^{k}$.

In the sequel, we will also see that the exponential structure lends
itself also to handle sources with certain structures of memory, most
notably, Markov sources, where 
\begin{equation}
P_{\theta}(\boldsymbol{x})=\frac{\exp\left\{ \sum_{j=1}^{k}\theta_{j}\sum_{i=0}^{n-1}\phi_{j}(x_{i},x_{i+1})\right\} }{Z_{n}(\theta)},\label{eq: expomarkov}
\end{equation}
and so, type classes are defined according to a given combination
of values of the statistics $\sum_{i=0}^{n-1}\phi_{j}(x_{i},x_{i+1})$,
as an extension of finite-alphabet Markov types \cite{Billingsley61,Boza71,Whittle55,DLS81,Natarajan85}.

Another related useful concept is the notion of a \emph{conditional
type class}. In the finite alphabet case, the conditional type class
of $\boldsymbol{y}\in{\cal Y}^{n}$ given $\boldsymbol{x}\in{\cal X}^{n}$
(where both ${\cal X}$ and ${\cal Y}$ are finite alphabets), is
defined as the set of all $\{\boldsymbol{y}'\}$ such that empirical
joint distribution $\hat{P}_{\boldsymbol{xy'}}$ is equal to $\hat{P}_{\boldsymbol{xy}}$.
A natural parallel extension of conditional type classes to the continuous
alphabet case, is defined WRT an exponential family of conditional
distributions (in the discrete case) or conditional PDFs (in the continuous
case). In the memoryless case, we define the single-letter conditional
probability function pertaining to an exponential family, as
\begin{equation}
P_{\theta}(x|y)=\frac{\exp\left\{ \sum_{j=1}^{k}\theta_{j}\phi_{j}(x,y)\right\} }{Z(y,\theta)},\label{eq: expochannel}
\end{equation}
with $Z(y,\theta)$ being a normalization constant such that $P_{\theta}(x|y)$
sums/integrates (over $x$) to unity. Here, the conditional type class
of $\boldsymbol{x}$ given $\boldsymbol{y}$ is the set of all $\{\boldsymbol{x}'\}$
such that for the given $\boldsymbol{y}$, $\sum_{i=1}^{n}\phi_{j}(x_{i}',y_{i})=\sum_{i=1}^{n}\phi_{j}(x_{i},y_{i})$,
for all $j=1,2,\ldots,k$. For example, a Gaussian conditional type
class is defined WRT the class 
\begin{equation}
P_{\sigma^{2},a}(x|y)=\frac{\exp\left\{ -(x-ay)^{2}/(2\sigma^{2})\right\} }{\sqrt{2\pi\sigma^{2}}},
\end{equation}
which is a conditional exponential family with $k=2$, $\theta_{1}=-\frac{1}{2\sigma^{2}}$,
$\theta_{2}=\frac{a}{\sigma^{2}}$, $\phi_{1}(x,y)=x^{2}$, $\phi_{2}(x,y)=xy$,
and 
\begin{equation}
Z(y,\theta)=\sqrt{2\pi\sigma^{2}}\exp\left\{ \frac{a^{2}y^{2}}{2\sigma^{2}}\right\} .
\end{equation}
In this case, the conditional type class is defined by prescribed
values of $\sum_{i=1}^{n}x_{i}^{2}$ and $\sum_{i=1}^{n}x_{i}y_{i}$.

In the sequel, we will demonstrate the usefulness of the concepts
of type classes and conditional type classes in several application
examples.

\subsection{Simple Gaussian Types \label{subsec: simpleGausstypes}}

As is widely recognized, the conventional method of employing types,
particularly for finite alphabets, hinges upon having a readily available,
explicit formulation for the exponential growth rate (as a function
of $n$) concerning the size of a given type class. Similarly, when
dealing with the continuous Gaussian scenario, a prerequisite is obtaining
a specific, well-defined expression for the volume of the associated
type class, as delineated in Section \ref{subsec: typedefs}. To commence,
let us consider the simplest scenario -- that of typicality WRT zero-mean,
Gaussian, IID sources, characterized by their variance. In this context,
the corresponding type class, as detailed in Section \ref{subsec: typedefs},
finds its definition in the realm of hyper-spherical surfaces, \emph{i.e.},
\begin{equation}
{\cal T}_{n}(s)\triangleq\left\{ \boldsymbol{x}\in\reals^{n}\colon\frac{1}{n}\sum_{i=1}^{n}x_{i}^{2}=s\right\} ,
\end{equation}
with a slight abuse of notation. Strictly speaking, the volume of
${\cal T}_{n}(s)$ is zero, if viewed as an object in the space $\reals^{n}$,
because its real dimension is $n-1$, as it is the surface area of
a hyper-sphere of radius $\sqrt{ns}$. The surface area of an $n$-dimensional
hyper-sphere of radius $R$ is given by $2\pi^{n/2}R^{n-1}/\Gamma(n/2)$,
where $\Gamma(\cdot)$ is the Gamma function, defined as 
\begin{equation}
\Gamma(u)\triangleq\int_{0}^{\infty}t^{u-1}e^{-t}\dd t,
\end{equation}
whose value for $u=n/2$, ($n$ being a positive integer) is given
by 
\begin{equation}
\Gamma\left(\frac{n}{2}\right)=\begin{cases}
\left(\frac{n}{2}-1\right)!, & n\mbox{ is even}\\
2^{-(n-1)/2}\cdot\sqrt{\pi}\times1\times3\times\ldots(n-2), & n\mbox{ is odd}
\end{cases}.
\end{equation}
Thus, the surface area of an $n$-dimensional hyper-sphere of radius
$\sqrt{ns}$ is the volume of ${\cal T}_{n}(s)$ in $n-1$ dimensions:
\begin{align}
\Vol\left\{ {\cal T}_{n}(s)\right\}  & =\frac{2\pi^{n/2}(\sqrt{ns})^{n-1}}{\Gamma(n/2)}\nonumber \\
 & \sim\frac{2\pi^{n/2}(ns)^{(n-1)/2}}{\sqrt{4\pi/n}(n/2e)^{n/2}}\nonumber \\
 & =\frac{(2\pi es)^{n/2}}{\sqrt{\pi s}},
\end{align}
where the notation $a_{n}\sim b_{n}$, for two positive sequences,
$\{a_{n}\}$ and $\{b_{n}\}$, means that $a_{n}/b_{n}\to1$ as $n\to\infty$.
Here, the second line follows from the Stirling approximation for
sufficiently large $n$. Note that the exponential factor, $(2\pi es)^{n/2}$
is exactly $e^{nh}$, where $h=\frac{1}{2}\ln(2\pi es)$ is the differential
entropy of a Gaussian RV with variance $s$, in other words $\Vol\{{\cal T}_{n}(s)\}$
is of the exponential order of $e^{nh}$ in parallel to the fact that
in the finite-alphabet case, the size of a type class is exponentially
$e^{nH}$, where $H$ is the empirical entropy associated with the
type. This result is not a coincidence and we will encounter it repeatedly
in the sequel.

Now, if our purpose is to integrate over $\reals^{n}$ a certain function
that depends on $\boldsymbol{x}$ only via $\sum_{i=1}^{n}x_{i}^{2}$,
we can proceed as follows: 
\begin{align}
\int_{\reals^{n}}f\left(\sum_{i=1}^{n}x_{i}^{2}\right)\dd\boldsymbol{x} & =\int_{0}^{\infty}\dd R\int_{\{\sum_{i}x_{i}^{2}=R^{2}\}}f\left(\sum_{i=1}^{n}x_{i}^{2}\right)\dd\boldsymbol{x}'\nonumber \\
 & =\int_{0}^{\infty}\dd(\sqrt{ns})\cdot\Vol\left\{ {\cal T}_{n}(s)\right\} f(ns)\nonumber \\
 & \sim\frac{1}{2}\sqrt{\frac{n}{\pi}}\cdot(2\pi e)^{n/2}\int_{0}^{\infty}\dd s\cdot s^{n/2-1}f(ns),
\end{align}
where the inner integration WRT $\boldsymbol{x}'$ in the first line
is over the hyper-sphere surface, whose dimension is $n-1$. We have
thus simplified an $n$-dimensional integral to become a one-dimensional
integral.
\begin{example}
\emph{Consider the calculation of the probability of the event $\sum_{i=1}^{n}X_{i}^{2}\ge nA$,
where $\{X_{i}\}$ are IID, zero-mean, Gaussian RVs with variance
$\sigma^{2}$ and $A>\sigma^{2}$. In this case, 
\begin{align}
\Pr\left\{ \sum_{i=1}^{n}X_{i}^{2}\ge nA\right\}  & =\int_{\reals^{n}}(2\pi\sigma^{2})^{-n/2}\exp\left\{ -\sum_{i=1}^{n}x_{i}^{2}/(2\sigma^{2})\right\} \cdot\I\left\{ \sum_{i=1}^{n}x_{i}^{2}\ge nA\right\} \dd\boldsymbol{x}\nonumber \\
 & \sim\frac{1}{2}\sqrt{\frac{n}{\pi}}(2\pi e)^{n/2}\int_{A}^{\infty}s^{n/2-1}(2\pi\sigma^{2})^{-n/2}\exp\left\{ -\frac{sn}{2\sigma^{2}}\right\} \dd s.
\end{align}
For $A>\sigma^{2}$, this integral is dominated by the value of the
integrand at $s=A$, and therefore, the above expression is of the
exponential order of 
\begin{equation}
\exp\left\{ -\frac{n}{2}\left[\frac{A}{\sigma^{2}}-\ln\left(\frac{A}{\sigma^{2}}\right)-1\right]\right\} .
\end{equation}
}
\end{example}
At this juncture, one might inquire about the necessity of employing
the method of types for the aforementioned example, as well as for
several other instances elaborated upon in the subsequent sections
of this chapter. After all, the exponential rate of the aforementioned
probability can be readily derived through a straightforward application
of the Chernoff bound, renowned for its exponential accuracy. However,
the rationale for employing the method of types, not just in this
simple instance, but also in the forthcoming sections, is threefold: 
\begin{itemize}
\item \emph{General applicability.} While the chosen example was intentionally
simple, serving as an illustrative vehicle for the underlying technique,
the method of types possesses a generality and adaptability that extends
to more intricate scenarios. Consider, for instance, an event that
encompasses a vector of diverse empirical statistics, confined within
a specific spatial region. Such intricate events are beyond the realm
of the Chernoff bound. 
\item \emph{Broad utility.} The capacity to gauge the volume of a type class
holds significance beyond the mere evaluation of probabilities linked
to rare events. Its utility extends to deriving universal hypothesis
testing strategies and universal decoders in instances where the source
and/or channel characteristics are unknown. For a comprehensive understanding,
refer to works such as \cite{me89,me93,me13,me19}, all of which underscore
its importance. This aspect will be elaborated upon in Section \ref{subsec: CMOT applications}. 
\item \emph{Enhanced precision.} Through the utilization of the Laplace
method for one-dimensional integration, we will come to realize in
the upcoming chapter that we can attain not only the accurate exponential
order found in the last integral (akin to the Chernoff bound or general
large-deviations bounds), but also an asymptotically precise pre-exponential
factor. 
\end{itemize}
It is imperative to bear these considerations in mind as we delve
into the subsequent sections of this chapter.

\subsection{More Refined Gaussian Types \label{subsec: RefinedGausstypes}}

Let us proceed to the next phase. Consider a scenario where the function
we need to integrate depends on $\boldsymbol{x}$, not solely through
$\sum_{i=1}^{n}x_{i}^{2}$, but also through $\sum_{i=1}^{n}x_{i}$.
For instance, this arises when calculating the probability of an event
like $\{\sum_{i=1}^{n}(X_{i}-\mu)^{2}\ge nA\}$. In this situation,
we must engage with more refined type classes, defined by specific
values of both $\sum_{i=1}^{n}x_{i}^{2}$ and $\sum_{i=1}^{n}x_{i}$.
In simpler terms, our type class now takes the form of an intersection
between a hyper-sphere surface and a hyper-plane. Unlike the prior
case where we dealt with a simple hyper-sphere, here, an apparent
closed-form formula for the volume of this $(n-2)$-dimensional construct,
defined by $\sum_{i=1}^{n}x_{i}^{2}=ns$ and $\sum_{i=1}^{n}x_{i}=n\mu$
for given constants $s>0$ and $\mu\in\mathbb{R}$ (with $s>\mu^{2}$),
is not readily available. At this juncture, an exact solution to this
challenge remains elusive. Nevertheless, we can furnish an approximation
that can be continually honed as $n$ becomes increasingly large.
This approximation suffices to derive the precise exponential scale
of the desired expression, thereby serving our immediate purpose.
Subsequently, we will acquaint ourselves with more advanced techniques
that, on occasion, enable a significantly more accurate assessment.

Let $\epsilon>0$ be arbitrarily small and consider the $\epsilon$-inflated
version of the type class described above, that is 
\begin{equation}
{\cal T}_{n}(s,\mu,\epsilon)\triangleq\left\{ \boldsymbol{x}\colon\bigg|\frac{1}{n}\sum_{i=1}^{n}x_{i}^{2}-s\bigg|\le\epsilon,~\bigg|\frac{1}{n}\sum_{i=1}^{n}x_{i}-\mu\bigg|\le\epsilon\right\} .
\end{equation}
Consider an auxiliary PDF of $n$ IID Gaussian RVs of mean $\mu$
and variance $s-\mu^{2}$, \emph{i.e.}, 
\begin{equation}
g(\boldsymbol{x})=\frac{\exp\left\{ -\frac{1}{2(s-\mu^{2})}\sum_{i=1}^{n}(x_{i}-\mu)^{2}\right\} }{\left[2\pi(s-\mu^{2})\right]^{n/2}}.
\end{equation}
Then, 
\begin{align}
1 & \ge\int_{{\cal T}_{n}(s,\mu,\epsilon)}g(\boldsymbol{x})\dd\boldsymbol{x}\nonumber \\
 & =\int_{{\cal T}_{n}(s,\mu,\epsilon)}\frac{\exp\left\{ -\frac{1}{2(s-\mu^{2})}\sum_{i=1}^{n}(x_{i}-\mu)^{2}\right\} }{\left[2\pi(s-\mu^{2})\right]^{n/2}}\dd\boldsymbol{x}\nonumber \\
 & =\int_{{\cal T}_{n}(s,\mu,\epsilon)}\frac{\exp\left\{ -\frac{1}{2(s-\mu^{2})}\left[\sum_{i=1}^{n}x_{i}^{2}-2\mu\sum_{i=1}^{n}x_{i}+n\mu^{2}\right]\right\} }{\left[2\pi(s-\mu^{2})\right]^{n/2}}\dd\boldsymbol{x}\nonumber \\
 & \ge\int_{{\cal T}_{n}(s,\mu,\epsilon)}\frac{\exp\left\{ -\frac{1}{2(s-\mu^{2})}\left[n(s+\epsilon)-2\mu\cdot n(\mu-\epsilon\cdot\mbox{sgn}(\mu))+n\mu^{2}\right]\right\} }{\left[2\pi(s-\mu^{2})\right]^{n/2}}\dd\boldsymbol{x}\nonumber \\
 & =\Vol\left\{ {\cal T}_{n}(s,\mu,\epsilon)\right\} \cdot\frac{\exp\left\{ -\frac{n}{2(s-\mu^{2})}\left[s+\epsilon-\mu^{2}+2\epsilon|\mu|\right]\right\} }{\left[2\pi(s-\mu^{2})\right]^{n/2}}\nonumber \\
 & =\Vol\left\{ {\cal T}_{n}(s,\mu,\epsilon)\right\} \cdot\frac{\exp\left\{ -\frac{n\epsilon(2|\mu|+1)}{2(s-\mu^{2})}\right\} }{\left[2\pi e(s-\mu^{2})\right]^{n/2}},
\end{align}
and so, 
\begin{equation}
\Vol\left\{ {\cal T}_{n}(s,\mu,\epsilon)\right\} \le\left[2\pi e(s-\mu^{2})\right]^{n/2}\cdot\exp\left\{ \frac{n\epsilon(2|\mu|+1)}{2(s-\mu^{2})}\right\} .
\end{equation}
To establish a lower bound, consider the application of the weak law
of large numbers (WLLN), which asserts that as $n$ approaches infinity,
the probability of the complement of ${\cal T}_{n}(s,\mu,\epsilon)$
under the PDF $g$ diminishes to zero, for any fixed $\epsilon>0$.
Remarkably, we can even allow $\epsilon$ to tend towards zero, albeit
at a pace that remains gentle relative to the growth of $n$. This
probability can be readily bounded from above by employing either
the Chebyshev inequality or the Chernoff bound. Denote the resultant
upper bound for this probability as $\delta_{n}$. This leads us to
the following expression: 
\begin{align}
1-\delta_{n} & \le\int_{{\cal T}_{n}(s,\mu,\epsilon)}g(\boldsymbol{x})\dd\boldsymbol{x}\nonumber \\
 & =\int_{{\cal T}_{n}(s,\mu,\epsilon)}\frac{\exp\left\{ -\frac{1}{2(s-\mu^{2})}\left[\sum_{i=1}^{n}x_{i}^{2}-2\mu\sum_{i=1}^{n}x_{i}+n\mu^{2}\right]\right\} }{\left[2\pi(s-\mu^{2})\right]^{n/2}}\dd\boldsymbol{x}\nonumber \\
 & \le\int_{{\cal T}_{n}(s,\mu,\epsilon)}\frac{\exp\left\{ -\frac{1}{2(s-\mu^{2})}\left[n(s-\epsilon)-2\mu\cdot n(\mu+\epsilon\cdot\mbox{sgn}(\mu))+n\mu^{2}\right]\right\} }{\left[2\pi(s-\mu^{2})\right]^{n/2}}\dd\boldsymbol{x}\nonumber \\
 & =\Vol\left\{ {\cal T}_{n}(s,\mu,\epsilon)\right\} \cdot\frac{\exp\left\{ -\frac{n}{2(s-\mu^{2})}\left[s-\epsilon-\mu^{2}-2\epsilon|\mu|\right]\right\} }{\left[2\pi(s-\mu^{2})\right]^{n/2}}\nonumber \\
 & =\Vol\left\{ {\cal T}_{n}(s,\mu,\epsilon)\right\} \cdot\frac{\exp\left\{ \frac{n\epsilon(2|\mu|+1)}{2(s-\mu^{2})}\right\} }{\left[2\pi e(s-\mu^{2})\right]^{n/2}},
\end{align}
and so, 
\begin{equation}
\Vol\left\{ {\cal T}_{n}(s,\mu,\epsilon)\right\} \ge(1-\delta_{n})\cdot\left[2\pi e(s-\mu^{2})\right]^{n/2}\cdot\exp\left\{ -\frac{n\epsilon(2|\mu|+1)}{2(s-\mu^{2})}\right\} .
\end{equation}
As we allow $\epsilon$ to approach infinitesimally small values,
we discern that the volume of ${\cal T}_{n}(s,\mu,\epsilon)$ essentially
aligns with the exponential order given by: 
\begin{equation}
\left[2\pi e(s-\mu^{2})\right]^{n/2}=(2\pi es)^{n/2}\cdot\left(1-\frac{\mu^{2}}{s}\right)^{n/2}.
\end{equation}
The first factor, $(2\pi es)^{n/2}$, corresponds to $e^{nh}$, as
we previously deduced in Section \ref{subsec: simpleGausstypes}.
Concurrently, the subsequent factor, $(1-\mu^{2}/s)^{n/2}$, embodies
the volume reduction attributed to the intersection with the ($\epsilon$-inflated)
hyper-plane, namely $n(\mu-\epsilon)\le\sum_{i=1}^{n}x_{i}\le n(\mu+\epsilon)$.
Consequently, it becomes evident that there is no sacrifice in terms
of the exponential order when the hyper-sphere intersects with the
hyper-plane that encompasses the origin ($\mu=0$). Stated differently,
the majority of volume is captured by elements within ${\cal T}_{n}(s,\mu,\epsilon)$
that exhibit a property where the sum of their coordinates is relatively
modest (in absolute value).

As evident, the underpinning of the volume's upper and lower bound
derivation is straightforward, yet this same concept retains its relevance
in more intricate scenarios. When faced with an $\epsilon$-enlarged
type class, characterized by linear and quadratic criteria on $\boldsymbol{x}$,
we construct an auxiliary Gaussian PDF, denoted as $g(\cdot)$, which
exhibits two key attributes: 
\begin{enumerate}
\item The likelihood of the type class under $g(\cdot)$ converges towards
unity as $n$ approaches infinity. 
\item The value of the PDF of all sequences situated within the type class
are virtually the same, differing only exponentially by a factor that
scales with $\epsilon$. This value of the PDF is denoted as $g_{0}$. 
\end{enumerate}
The volume of the type class then aligns with the exponential order
of $1/g_{0}$. In the prior calculation, $g_{0}$ equates to $[2\pi e(s-\mu^{2})]^{-n/2}$,
leading to an exponential volume of $1/g_{0}=[2\pi e(s-\mu^{2})]^{n/2}$.
Given that our objective is to pinpoint the correct exponential order
rather than striving for precise evaluation at this stage, the demand
in the first item mentioned earlier can actually be considerably relaxed.
It is even permissible for the type class probability to approach
zero, as long as the rate of decay remains sub-exponential in $n$.
\begin{example}
\emph{Consider the calculation of the probability of the event $\{\sum_{i=1}^{n}(X_{i}-A)^{2}\ge nB\}$
when $\{X_{i}\}$ are IID, zero-mean Gaussian RVs with variance $\sigma^{2}$.
To this end, we divide the set ${\cal E}\triangleq\{\boldsymbol{x}\colon\sum_{i=1}^{n}(x_{i}-A)^{2}\ge nB\}$
into disjoint $\epsilon$-inflated type classes that together cover
${\cal E}$, where $s$ and $\mu$ are odd integer multiples of $\epsilon$.
These are all $\{{\cal T}_{n}(s,\mu,\epsilon)\}$ with the property
$s-2A\mu+A^{2}>B$ where $s=i\epsilon$ and $\mu=j\epsilon$, $i$
being an odd positive integer and $j$ being an odd integer. To avoid
the necessity of dealing with contributions of infinitely many such
type classes, we proceed as follows. Let us partition the set ${\cal E}$
into two disjoint subsets, ${\cal E}_{1}\triangleq\{\boldsymbol{x}\colon nB\le\sum_{i=1}^{n}(x_{i}-A)^{2}<nC\}$
and ${\cal E}_{2}\triangleq\{\boldsymbol{x}\colon\sum_{i=1}^{n}(x_{i}-A)^{2}\ge nC\}$,
for some $C>B$ arbitrarily large. The idea is that ${\cal E}_{1}$
contains finitely many types, whereas the contribution of ${\cal E}_{2}$
can be upper bounded by simple (crude) bound, which for large enough
$C$, would yield an exponential decay faster than that of ${\cal E}_{1}$,
and so, the contribution of ${\cal E}_{2}$ can be neglected altogether.
We will thus show that $\Pr\{{\cal E}\}\doteq\Pr\{{\cal E}_{1}\}$,
where $\doteq$ denotes equality on the exponential scale, i.e., two
positive sequences $\{a_{n}\}$ and $\{b_{n}\}$ satisfy that $a_{n}\doteq b_{n}$
if $\lim_{n\to\infty}\frac{1}{n}\log\frac{a_{n}}{b_{n}}=0$. Specifically,
first observe that 
\begin{align}
\sum_{i=1}^{n}(x_{i}-A)^{2} & =\sum_{i=1}^{n}x_{i}^{2}-2A\sum_{i=1}^{n}x_{i}+nA^{2}\nonumber \\
 & \le\sum_{i=1}^{n}x_{i}^{2}+2|A|\cdot\bigg|\sum_{i=1}^{n}x_{i}\bigg|+nA^{2}\nonumber \\
 & \le\sum_{i=1}^{n}x_{i}^{2}+2|A|\cdot\sqrt{n\sum_{i=1}^{n}x_{i}^{2}}+nA^{2}\nonumber \\
 & =n\cdot\left(\sqrt{\frac{1}{n}\sum_{i=1}^{n}x_{i}^{2}}+|A|\right)^{2},
\end{align}
where the second inequality follows from the Schwarz--Cauchy inequality.
Therefore, 
\begin{eqnarray}
\Pr\{{\cal E}_{2}\} & = & \Pr\left\{ \sum_{i=1}^{n}(x_{i}-A)^{2}\ge nC\right\} \nonumber \\
 & \le & \Pr\left\{ n\cdot\left(\sqrt{\frac{1}{n}\sum_{i=1}^{n}x_{i}^{2}}+|A|\right)^{2}\ge nC\right\} \nonumber \\
 & = & \Pr\left\{ \sum_{i=1}^{n}x_{i}^{2}\ge n(\sqrt{C}-|A|)^{2}\right\} \nonumber \\
 & \le & \exp\left\{ -\frac{n}{2}\left[\frac{(\sqrt{C}-|A|)^{2}}{\sigma^{2}}-\ln\left(\frac{(\sqrt{C}-|A|)^{2}}{\sigma^{2}}\right)-1\right]\right\} ,
\end{eqnarray}
where the last inequality is obtained from the Chernoff bound. By
selecting large enough $C$, it becomes apparent that $\Pr\{{\cal E}_{2}\}$
must decay with an arbitrarily fast exponential rate. In particular,
it can be made faster (and hence negligible) compared to the contribution
of ${\cal E}_{1}$. It is therefore enough to confine attention to
${\cal E}_{1}$. Now, within ${\cal E}_{1}$, there are finitely many
type classes $\{{\cal T}_{n}(i\epsilon,j\epsilon)\}$, as $i$ cannot
exceed $C/(2\epsilon)$ and $|j|$ cannot exceed $\sqrt{C}/(2\epsilon)$
(because in the definition of ${\cal T}_{n}(s,\mu,\epsilon)$, $|\mu|$
cannot exceed $\sqrt{s}$, or else ${\cal T}_{n}(s,\mu,\epsilon)$
would be empty for small $\epsilon$). It follows then that the total
number of $\epsilon$-inflated type classes is less than $C^{3/2}/(2\epsilon^{2})$.
Therefore, $\Pr\{{\cal E}_{1}\}\doteq\Pr\{{\cal E}\}$ is determined
by the probability of the dominant type class within ${\cal E}_{1}$.
In the limit of small $\epsilon$, each such type contributes $\Vol\{{\cal T}_{n}(s,\mu,\epsilon)\}\doteq[2\pi e(s-\mu^{2})]^{n/2}$
times the PDF within that type, $g(\boldsymbol{x})\doteq(2\pi\sigma^{2})^{-n/2}e^{-ns/(2\sigma^{2})}$,
and it follows that the asymptotic exponent of $\Pr\{{\cal E}\}$
is given by 
\begin{equation}
\inf_{\{(s,\mu)\colon s-2A\mu+A^{2}\ge B\}}\frac{1}{2}\left[\frac{s}{\sigma^{2}}-\ln\left(\frac{s-\mu^{2}}{\sigma^{2}}\right)-1\right].
\end{equation}
Note that the objective function of this minimization can be interpreted
as the KL divergence between two Gaussian PDFs, ${\cal N}(\mu,s-\mu^{2})$
and ${\cal N}(0,\sigma^{2})$, in analogy to the form of exponential
rates of probabilities of rare events that are computed using the
traditional method of types, where the KL divergence between two finite-alphabet
distributions is minimized subject to a constraint (or constraints)
corresponding to the event in question (see also the calculation near
the end of Section \ref{subsec: simpleGausstypes}). This is also
agrees with basic foundations in large-deviations theory \cite{DZ93}.}
\end{example}

\subsection{Conditional Gaussian Types \label{subsec: ConditionalGausstypes}}

In analogy to the finite-alphabet case, the notion of conditional
types exists also in the Gaussian case. Given a sequence $\boldsymbol{y}=(y_{1},y_{2},\ldots,y_{n})\in\reals^{n}$,
a conditional Gaussian type class is defined as the set of $\{\boldsymbol{x}\}$
with given values of $\frac{1}{n}\sum_{i=1}^{n}x_{i}^{2}$ and $\frac{1}{n}\sum_{i=1}^{n}x_{i}y_{i}$.
In the $\epsilon$-inflated version, this amounts to 
\begin{equation}
{\cal T}_{n}(s,c,\epsilon|\boldsymbol{y})=\left\{ \boldsymbol{x}\colon\bigg|\frac{1}{n}\sum_{i=1}^{n}x_{i}^{2}-s\bigg|\le\epsilon,~\bigg|\frac{1}{n}\sum_{i=1}^{n}x_{i}y_{i}-c\bigg|\le\epsilon\right\} ,
\end{equation}
where $s\ge c^{2}/P_{y}$ and $P_{y}\triangleq\frac{1}{n}\sum_{i=1}^{n}y_{i}^{2}$,
due to the Schwarz--Cauchy inequality. In fact, this is an extension
of the refined Gaussian types considered in Section \ref{subsec: RefinedGausstypes},
where $y_{i}=1$ for all $i$. To estimate the volume of this conditional
type class, consider the Gaussian channel, 
\begin{equation}
g(\boldsymbol{x}|\boldsymbol{y})=\frac{\exp\left\{ -\frac{1}{2\sigma^{2}}\sum_{i=1}^{n}(x_{i}-\alpha y_{i})^{2}\right\} }{(2\pi\sigma^{2})^{n/2}},
\end{equation}
and let us select the parameters of this channel to be 
\begin{equation}
\alpha=\frac{c}{P_{y}},
\end{equation}
and 
\begin{equation}
\sigma^{2}=s-\frac{c^{2}}{P_{y}},
\end{equation}
for reasons that will become apparent shortly. It is easy to check
that the channel $g(\boldsymbol{x}|\boldsymbol{y})$ has the two desired
properties: It assigns a high probability and an approximately uniform
distribution within ${\cal T}_{n}(s,c,\epsilon|\boldsymbol{y})$,
which is of the exponential order of 
\begin{equation}
g_{0}=\left[2\pi e(s-c^{2}/P_{y})\right]^{-n/2}=e^{-nh(X|Y)},
\end{equation}
where $h(X|Y)$ is the conditional entropy of a Gaussian zero-mean,
RV $X$, with variance $s$, given a jointly Gaussian, zero-mean,
RV $Y$ with variance $P_{y}$ and $\E\{XY\}=c$. The expression $s-c^{2}/P_{y}$
is then the conditional variance of $X$ given $Y$, which is also
the minimum mean square error (MMSE) in estimating $X$ based on $Y$.
\begin{example}
\emph{Consider a simplified version of the problem of universal decoding
of \cite{me93} for the additive white Gaussian noise (AWGN) channel,
\begin{equation}
Y_{i}=\alpha X_{i}+Z_{i},~~~~~i=1,2,\ldots,n,
\end{equation}
where $\{Z_{i}\}$ are IID, zero-mean Gaussian RVs with variance $\sigma^{2}$,
$\alpha$ is an unknown fixed parameter, $\{X_{i}\}$ are the channel
inputs, and $\{Y_{i}\}$ are the channel outputs. Consider a random
codebook for channel coding, where $M=e^{nR}$ codewords of length
$n$ are selected independently at random where each codeword is drawn
under a PDF, $q(\boldsymbol{x})$, which is uniform across the surface
of a hyper-sphere of radius $\sqrt{nP}$. In \cite{me93} it is shown
that $\ln\Vol\{{\cal T}_{n}(P,\frac{1}{n}\sum_{i=1}^{n}x_{i}y_{i},\epsilon)\}$
can serve as a universal decoding metric (independent of the unknown
$\alpha$), which achieves the same random coding exponent as that
of the maximum-likelihood (ML) decoder, that is cognizant of $\alpha$.
This is equivalent to a decoder that maximizes $|\sum_{i=1}^{n}x_{i}y_{i}|$
among all codewords. In \cite{me93}, the problem is more general
in the sense that an interference signal may also be present, and
so, more interesting decoders are derived (see also \cite{HM15,HSMM19}
for further developments).}
\end{example}
The notion of a conditional Gaussian type can be easily extended to
account for conditioning on more than one vector $\boldsymbol{y}$.
Let $\boldsymbol{y}^{1},\boldsymbol{y}^{2},\ldots\boldsymbol{y}^{k}$
be $k$ given vectors in $\reals^{n}$, where $k$ is fixed, independently
of $n$. Consider the conditional type defined by 
\begin{equation}
{\cal T}_{n}(s,c_{1},\ldots,c_{k},\epsilon|\boldsymbol{y}^{1},\ldots,\boldsymbol{y}^{k})\triangleq\left\{ \boldsymbol{x}\colon\bigg|\frac{1}{n}\sum_{i=1}^{n}x_{i}^{2}-s\bigg|\le\epsilon,~\bigg|\frac{1}{n}\sum_{i=1}^{n}x_{i}y_{i}^{j}-c_{j}\bigg|\le\epsilon,~\forall~j=1,\ldots,k\right\} .
\end{equation}
Here, we can use a conditional PDF of the form 
\begin{equation}
g(\boldsymbol{x}|\boldsymbol{y}^{1},\ldots,\boldsymbol{y}^{k})=\frac{\exp\left\{ -\frac{1}{2\sigma^{2}}\sum_{i=1}^{n}\left(x_{i}-\sum_{j=1}^{k}\alpha_{i}y_{i}^{j}\right)^{2}\right\} }{(2\pi\sigma^{2})^{n/2}},
\end{equation}
and tune the parameters $(\sigma^{2},\alpha_{1},\ldots,\alpha_{k})$
such that ${\cal T}_{n}(s,c_{1},\ldots,c_{k},\epsilon|\boldsymbol{y}^{1},\ldots,\boldsymbol{y}^{k})$
would have high probability for large $n$. The resulting volume would
then be of the exponential order of $\exp\{nh(X|Y_{1},\ldots,Y_{k})\}$
where 
\begin{equation}
h(X|Y_{1},\ldots,Y_{k})=\frac{1}{2}\ln\left(2\pi e\mbox{MMSE}\{X|Y_{1},\ldots,Y_{k}\}\right),
\end{equation}
with $\mbox{MMSE}\{X|Y_{1},\ldots,Y_{k}\}$ being the MMSE of estimating
$X$ based on $Y_{1},\ldots,Y_{k}$ where $(X,Y_{1},\ldots,Y_{k})$
is a zero-mean Gaussian vector with $\E\{X^{2}\}=s$, $\E\{XY_{j}\}=c_{j}$
and a given covariance matrix of $(Y_{1},\ldots,Y_{k})$ with $\E\{Y_{m}Y_{l}\}=\frac{1}{n}\sum_{i=1}^{n}y_{i}^{m}y_{i}^{l}$.
It is not necessary to find the coefficients of the optimal (linear)
estimator of $X$ based on $(Y_{1},\ldots,Y_{k})$ in order to calculate
$\mbox{MMSE}\{X|Y_{1},\ldots,Y_{k}\}$. It is possible to calculate
$\mbox{MMSE}\{X|Y_{1},\ldots,Y_{k}\}$ directly from the covariance
matrix of $(X,Y_{1},\ldots,Y_{k})$, based on the following information-theoretic
consideration. Let $\Lambda(Y_{1},\ldots,Y_{k})$ and $\Lambda(X,Y_{1},\ldots,Y_{k})$
denote the covariance matrices of $(Y_{1},\ldots,Y_{k})$ and $(X,Y_{1},\ldots,Y_{k})$,
respectively. These matrices must both be positive definite, otherwise,
the problem is singular. Now, on the one hand, 
\begin{align}
h(X|Y_{1},\ldots,Y_{k}) & =h(X,Y_{1},\ldots,Y_{k})-h(Y_{1},\ldots,Y_{k})\nonumber \\
 & =\frac{1}{2}\ln\left[(2\pi e)^{k+1}|\Lambda(X,Y_{1},\ldots,Y_{k})|\right]-\frac{1}{2}\ln\left[(2\pi e)^{k}|\Lambda(Y_{1},\ldots,Y_{k})|\right]\nonumber \\
 & =\frac{1}{2}\ln\left[2\pi e\cdot\frac{|\Lambda(X,Y_{1},\ldots,Y_{k})|}{|\Lambda(Y_{1},\ldots,Y_{k})|}\right],
\end{align}
and on the other hand, denoting by $(\alpha_{1}^{*},\ldots,\alpha_{k}^{*})$
the coefficients of the optimal (linear) estimator, we have 
\begin{align}
h(X|Y_{1},\ldots,Y_{k}) & =h\left(X-\sum_{i=1}^{k}\alpha_{i}^{*}Y_{i}\bigg|Y_{1},\ldots,Y_{k}\right)\nonumber \\
 & =h\left(X-\sum_{i=1}^{k}\alpha_{i}^{*}Y_{i}\right)\nonumber \\
 & =\frac{1}{2}\ln\left(2\pi e\cdot\mbox{MMSE}\left\{ X|Y_{1},\ldots,Y_{k}\right\} \right),
\end{align}
where the second equality stems from the orthogonality principle,
which in the Gaussian case implies independence between $X-\sum_{i=1}^{k}\alpha_{i}^{*}Y_{i}$
and $(Y_{1},\ldots,Y_{k})$. By equating the two expressions of $h(X|Y_{1},\ldots,Y_{k})$,
we have 
\begin{equation}
\mbox{MMSE}\left\{ X|Y_{1},\ldots,Y_{k}\right\} =\frac{\left|\Lambda(X,Y_{1},\ldots,Y_{k})\right|}{\left|\Lambda(Y_{1},\ldots,Y_{k})\right|}.
\end{equation}
Thus, the volume can be calculated directly, without recourse of finding
first the optimal coefficients.

\subsection{Gauss--Markov Types \label{subsec: GaussMarkovTypes}}

So far we have dealt with Gaussian types defined by empirical second
order statistics that correspond to memoryless Gaussian sources, namely,
the empirical mean and the empirical second moment. As we mentioned
before, in the finite-alphabet case, the method of types has been
extended to Markov-types, namely, types defined by counts of transitions
between consecutive letters along a sequence, that it, the number
of time indices $\{i\}$ along an $n$-sequence $\boldsymbol{x}$
such that $x_{i-1}=a$ and $x_{i}=b$, where $a,b\in{\cal X}$ \cite{Csiszar98},
\cite{DLS81}, \cite{Natarajan85}. But what would be the corresponding
Markov extension of Gaussian types?

The simplest definition of a first-order Gauss--Markov type class
is defined as the set of all $\boldsymbol{x}\in\reals^{n}$ with prescribed
values of empirical variance, $\frac{1}{n}\sum_{i=1}^{n}x_{i}^{2}$
and the empirical first autocorrelation, $\frac{1}{n}\sum_{i=1}^{n}x_{i}x_{i-1}$
(for a given $x_{0}$). The $\epsilon$-inflated version would then
be naturally defined as 
\begin{equation}
{\cal T}_{n}(s_{0},s_{1},\epsilon)\triangleq\left\{ \boldsymbol{x}\colon\bigg|\frac{1}{n}\sum_{i=1}^{n}x_{i}^{2}-s_{0}\bigg|\le\epsilon,~\bigg|\frac{1}{n}\sum_{i=1}^{n}x_{i}x_{i-1}-s_{1}\bigg|\le\epsilon\right\} ,
\end{equation}
where $|s_{1}|\le s_{0}$. What is the volume of ${\cal T}_{n}(s_{0},s_{1},\epsilon)$?

The basic idea is the same as before: We seek a Gaussian PDF, which
assigns to ${\cal T}_{n}(s_{0},s_{1},\epsilon)$ a high probability,
and at the same time, it is approximately uniform (in the exponential
sense) across ${\cal T}_{n}(s_{0},s_{1},\epsilon)$. Given $s_{0}$
and $s_{1}$, let 
\begin{equation}
\sigma^{2}=s_{0}-\frac{s_{1}^{2}}{s_{0}}
\end{equation}
and 
\begin{equation}
\rho=\frac{s_{1}}{s_{0}},
\end{equation}
and consider the first-order Gauss--Markov process (also known as
first-order autoregressive process), 
\begin{equation}
X_{i}=\rho X_{i-1}+Z_{i},~~~~i=1,2,\ldots,n,~~X_{0}=x_{0},
\end{equation}
where $\{Z_{i}\}$ are IID, zero-mean, Gaussian RVs with variance
$\sigma^{2}$. The joint PDF of a given sample $\boldsymbol{x}$ from
this process is given by 
\begin{equation}
g(\boldsymbol{x})=\frac{\exp\left\{ -\frac{1}{2\sigma^{2}}\sum_{i=1}^{n}(x_{i}-\rho x_{i-1})^{2}\right\} }{(2\pi\sigma^{2})^{n/2}}.
\end{equation}
Neglecting edge effects, it is apparent that $g(\boldsymbol{x})$
depends on $\boldsymbol{x}$ only via $\sum_{i=1}^{n}x_{i}^{2}$ and
$\sum_{i=1}^{n}x_{i}x_{i-1}$, and so, within ${\cal T}_{n}(s_{0},s_{1},\epsilon)$,
the PDF is essentially (neglecting $\epsilon$), $g_{0}=[2\pi e(s_{0}-s_{1}^{2}/s_{0})]^{-n/2}$.
Also, by the ergodicity of the process, ${\cal T}_{n}(s_{0},s_{1},\epsilon)$
has high probability for large $n$ and fixed $\epsilon>0$, and so,
both conditions are satisfied. The volume is, therefore, of the exponential
of order of 
\begin{align}
\frac{1}{g_{0}} & =\left[2\pi e\left(s_{0}-\frac{s_{1}^{2}}{s_{0}}\right)\right]^{n/2}\nonumber \\
 & =\exp\left\{ \frac{n}{2}\ln\left[2\pi e\left(s_{1}-\frac{s_{1}^{2}}{s_{0}}\right)\right]\right\} \nonumber \\
 & =\exp\left\{ \frac{n}{2}\ln(2\pi e\sigma^{2})\right\} \nonumber \\
 & =e^{nh(X_{2}|X_{1})},
\end{align}
where $h(X_{2}|X_{1})$ is the conditional differential entropy of
$X_{2}$ given $X_{1}$, in analogy to the parallel result for finite-alphabet
Markov types, where the size of a type class is exponentially $e^{nH(X_{2}|X_{1})}$,
where $H(X_{2}|X_{1})$ is the conditional entropy associated with
the corresponding Markov process.

The intuitive explanation for this expression of the volume is as
follows: Consider the linear transformation that maps a realization
$\boldsymbol{z}=(z_{1},\ldots,z_{n})$ of the random vector $\boldsymbol{Z}=(Z_{1},\ldots,Z_{n})$
into $\boldsymbol{x}=(x_{1},\ldots,x_{n})$, which is a realization
of $\boldsymbol{X}=(X_{1},\ldots,X_{n})$. This transformation, which
is given by $x_{t}=\sum_{i=0}^{t}\rho^{i}z_{t-i}$, can be represented
by an $n\times n$ triangular transformation matrix, \emph{i.e.},
\begin{equation}
\left(\begin{array}{c}
x_{1}\\
x_{2}\\
\ldots\\
x_{n}
\end{array}\right)=\left(\begin{array}{ccccc}
1 & 0 & 0 & \ldots & 0\\
\rho & 1 & 0 & \ldots & 0\\
\ldots & \ldots & \ldots & \ldots & \ldots\\
\rho^{n-1} & \rho^{n-2} & \ldots & \rho & 1
\end{array}\right)\cdot\left(\begin{array}{c}
z_{1}\\
z_{2}\\
\ldots\\
z_{n}
\end{array}\right).
\end{equation}
Now, consider the $\epsilon$-inflated surface of the hyper-sphere
of radius $\sqrt{n\sigma^{2}}$ of $\boldsymbol{z}$-sequences, which
form the Gaussian type of the driving noise process, $\{Z_{i}\}$.
The volume of this type class is exponentially $[2\pi e\sigma^{2}]^{n/2}=[2\pi e(s_{0}-s_{1}^{2}/s_{0})]^{n/2}$.
But these typical $\boldsymbol{z}$-sequences are mapped into corresponding
$\boldsymbol{x}$-sequences, by the above triangular transformation
matrix whose diagonal terms are all equal to $1$, and hence its Jacobian
is also equal to $1$. In other words, the transformation from $\boldsymbol{z}$
to $\boldsymbol{x}$ preserves volumes, and so, the volume of the
$\epsilon$-inflated surface of a hyper-sphere of typical $\boldsymbol{z}$-sequences
is transformed by the above matrix into a hyper-ellipsoid of typical
$\boldsymbol{x}$-sequences of exactly the same volume.
\begin{example}
\emph{Consider the calculation of exponential decay rate of 
\begin{equation}
\Pr\left\{ \sum_{t=1}^{n}X_{t}X_{t-1}\ge\rho\sum_{t=1}^{n}X_{t}^{2}\right\} ,
\end{equation}
for some $\rho>0$, where $\{X_{t}\}$ are zero-mean, Gaussian RVs
with variance $\sigma^{2}$. Since the volume of the type is of the
exponential order of $[2\pi e(s_{0}-s_{1}^{2}/s_{0})]^{n/2}=[2\pi es_{0}(1-s_{1}^{2}/s_{0}^{2})]^{n/2}$,
and the PDF within a type class is $(2\pi\sigma^{2})^{-n/2}\exp\{-ns_{0}/(2\sigma^{2})\}$,
then the exponent is given by 
\begin{align}
 & \inf_{\{(s_{0},s_{1})\colon s_{0}\ge0,~s_{1}/s_{0}\ge\rho\}}\left[\frac{s_{0}}{2\sigma^{2}}+\frac{1}{2}\ln(2\pi\sigma^{2})-\frac{1}{2}\ln(2\pi es_{0})-\frac{1}{2}\ln\left(1-\frac{s_{1}^{2}}{s_{0}^{2}}\right)\right]\nonumber \\
 & =\inf_{s_{0}\ge0}\left[\frac{s_{0}}{2\sigma^{2}}+\frac{1}{2}\ln(2\pi\sigma^{2})-\frac{1}{2}\ln(2\pi es_{0})\right]-\frac{1}{2}\ln(1-\rho^{2})\nonumber \\
 & =-\frac{1}{2}\ln(1-\rho^{2}).
\end{align}
}
\end{example}
More generally, consider a $k$-th order Gauss--Markov type, defined
by 
\begin{equation}
{\cal T}_{n}(s_{0},s_{1},\ldots,s_{k},\epsilon)\triangleq\left\{ \boldsymbol{x}\colon\bigg|\frac{1}{n}\sum_{i=1}^{n}x_{i}x_{i-j}-s_{j}\bigg|\le\epsilon,~\forall~j=0,1,\ldots,k\right\} ,
\end{equation}
for given $(s_{0},s_{1},\ldots,s_{k})$ and some $(x_{0},x_{-1},\ldots,x_{-(k-1)})$.
It is assumed that the $(k+1)\times(k+1)$ matrix $S$ whose $(i,j)$-th
entry ($i,j\in\{0,1,\ldots,k\}$) is $s_{|i-j|}$ is a positive definite
matrix. Here, we find a matching $k$-th order autoregressive (AR)
process, 
\begin{equation}
X_{t}=\sum_{i=1}^{k}a_{i}X_{t-i}+Z_{t},~~~~~~t=1,2,\ldots,
\end{equation}
where $\{Z_{t}\}$ is again Gaussian white noise with variance $\sigma^{2}$,
such that $\E\{X_{t}X_{t-i}\}=s_{i}$ for all $i=0,1,\ldots,k$. Given
$(s_{0},s_{1},\ldots,s_{k})$, the corresponding parameter vector,
$(\sigma^{2},a_{1},\ldots,a_{k})$ of the AR process is obtained by
solving the Yule--Walker equations \cite[Eqs.\ (12-41a), (12-41b)]{Papoulis91},
\begin{equation}
\sum_{i=1}^{k}a_{i}s_{|i-j|}=s_{j},~~~~~j=1,2,\ldots,k,
\end{equation}
and 
\begin{equation}
\sigma^{2}=s_{0}-\sum_{i=1}^{k}a_{i}s_{i}.
\end{equation}
The corresponding PDF $g$, which is given by 
\begin{equation}
g(\boldsymbol{x})=\frac{1}{(2\pi\sigma^{2})^{n/2}}\exp\left\{ -\frac{1}{2\sigma^{2}}\sum_{t=1}^{n}\left(x_{t}-\sum_{i=1}^{k}a_{i}x_{t-i}\right)^{2}\right\} ,
\end{equation}
has the two desired properties of exponential uniformity within ${\cal T}_{n}(s_{0},s_{1},\ldots,s_{k},\epsilon)$
and assigning high probability to ${\cal T}_{n}(s_{0},s_{1},\ldots,s_{k},\epsilon)$.
Here too, $g_{0}=(2\pi e\sigma^{2})^{-n/2}$ which implies that the
volume of the type class is essentially 
\begin{equation}
\frac{1}{g_{0}}=(2\pi e\sigma^{2})^{n/2}.
\end{equation}
The intuition is the same as before: The mapping from $\boldsymbol{x}$
to $\boldsymbol{z}$ is by a triangular matrix whose diagonal entries
are all equal to 1, and so is its Jacobian. Therefore, it preserves
volumes, and so is the inverse transformation, which maps the hyper-sphere
surface of volume $(2\pi e\sigma^{2})^{n/2}$ in the $\boldsymbol{z}$-domain
into the typical hyper-ellipsoid of the same volume in the $\boldsymbol{x}$-domain.
Since $\sigma^{2}$ is the variance of the innovation process, the
differential entropy rate of $\{X_{t}\}$ is given by 
\begin{align}
h & =\lim_{n\to\infty}\frac{h(\boldsymbol{X})}{n}\nonumber \\
 & =\frac{1}{4\pi}\int_{-\pi}^{\pi}\ln[2\pi eS(e^{i\omega})]\dd\omega~~~~~~~~~~~~~i\triangleq\sqrt{-1}\nonumber \\
 & =\frac{1}{2}\ln(2\pi e)+\frac{1}{4\pi}\int_{-\pi}^{\pi}\ln S(e^{i\omega})\dd\omega\nonumber \\
 & =\frac{1}{2}\ln(2\pi e)+\frac{1}{4\pi}\int_{-\pi}^{\pi}\ln\left[\frac{\sigma^{2}}{\bigg|1-\sum_{j=1}^{k}a_{j}e^{-j\omega i}\bigg|^{2}}\right]\dd\omega\nonumber \\
 & =\frac{1}{2}\ln(2\pi e)+\frac{1}{2}\ln\sigma^{2}\nonumber \\
 & =\frac{1}{2}\ln(2\pi e\sigma^{2}),
\end{align}
where $S(e^{i\omega})$ is the spectrum of $\{X_{t}\}$ and where
we have used the Kolmogorov--Szeg\"{o} relation \cite[p.\ 491]{Papoulis91}
between the spectrum and the innovation variance.\footnote{Note that $\int_{-\pi}^{\pi}\ln[1-\sum_{j=1}^{k}a_{j}e^{-j\omega i}]\dd\omega=0$
since all zeroes of the function $1-\sum_{j=1}^{k}a_{j}z^{-j}$ must
be within the unit circle.} This implies that the volume continues to be of the exponential order
of $e^{nh}$. Similarly as before, $\sigma^{2}$ can be found directly
from the covariance matrix of $(s_{0},s_{1},\ldots,s_{k})$, as the
ratio between the determinants of the covariance matrix of order $(k+1)\times(k+1)$,
and the covariance matrix of order $k\times k$.

Is it possible to calculate the volume of a type class that is defined
by prescribed values of both the empirical autocorrelation and the
correlation with a given $\boldsymbol{y}$? This turns out to be considerably
more tricky (see the discussion in \cite{me93}) and it requires more
advanced tools that will be provided in the next chapter.

\subsection{Types Classes Pertaining to Exponential Families \label{subsec:Types-Classes-Pertaining}}

So far, we have considered various kinds of Gaussian types, which
are defined WRT given values of first and second order empirical statistics,
like the empirical mean, the empirical second moment, the empirical
correlation and autocorrelation, and so on. We now move on to extend
the scope to deal with types associated with empirical moments or
arbitrary functions. As described in Section \ref{subsec: CMoTintro},
consider the type class of all sequences $\{\boldsymbol{x}\}$ that
share the same combination of values of statistics $\frac{1}{n}\sum_{i=1}^{m}\phi_{j}(x_{i})$,
$j=1,2,\ldots,k$. More formally, consider the $\epsilon$-inflated
type class 
\begin{equation}
{\cal T}_{n}(q,\epsilon)\triangleq\left\{ \boldsymbol{x}\colon\bigg|\frac{1}{n}\sum_{i=1}^{m}\phi_{j}(x_{i})-q_{j}\bigg|\le\epsilon,~\forall~1\le j\le k\right\} .
\end{equation}
where $q=(q_{1},\ldots,q_{k})$. What is the volume of ${\cal T}_{n}(q,\epsilon)$?
Using the same general idea as before, we seek a PDF of $\boldsymbol{x}$
which would assign to all members of ${\cal T}_{n}(q,\epsilon)$ approximately
the same PDF (in the exponential scale), and at the same time, the
probability of ${\cal T}_{n}(q,\epsilon)$ would be large for large
$n$. As discussed in Section \ref{subsec: CMoTintro}, consider the
PDF 
\begin{equation}
P_{\theta}(\boldsymbol{x})=\frac{\exp\left\{ \sum_{j=1}^{k}\theta_{j}\sum_{i=1}^{n}\phi_{j}(x_{i})\right\} }{[Z(\theta)]^{n}},
\end{equation}
where $\theta=(\theta_{1},\ldots,\theta_{k})$ and 
\begin{equation}
Z(\theta)=\int_{{\cal X}}\exp\left\{ \sum_{j=1}^{k}\theta_{j}\phi_{j}(x)\right\} \dd x,
\end{equation}
assuming that ${\cal X}$ is a continuous alphabet, and where it is
understood that in the discrete case, the integration over ${\cal X}$
is replaced by summation. Clearly, $P_{\theta}(\boldsymbol{x})$ assigns
exponentially the same PDF to all members of ${\cal T}_{n}(q,\epsilon)$,
but ${\cal T}_{n}(q,\epsilon)$ has high probability only if $\theta$
is tuned accordingly for the given vector, $q$. If we can select
$\theta$ such that 
\begin{equation}
\E\left\{ \phi_{j}(X)\right\} \equiv\frac{\partial\ln Z(\theta)}{\partial\theta_{j}}=q_{j},\label{theta(q)}
\end{equation}
simultaneously for all $1\le j\le k$, then by the WLLN, ${\cal T}_{n}(q,\epsilon)$
would have high probability. Let us assume then, that $q$ is such
that there exists a parameter vector $\theta$ that solves the set
of $k$ simultaneous equations \eqref{theta(q)}, which can be presented
in the vector form as 
\begin{equation}
\nabla\ln Z(\theta)=q.
\end{equation}
Let $\theta=G(q)$ denote solution to this vector equation. In other
words, $G(q)$ is the inverse mapping of $F(\theta)=\nabla\ln Z(\theta)$,
provided that it exists. The PDF of every $\boldsymbol{x}\in{\cal T}_{n}(q,\epsilon)$
is exponentially 
\begin{equation}
\frac{\exp\left\{ n\sum_{j=1}^{k}\theta_{j}q_{j}\right\} }{[Z(\theta)]^{n}}=\frac{\exp\left\{ nq^{T}G(q)\right\} }{[Z(G(q))]^{n}},
\end{equation}
and so, the volume of ${\cal T}_{n}(q,\epsilon)$ is of the exponential
order of the reciprocal 
\begin{equation}
\exp\left\{ n\left[\ln Z(G(q))-q^{T}G(q)\right]\right\} .
\end{equation}
Note that the (differential) entropy associated with $P_{\theta}$
is given by 
\begin{equation}
h[q]=\E\left\{ \ln\frac{1}{P_{\theta}(X)}\right\} =\ln Z(\theta)-q^{T}\theta=\ln Z(G(q))-q^{T}G(q),
\end{equation}
and so, once again, the volume is of the exponential order of $e^{nh[q]}$.

It is interesting to relate the asymptotic evaluation of the log-volume
of a type class to the \emph{principle of maximum entropy} (see, e.g.,
\cite[Chapter 12]{CT06} and references therein). We argue that $h[q]$
is the largest possible differential entropy of any RV, $X$, that
satisfies the moment constraints, $\E\{\phi_{j}(X)\}=q_{j}$, $j=1,2,\ldots,k$.
To see why this is true, consider the following chain of equalities:
\begin{align}
\sup_{\{X\colon\E\{\phi_{j}(X)\}=q_{j},~1\le j\le k\}}h(X) & =\sup_{X}\inf_{\theta}\left[h(X)+\sum_{j=1}^{k}\theta_{j}\left(\E\{\phi_{j}(X)\}-q_{j}\right)\right]\nonumber \\
 & =\sup_{f}\inf_{\theta}\int_{-\infty}^{\infty}\dd xf(x)\left[\ln\frac{1}{f(x)}+\sum_{j=1}^{k}\theta_{j}\left(\phi_{j}(x)-q_{j}\right)\right]\nonumber \\
 & =\sup_{f}\inf_{\theta}\int_{-\infty}^{\infty}\dd xf(x)\left[\ln\frac{\exp\left\{ \sum_{j=1}^{k}\theta_{j}\phi_{j}(x)\right\} }{f(x)}-\sum_{j=1}^{k}\theta_{j}q_{j}\right]\nonumber \\
 & =\sup_{f}\inf_{\theta}\int_{-\infty}^{\infty}\dd xf(x)\left[\ln\frac{P_{\theta}(x)\cdot Z(\theta)}{f(x)}-q^{T}\theta\right]\nonumber \\
 & \trre[=,a]\inf_{\theta}\sup_{f}\left\{ -D(f\|P_{\theta})+\ln Z(\theta)-q^{T}\theta\right\} \nonumber \\
 & =\inf_{\theta}\left\{ \ln Z(\theta)-q^{T}\theta\right\} \nonumber \\
 & \trre[=,b]\ln Z(G(q))-q^{T}G(q)\nonumber \\
 & =h[q],
\end{align}
where $(a)$ follows from the minimax theorem and fact that the objective
is convex in $\theta$ and concave in $f$ and $(b)$ follows from
the fact that the minimizing $\theta$ is $\theta^{*}=G(q)$, which
is obtained by equating to zero the gradient of the convex function
$\ln Z(\theta)-q^{T}\theta$. As can be seen, the maximizing $f$
is exactly $P_{\theta}$ with $\theta=G(q)$.
\begin{example}
\emph{The volume of the ``Laplacian type class,'' where $k=1$ and
$\phi_{1}(x)=|x|$ is exponentially 
\begin{equation}
(2eq)^{n}=\exp(nh[q]),
\end{equation}
where 
\begin{equation}
h[q]=\ln(2eq)
\end{equation}
is the differential entropy of a Laplacian RV with $\E\{|X|\}=q$.
More generally, the ``generalized Gaussian type class'' is defined
for $k=1$ and $\phi_{1}(x)=|x|^{m}$ (for arbitrary $m>0$), where
the volume exponent is given by the differential entropy of the generalized
Gaussian RV with $\E\{|X|^{m}\}=q$, which is given by 
\begin{equation}
h[q]=\frac{1}{m}\ln\left(\frac{meq}{2c_{m}}\right),
\end{equation}
where 
\begin{equation}
c_{m}=\left[\frac{m}{2^{1+1/m}\Gamma(1/m)}\right]^{m}.
\end{equation}
}
\end{example}
\emph{The method of types for exponential families is flexible enough
to evaluate exponential rates of moments and probabilities of events
defined WRT statistics that are different from the sufficient statistics
of underlying PDF. Consider the following example.}
\begin{example}
\emph{Suppose that $X_{1},X_{2},\ldots,X_{n}$ are IID, zero-mean,
Gaussian RVs with variance $\sigma^{2}$ and we wish to assess the
probability that $\sum_{i=1}^{n}|X_{i}|\ge nA$, where $A\ge\sqrt{\frac{2}{\pi}}\sigma$.
In such a case, we may define type classes as above with $k=2$, $\phi_{1}(x)=|x|$
and $\phi_{2}(x)=x^{2}$, where $\phi_{1}$ is needed to support the
statistics of the event in question, and $\phi_{2}$ is for the underlying
Gaussian PDF. Then, each type class, ${\cal T}_{n}(q_{1},q_{2},\epsilon)$,
contributes a probability of the exponential order of 
\begin{equation}
e^{nh[q_{1},q_{2}]}\cdot(2\pi\sigma^{2})^{-n/2}e^{-nq_{2}/(2\sigma^{2})},
\end{equation}
and so, the dominant type class contributes an exponential order of
\begin{equation}
\inf_{(q_{1},q_{2})\colon q_{1}\ge A,~q_{2}\ge q_{1}^{2}}\left\{ \frac{q_{2}}{2\sigma^{2}}-h[q_{1},q_{2}]\right\} +\frac{1}{2}\ln(2\pi\sigma^{2}),
\end{equation}
where the constraint $q_{2}\ge q_{1}^{2}$ follows from the inequality
$\frac{1}{n}\sum_{i=1}^{n}x_{i}^{2}\ge(\frac{1}{n}\sum_{i=1}^{n}|x_{i}|)^{2}$.}
\end{example}
Finally, we point out that extension to conditional types and Markov
types can be carried out conceptually straightforwardly following
the same ideas described above in the context of Gaussian types. In
both cases, the main engine is corresponding the exponential family,
which is defined in \eqref{eq: expomarkov} for Markov types and in
\eqref{eq: expochannel} for conditional types.

\subsection{Applications \label{subsec: CMOT applications}}

The Gaussian method of types has found application in various contexts
and levels of generality across prior research. In this section, we
provide a brief overview of these contexts along with some of the
outcomes achieved.

In \cite{me93}, the challenge of universal decoding for memoryless
Gaussian channels with unknown deterministic interference was tackled,
and the method of Gaussian types played a central role in the analysis.
As highlighted in \cite[Eq.\ (5)]{me93}, the universal decoding metric
for the Gaussian channel hinges on the volume of the conditional Gaussian
type class of a channel input vector $\boldsymbol{x}$, given a channel
output vector $\boldsymbol{y}$. The effectiveness of this decoding
metric is contingent on having an explicit formula for the exponential
rate of this volume.

The extension from the memoryless case to Gaussian channels with intersymbol
interference remained an open question after \cite{me93}, as estimating
the corresponding volume was non-trivial. The gap was eventually bridged
in \cite{HM15} and \cite{HSMM19} using more advanced methodologies
to be discussed later. A similar connection between universal decoding
metrics and volumes of conditional type classes was observed in a
broader context of universal decoding for arbitrary channels concerning
a specific class of decoding metrics \cite{me13}. Additional insights
can be found in \cite[Section 4]{MS08}.

The method of Gaussian types has also played a pivotal role in deducing
random coding exponents for typical random codes in distinctive scenarios.
For instance, in the context of the colorful Gaussian channel \cite{me19}
and the dirty-paper channel \cite{TM23}, this method was crucial.
Both studies relied on the concept of conditional type classes and
their associated volumes, and the presence of explicit expressions
played a vital role in achieving exponentially tight results. In \cite[Subsection IV.A]{AM98},
the method of Gaussian types found application in addressing the problem
of optimal guessing subject to a fidelity constraint within the realm
of memoryless Gaussian sources. This corresponds to a parallel result
for finite-alphabet memoryless sources, for which the conventional
method of types is employed. By leveraging outcomes pertaining to
the exponential order of the volume of both simple Gaussian types
and conditional Gaussian types, the optimal achievable guessing exponent
was deduced. The crux of this derivation revolves around creating
a continuous version of the type-covering lemma. This lemma establishes
the capability to encompass a Euclidean hyper-sphere with a radius
of $\sqrt{n\sigma^{2}}$ using exponentially $\exp\{\frac{n}{2}\ln\frac{\sigma^{2}}{D}\}$
Euclidean hyper-spheres, each with a radius of $\sqrt{nD}$, where
$D<\sigma^{2}$. This type-covering result was reaffirmed and expanded
to support successive refinement coding theorems in \cite{ZTM17},
also employing Gaussian types. Interestingly, it seems that the authors
of \cite{ZTM17} were unaware of the initial version of this result
in \cite{AM98}. Gaussian types were also harnessed by Kelly and Wagner
in \cite{KW12} concerning the reliability of source coding with side
information (the Wyner--Ziv problem). Moreover, Scarlett \cite{Scarlett15}
and Scarlett and Tan \cite{ST15} employed Gaussian types (termed
``power types'') for second-order asymptotic analyses in their respective
works. Similar trains of thought were explored in \cite{ICW15} within
the domain of compression for similarity queries. Additional related
references include \cite{Tan14} and \cite{Vamvatsikos07}. Furthermore,
an analogous type-covering lemma for Laplacian type classes was established
in \cite{ZAC06} (also covered in \cite{ZAC09}).

The method of types extended to general exponential families found
application in \cite{me89} within the domain of model order estimation.
Just as mentioned previously, in this context as well, the existence
of an expression for the volume of a type class played a pivotal role
in deducing the model order selection criterion. Additionally, in
\cite{MKLS94}, the method of types was employed for exponential families
within the context of a continuous-alphabet extension of widely recognized
lower bounds for mismatched capacity, utilizing random coding analysis.
This showcases the versatility of the method across diverse problem
domains.

\newpage{}

\section{The Laplace Method of Integration and the Saddle-Point Method \label{sec: laplacesaddlepoint}}

\subsection{Introduction \label{subsec:Introduction Laplace}}

The Laplace method of integration (see, e.g., \cite[Chapter 4]{deBruijn81}
and \cite[Section 4.2]{me09}) is a powerful technique for approximating
definite integrals of the form: 
\begin{equation}
\int_{a}^{b}g(x)e^{nf(x)}\dd x,
\end{equation}
where the parameter $n$ is significantly large ($n\gg1$), and the
functions $f$ and $g$ exhibit sufficient regularity WRT the real
variable $x$. Importantly, these functions are assumed to remain
independent of $n$. More generally, $x$ may designate a $d$-dimensional
vector, where $d$ is independent of the large parameter $n$, whereas
the integration occurs over $\reals^{d}$ or a subset thereof.

The significance of this method is twofold. Firstly, it offers intrinsic
utility by itself, providing accurate asymptotic approximations for
integrals. However, its greater importance lies in its role as the
foundation for the saddle-point method, an extension that applies
the principles of the Laplace method to the integration of complex
functions along contours within the complex plane. The saddle-point
method finds broad applications across diverse disciplines, including
physics, probability, statistics, and engineering. Notably, this chapter
emphasizes that the method holds promise in the realm of information
theory as well. In many instances, the saddle-point method can serve
as a viable alternative to the extended method of types discussed
in Chapter \ref{sec: CMoT}. This advantage becomes particularly apparent
when it comes to circumventing the need for $\epsilon$-inflation
of type classes, a strategy employed in Chapter \ref{sec: CMoT}.
The Laplace method and saddle-point method offer a distinct advantage
by not only yielding the accurate exponential rate, as demonstrated
in Chapter \ref{sec: CMoT}, but also by providing the correct pre-exponential
term. Remarkably, this method furnishes approximations that exhibit
asymptotic precision. Specifically, as the large parameter $n$ grows
without bound, the ratio between the approximation and the actual
value converges to unity, signifying an increasingly faithful representation
of the underlying quantity.

It is important to note that the content presented in this chapter
exhibits some overlap with the material found in \cite[Sections 4.2 and 4.3]{me09}
and in \cite[Chapters 4 and 5]{deBruijn81}. As a result, several
intricate technical aspects related to the Laplace method and the
saddle-point method are either succinctly addressed or occasionally
omitted (though appropriately cross-referenced to \cite{deBruijn81,me09}).
Instead, the focus here lies on considering these methods in the context
of their capacity to stand as valid alternatives to the generalized
method of types, as described in Chapter \ref{sec: CMoT}. This pertains
to both its discrete and continuous alphabet variations. For readers
seeking a more comprehensive treatment with meticulous attention to
detail and rigor, we recommend delving into the pertinent chapters
of \cite{deBruijn81} and \cite{me09}.

\subsection{The Laplace Method of Integration \label{subsec: laplace}}

Commencing with the Laplace method, we turn our attention to an illustrative
example tied to the domain of universal source coding (as expounded
in references such as \cite{Davisson73} and \cite[Section 13.2]{CT06}).
This example serves as a compelling application that underscores the
significance of the Laplace method within the realm of information
theory. Through this example, we emphasize how the Laplace method
finds relevance and utility in tackling fundamental challenges in
information-theoretic contexts.
\begin{example}[Universal coding]
\emph{\label{exa: univcoding} Consider a family of binary memoryless
(Bernoulli) sources defined over the alphabet $\{0,1\}$, parameterized
by $\theta\in[0,1]$, which represents the probability of emitting
a $'1'$. The probability mass function of this source is given by:
\begin{equation}
P_{\theta}(\boldsymbol{x})=(1-\theta)^{n-n_{1}}\theta^{n_{1}},
\end{equation}
where $\boldsymbol{x}\in\{0,1\}^{n}$, and $n_{1}\leq n$ is the count
of occurrences of $'1'$ in $\boldsymbol{x}$. When dealing with an
unknown $\theta$, a universal code is devised using the Shannon code,
adapted to a weighted mixture of these sources: 
\begin{equation}
P(\boldsymbol{x})=\int_{0}^{1}\dd\theta w(\theta)P_{\theta}(\boldsymbol{x})=\int_{0}^{1}\dd\theta w(\theta)e^{nf(\theta)},
\end{equation}
where $w(\cdot)$ is a positive function that integrates to unity
across the interval $[0,1]$, and 
\begin{equation}
f(\theta)=\ln(1-\theta)+q\ln\left(\frac{\theta}{1-\theta}\right);\quad q=\frac{n_{1}}{n}.
\end{equation}
This necessitates the computation of an integral involving an exponential
function of $n$ (in this case, across the interval $[0,1]$) to evaluate
the performance of this universal code. An asymptotically exact evaluation
of such an integral is crucial in the quest of characterizing, not
only the main term of the achievable compression ratio, but also the
redundancy terms (see Example \ref{exa: univcodingrevisited} below).}
\end{example}
Consider first an integral of the form: 
\begin{equation}
F_{n}\triangleq\int_{-\infty}^{+\infty}e^{nf(x)}\dd x,
\end{equation}
where the function $f(\cdot)$ is independent of $n$. It will be
assumed that the function $f$ satisfies the following assumptions: 
\begin{enumerate}
\item $f$ is real and continuous. 
\item $f$ has a unique global maximum at $x=x_{0}$: $f(x)<f(x_{0})~~\forall x\ne x_{0}$,
and $\exists b>0,~c>0$ such that $|x-x_{0}|\ge c$ implies $f(x)\le f(x_{0})-b$. 
\item The integral defining $F_{n}$ converges for all large enough $n$.
Without loss of generality, let this sufficiently large $n$ be $n=1$,
\emph{i.e.}, $\int_{-\infty}^{+\infty}e^{f(x)}\dd x<\infty$. 
\item The derivative $f'(x)$ exists at a certain open neighborhood of $x=x_{0}$,
and $f''(x_{0})<0$. Thus, $f'(x_{0})=0$. 
\end{enumerate}
These assumptions pave the way to approximate $f(x)$, at the vicinity
of $x=x_{0}$, by a second-order Taylor series expansion, 
\begin{equation}
f(x)\approx f(x_{0})+\frac{f''(x_{0})}{2}(x-x_{0})^{2}=f(x_{0})-\frac{|f''(x_{0})|}{2}(x-x_{0})^{2},
\end{equation}
which renders $F_{n}$ as being dominated by the constant $e^{nf(x_{0})}$,
multiplied by a Gaussian integral, namely, the integral of $\exp\{-\frac{N}{2}|f''(x_{0})|(x-x_{0})^{2}\}$,
whereas the combined contribution of all the range away from $x_{0}$
is negligibly small for large $n$. Accordingly, as shown in \cite[Chapter 4]{deBruijn81}
and \cite[Section 4.2]{me09}, we arrive at the Laplace method approximation,
given by 
\begin{equation}
\int_{-\infty}^{+\infty}e^{nf(x)}\dd x\sim e^{nf(x_{0})}\cdot\sqrt{\frac{2\pi}{n|f''(x_{0})|}}.\label{eq: laplace1}
\end{equation}
This approximation continues to apply if $F_{n}$ is defined as an
integral over any finite or half-infinite interval that contains the
maximizer $x=x_{0}$ as an internal point. On the other hand, if the
maximizer $x_{0}$ falls at one of the endpoints of the integration
range, and $f'(x_{0})$ does not vanish, the Gaussian integral approximation
ceases to apply, and the local behavior around the maximum would be
approximated by an exponential $\exp\{-n|f'(x_{0})|(x-x_{0})\}$ instead,
which gives a different pre-exponential factor, yet the exponential
factor $e^{nf(x_{0})}$ would continue to be present. A further extension
for the case where $x_{0}$ is an internal point at which the derivative
vanishes, is the following: 
\begin{equation}
\int_{-\infty}^{+\infty}g(x)e^{nf(x)}\dd x\sim g(x_{0})e^{nf(x_{0})}\cdot\sqrt{\frac{2\pi}{n|f''(x_{0})|}},\label{eq: laplace2}
\end{equation}
where $g$ is another function that does not depend on $n$. In a
more general context, when the integration variable $x$ represents
a $d$-dimensional vector, where $d$ is a positive integer independent
of $n$, and the integration takes place over $\reals^{d}$ or a subset
thereof, with $x_{0}$ positioned as an internal point within the
integration region, we must replace $|f''(x_{0})|$ in both \eqref{eq: laplace1}
and \eqref{eq: laplace2} with the absolute value of the determinant
of the Hessian matrix of $f$ evaluated at $x=x_{0}$. Additionally,
the factor $n$ that multiplies $|f''(x_{0})|$ should be substituted
with $n^{d}$. This adjustment arises from a corresponding approximation
involving a multi-dimensional Gaussian integral. If the global maximum
of $f$ is achieved by more than one point, and the number of maximizers
is finite or countable, then the contributions from all of these maximizers
should be aggregated or summed together.

We next look into a few examples.
\begin{example}[Universal coding revisited]
\emph{\label{exa: univcodingrevisited} Applying the Laplace integral
approximation to Example \ref{exa: univcoding}, we have 
\begin{equation}
P(\boldsymbol{x})=\int_{0}^{1}w(\theta)\exp\left\{ n\left[\ln(1-\theta)+q\ln\left(\frac{\theta}{1-\theta}\right)\right]\right\} \dd\theta\sim w(q)e^{-nH(q)}\sqrt{\frac{2\pi q(1-q)}{n}},
\end{equation}
where $H(q)\triangleq-q\ln q-(1-q)\ln(1-q)$ is the empirical entropy
of $\boldsymbol{x}$, and so, the compression ratio pertaining to
the Shannon code WRT the mixture is 
\begin{equation}
\frac{-\ln P(\boldsymbol{x})}{n}=H(q)+\frac{\ln n}{2n}-\frac{\ln\left[w(q)\sqrt{2\pi q(1-q)}\right]}{n}+o\left(\frac{1}{n}\right).
\end{equation}
The principal component of the normalized redundancy can be expressed
as $\frac{\ln n}{2n}$, a well-established result (for more details,
refer to \cite{KT81}). Similarly, when considering a mixture encompassing
all sources with an alphabet size of $r$, this entails integration
over $r-1$ letter probabilities, resulting in a dominant redundancy
term of $\frac{(r-1)\ln n}{2n}$. }
\end{example}
\begin{example}[Extreme Value Statistics]
\emph{\label{exa: extremevaluestatistics} Consider a set of non-negative,
IID RVs $\{X_{i}\}_{i=1}^{n}$, each characterized by the PDF $p(x)$.
Our goal is to evaluate the expectation of the minimum value among
these variables, $\E\{\min_{i\le i\le n}X_{i}\}$. Let us explore
the following sequence of equalities to facilitate this assessment.
Denoting the cumulative distribution function of each $X_{i}$ by
$F(x)$, we have 
\begin{align}
\E\left\{ \min_{1\le i\le n}X_{i}\right\}  & =\int_{0}^{\infty}\Pr\left\{ \min_{i\le i\le n}X_{i}\ge x\right\} \dd x\nonumber \\
 & =\int_{0}^{\infty}\Pr\left[\bigcap_{i=1}^{n}\{X_{i}\ge x\}\right]\dd x\nonumber \\
 & =\int_{0}^{\infty}\left[1-F(x)\right]^{n}\dd x\nonumber \\
 & =\int_{0}^{\infty}\exp\{n\ln[1-F(x)]\}\dd x,
\end{align}
hence we may use the Laplace method with $f(x)=\ln[1-F(x)]$. Here,
the maximum of $f(x)$ is obtained at the edge-point of the integration
domain, $x_{0}=0$ and $f'(0)=-p(0)<0$. Therefore, the approximation
is not by a Gaussian integral, but a simple exponential, 
\begin{equation}
\int_{0}^{\infty}\exp\{-n|f'(0)|x\}\dd x=\frac{1}{n|f'(0)|},
\end{equation}
which yields 
\begin{equation}
\E\left\{ \min_{i\le i\le n}X_{i}\right\} \sim\frac{1}{np(0)}.
\end{equation}
However, if $p(0)=0$ while $p'(0)>0$, the Laplace approximation
is executed through a Gaussian integral over half of the real line.
In such a scenario, the outcome is as follows: 
\begin{equation}
\E\left\{ \min_{i\le i\le n}X_{i}\right\} \sim\frac{1}{2}\sqrt{\frac{2\pi}{np'(0)}}.
\end{equation}
}
\end{example}
The last example in this section supports the Stirling approximation.
\begin{example}[The Stirling formula]
\emph{\label{exa: Stirling} Beginning from the identity $\int_{0}^{\infty}\dd xe^{-sx}=1/s$,
and differentiating both sides $n$ times WRT $s$, the left-hand
side becomes $(-1)^{n}\int_{0}^{\infty}x^{n}e^{-sx}\dd x$, and the
right-hand side (RHS) gives $(-1)^{n}n!/s^{n+1}$, which together
yield the identity 
\begin{equation}
n!=s^{n+1}\int_{0}^{\infty}x^{n}e^{-sx}\dd x,
\end{equation}
holding true for every $s>0$. On substituting $s=n$, we get 
\begin{equation}
n!=n^{n+1}\int_{0}^{\infty}x^{n}e^{-nx}\dd x=n^{n+1}\int_{0}^{\infty}e^{n(\ln x-x)}\dd x.
\end{equation}
Assessing this integral using the Laplace method, we have $f(x)=\ln x-x$,
which is maximized at $x_{0}=1$, with $f(x_{0})=f''(x_{0})=-1$.
Thus, 
\begin{equation}
n!\sim n^{n+1}e^{-n\cdot1}\sqrt{\frac{2\pi}{n\cdot1}}=\left(\frac{n}{e}\right)^{n}\sqrt{2\pi n},
\end{equation}
which is the well-known Stirling formula for approximating $n!$.}
\end{example}

\subsection{The Saddle-Point Method \label{subsec: saddlepoint}}

We now broaden our focus to encompass integrals along paths within
the complex plane, a concept that arises more frequently than one
might anticipate. As previously mentioned, the extension of the Laplace
integration technique to the realm of complex functions is referred
to as the saddle-point method or the steepest descent method, with
explanations for these names becoming apparent in the forthcoming
presentation. Specifically, our current interest lies in evaluating
an integral represented as follows: 
\begin{equation}
F_{n}=\int_{{\cal P}}g(z)e^{nf(z)}\dd z.
\end{equation}
In this context, the variable $z$ takes on complex values and ${\cal P}$
designates a certain path within the complex plane, originating from
a point $A$ and concluding at a point $B$. Our initial focus will
be on the case $g(z)\equiv1$, and we make the assumption that ${\cal P}$
exclusively lies within a region where the function $f$ is analytic.

At first glance, the reader might question the relevance of complex
integrals when dealing with quantities that are inherently real ---
such as probabilities, expectations, volumes of high-dimensional objects,
etc. The answer lies in the fact that even if these quantities are
real, there are instances where expressing a certain term in a calculation
as an inverse Fourier transform or an inverse Laplace transform, or
inverse Z-transform, becomes useful and beneficial. These inverse
transforms are represented through complex integrals. To illustrate,
consider the following straightforward example: Computing the volume
of an $n$-dimensional hyper-sphere with radius $R$. This task can
be approached by interpreting the volume as the integral of $U(R^{2}-\sum_{i=1}^{n}x_{i}^{2})$
over $\reals^{n}$, where $U(t)$ signifies the Heaviside unit step
function. Next, we express $U(t)$ as the inverse Laplace transform
of $1/s$, subsequently we interchange the integration order, and
finally, we apply the saddle-point method to evaluate the complex
integration. As we proceed, we will delve into the meticulous execution
of this concept.

The first observation of significance is that the integral's value
depends solely upon the endpoints, $A$ and $B$, regardless of the
of the particular path ${\cal P}$. To illustrate, let us consider
an alternative path denoted as ${\cal P}'$, connecting points $A$
and $B$, while ensuring that the function $f$ remains free of singularities
within the enclosed region formed by ${\cal P}\cup{\cal P}'$. Under
these conditions, the integral of $e^{nf(z)}$ across the closed path
encompassing both ${\cal P}$ and ${\cal P}'$ --- traversing from
$A$ to $B$ via ${\cal P}$ and then returning from $B$ to $A$
through ${\cal P}'$ --- becomes null, indicating that the integrals
along ${\cal P}$ and ${\cal P}'$ between $A$ and $B$ hold identical
values. In essence, this imparts us with the liberty to elect our
preferred integration path, so long as we exercise caution to avoid
traversing too closely to the opposing side of any potential singularity
point. This consideration gains significance as we proceed with our
upcoming analyses.

An additional crucial observation pertains to another fundamental
property of analytic complex functions: The \emph{maximum-modulus
theorem}. This theorem essentially states that the magnitude of an
analytic function lacks any maxima. Although a comprehensive proof
of this theorem is beyond our scope, its essence can be captured as
follows: Consider an analytic function expressed as: 
\begin{equation}
f(z)=u(z)+jv(z)=u(x,y)+jv(x,y),
\end{equation}
where $u$ and $v$ are real-valued functions. When $f$ is analytic,
the Cauchy--Riemann conditions must hold for the partial derivatives
of $u$ and $v$: 
\begin{equation}
\frac{\partial u}{\partial x}=\frac{\partial v}{\partial y};\quad\frac{\partial u}{\partial y}=-\frac{\partial v}{\partial x}.
\end{equation}
Taking the second-order partial derivative of $u$, we arrive at:
\begin{equation}
\frac{\partial^{2}u}{\partial x^{2}}=\frac{\partial^{2}v}{\partial x\partial y}=\frac{\partial^{2}v}{\partial y\partial x}=-\frac{\partial^{2}u}{\partial y^{2}},
\end{equation}
where the first and third equalities stem from the Cauchy--Riemann
conditions. Alternatively, we can write: 
\begin{equation}
\frac{\partial^{2}u}{\partial x^{2}}+\frac{\partial^{2}u}{\partial y^{2}}=0,
\end{equation}
which is recognized as the Laplace equation. Consequently, any point
where $\frac{\partial u}{\partial x}=\frac{\partial u}{\partial y}=0$
cannot be a local maximum or minimum of $u$. If it were a local maximum
along the $x$-direction, then $\frac{\partial^{2}u}{\partial x^{2}}<0$,
implying that $\frac{\partial^{2}u}{\partial y^{2}}$ must be positive,
making it a local minimum along the $y$-direction, and vice versa.
Put simply, points where partial derivatives of $u$ are zero are,
in fact, saddle points. This line of reasoning applies to the modulus
of the integrand $e^{nf(z)}$ due to: 
\begin{equation}
\bigg|\exp\{nf(z)\}\bigg|=\exp\left[n\real\{f(z)\}\right]=e^{nu(z)}.
\end{equation}
Furthermore, if $f'(z)=0$ at some $z=z_{0}$, then $u'(z_{0})=0$
as well, establishing that $z_{0}$ is a saddle point of $|e^{nf(z)}|$.
Thus, points where $f$ exhibits zero derivatives are saddle points.

Armed with this foundational understanding, let us return to our integral
$F_{n}$. Given the flexibility to select the path ${\cal P}$, suppose
we can identify a trajectory that crosses a saddle point $z_{0}$
(hence the name of the method) and where the maximum value of $|e^{nf(z)}|$
is achieved at $z=z_{0}$. In this scenario, much like in the Laplace
method, we anticipate that the integral's dominant contribution would
stem from $e^{nf(z_{0})}$. Naturally, this path would be suitable
only if it intersects the saddle point $z_{0}$ along a direction
WRT which $z_{0}$ represents a local maximum of $|e^{nf(z)}|$ or
equivalently, of $u(z)$. Moreover, for the application of our prior
Laplace method findings, we aim to configure ${\cal P}$ so that any
point $z$ in proximity to $z_{0}$, where the Taylor expansion reads
(due to the fact that $f'(z_{0})=0$): 
\begin{equation}
f(z)\approx f(z_{0})+\frac{1}{2}f''(z_{0})(z-z_{0})^{2},
\end{equation}
where the second term, $\frac{1}{2}f''(z_{0})(z-z_{0})^{2}$, is exclusively
real and negative. Consequently, it assumes a local behavior akin
to a negative parabola, mirroring the behavior observed in the Laplace
method. This implication manifests as: 
\begin{equation}
\mbox{\ensuremath{\arg}}\left\{ f''(z_{0})\right\} +2\mbox{\ensuremath{\arg}}(z-z_{0})=\pi,
\end{equation}
or equivalently: 
\begin{equation}
\mbox{\ensuremath{\arg}}(z-z_{0})=\frac{\pi-\mbox{\ensuremath{\arg}}\{f''(z_{0})\}}{2}\triangleq\theta.
\end{equation}
In essence, ${\cal P}$ should traverse $z_{0}$ along the direction
$\theta$. This orientation is called the \emph{axis} of $z_{0}$
and can be demonstrated to be the direction of steepest descent from
the summit at $z_{0}$ --- hence the name \emph{steepest-descent
method}. Notably, it is worth mentioning that in the $\theta-\pi/2$
direction, which stands perpendicular to the axis, $\mbox{\ensuremath{\arg}}[f''(z_{0})(z-z_{0})^{2}]=\pi-\pi=0$.
Consequently, $f''(z_{0})(z-z_{0})^{2}$ emerges as real and positive
in this direction, akin to a positive parabolic pattern. This indicates
that along this direction, $z_{0}$ constitutes a local minimum.

Visually speaking, our strategy involves the selection of a path ${\cal P}$
connecting $A$ to $B$, constructed as three distinct segments (as
depicted in Figure \ref{fig: saddlepointpath}): $A\to A'$ and $B'\to B$
form the arbitrary initial and final sections of the integral path.
The middle part, connecting $A'$ to $B'$ and localized near $z_{0}$,
consists of a straight line aligned with the axis of $z_{0}$. 
\begin{figure}[t]
\centering{}\includegraphics[scale=0.6]{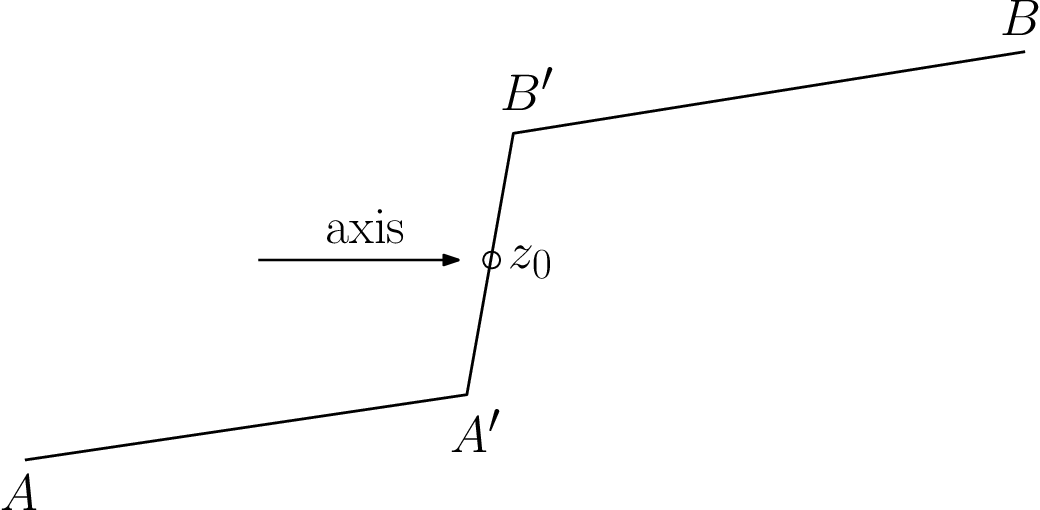}\caption{{\small{}A path ${\cal P}$ from $A$ to $B$, passing via $z_{0}$
along the axis. }\label{fig: saddlepointpath}}
\end{figure}

Accordingly, let us decompose $F_{n}$ into its three parts: 
\begin{equation}
F_{n}=\int_{A}^{A'}e^{nf(z)}\dd z+\int_{A'}^{B'}e^{nf(z)}\dd z+\int_{B'}^{B}e^{nf(z)}\dd z.
\end{equation}
As for the first and the third terms, 
\begin{align}
\bigg|\left(\int_{A}^{A'}+\int_{B'}^{B}\right)\dd ze^{nf(z)}\bigg| & \le\left(\int_{A}^{A'}+\int_{B'}^{B}\right)\dd z|e^{nf(z)}|\nonumber \\
 & =\left(\int_{A}^{A'}+\int_{B'}^{B}\right)\dd ze^{n\real\{f(z)\}}
\end{align}
whose contribution is negligible compared to $e^{nf(z_{0})}$, exactly
like the tails in the Laplace method. As for the middle integral,
\begin{equation}
\int_{A'}^{B'}e^{nf(z)}\dd z\approx e^{nf(z_{0})}\int_{A'}^{B'}\exp\left\{ \frac{nf''(z_{0})(z-z_{0})^{2}}{2}\right\} \dd z.
\end{equation}
By transitioning from the complex integration variable $z$ to the
real variable $x$, ranging from $-\delta$ to $+\delta$, with $z=z_{0}+xe^{j\theta}$
(following the axis direction), we end up with exactly the Gaussian
integral encountered in the Laplace method, resulting in: 
\begin{equation}
\int_{A'}^{B'}\exp\{nf''(z_{0})(z-z_{0})^{2}/2\}\dd z=e^{j\theta}\sqrt{\frac{2\pi}{n|f''(z_{0})|}}
\end{equation}
where the factor $e^{j\theta}$ is due to the change of variable ($\dd z=e^{j\theta}\dd x$).
Thus, 
\begin{equation}
F_{n}\sim e^{j\theta}\cdot e^{nf(z_{0})}\sqrt{\frac{2\pi}{n|f''(z_{0})|}},
\end{equation}
and somewhat more generally, 
\begin{equation}
\int_{{\cal P}}g(z)e^{nf(z)}\dd z\sim e^{j\theta}g(z_{0})e^{nf(z_{0})}\sqrt{\frac{2\pi}{n|f''(z_{0})|}}.
\end{equation}

Consider next a few simple examples. 
\begin{example}[The size of a type class of binary sequences]
\emph{ To count the number of binary sequences of length $n$ with
exactly $k$ 1's and $(n-k)$ 0's, we use the notation $m_{k}$. Let
us examine the complex function 
\begin{equation}
M(z)=(1+z^{-1})^{n}=\sum_{k=0}^{n}m_{k}z^{-k}.
\end{equation}
The second equality expresses the fact that $M(z)$ can be viewed
as the Z-transform of the sequence $\{m_{k}\}_{k=0}^{n}$, and so,
$m_{k}$ is given by the inverse Z-transform of $M(z)$: 
\begin{align}
m_{k} & =\frac{1}{2\pi j}\oint_{{\cal P}}(1+z^{-1})^{n}z^{k-1}\dd z\nonumber \\
 & =\frac{1}{2\pi j}\oint_{{\cal P}}\frac{1}{z}\exp\left\{ n\left[\ln(1+z^{-1})+q\ln z\right]\right\} \dd z,
\end{align}
where $q=k/n$ and ${\cal P}$ is any counterclockwise closed path
that surrounds the origin. Here, $g(z)=1/z$ and 
\begin{equation}
f(z)=\ln(1+z^{-1})+q\ln z=\ln(1+z)-(1-q)\ln z,
\end{equation}
whose saddle point is $z_{0}=\frac{1-q}{q}$. If we choose ${\cal P}$
to be the circle $|z|=\frac{1-q}{q}$, it intersects the point $z_{0}$,
situated on the real line, in a vertical manner. Remarkably, this
alignment corresponds to the axis of $z_{0}$. A straightforward calculation
yields 
\begin{equation}
f''(z_{0})=\frac{q^{3}}{1-q}
\end{equation}
which gives 
\begin{equation}
m_{k}\sim\frac{e^{j\pi/2}}{z_{0}}\cdot e^{nf(z_{0})}\cdot\frac{1}{2\pi j}\cdot\sqrt{\frac{2\pi}{n|f'(z_{0})|}}=\frac{e^{j\pi/2}}{(1-q)/q}\cdot e^{nH(q)}\cdot\frac{1}{2\pi j}\cdot\sqrt{\frac{2\pi(1-q)}{nq^{3}}}=\frac{e^{nH(q)}}{\sqrt{2\pi nq(1-q)}},
\end{equation}
where, as before, $H(q)\triangleq-q\ln q-(1-q)\ln(1-q)$ is the binary
entropy function. }

\emph{Another approach to assess $m_{k}$ is to present it as 
\begin{equation}
m_{k}=\sum_{\boldsymbol{x}\in\{0,1\}^{n}}\delta\left(k-\sum_{i=1}^{n}x_{i}\right),
\end{equation}
where $\delta(\cdot)$ is the Kroenecker delta function, which in
turn is represented as the inverse Fourier transform of the unit spectrum:
\begin{equation}
\delta\left(k-\sum_{i=1}^{n}x_{i}\right)=\frac{1}{2\pi}\int_{-\pi}^{\pi}\exp\left\{ j\omega\left(k-\sum_{i=1}^{n}x_{i}\right)\right\} \dd\omega.
\end{equation}
Upon substituting this identity into the above representation of $m_{k}$
and rearranging the order of the summation and the integration, the
outer integral can be assessed using the saddle-point method. The
reader is referred to \cite[Section 4.3, Example 2, pp.\ 108-109]{me09}
for further details.}
\end{example}
Our next example belongs to the realm of continuous alphabets.
\begin{example}[Surface area of a hyper-sphere]
\emph{\label{exa: hyperspheresurface}  This example is closely connected
to the concept of simple Gaussian-type classes, as discussed in Chapter
\ref{sec: CMoT}. While there exists an exact closed-form expression
for the surface area of an $n$-dimensional Euclidean hyper-sphere,
we explore this example to illustrate the asymptotic accuracy of the
saddle-point method. Our starting point is the representation of the
surface area of an $n$-dimensional Euclidean hyper-sphere with radius
$r$ as follows: 
\begin{equation}
S_{n}(r)=2r\int_{\reals^{n}}\delta\left(r^{2}-\sum_{i=1}^{n}x_{i}^{2}\right)\dd\boldsymbol{x},
\end{equation}
where $\delta(\cdot)$ designates the Dirac delta function. To see
why this true, observe that $S_{n}(r)$ integrates to 
\begin{align}
V_{n}(R) & =\int_{0}^{R}S_{n}(r)\dd r\nonumber \\
 & =\int_{0}^{R}2r\int_{\reals^{n}}\delta\left(r^{2}-\sum_{i=1}^{n}x_{i}^{2}\right)\dd\boldsymbol{x}\dd r\nonumber \\
 & =\int_{\reals^{n}}\left[\int_{0}^{R}2r\delta\left(r^{2}-\sum_{i=1}^{n}x_{i}^{2}\right)\dd r\right]\dd\boldsymbol{x}\nonumber \\
 & =\int_{\reals^{n}}\left[\int_{0}^{R^{2}}\delta\left(r^{2}-\sum_{i=1}^{n}x_{i}^{2}\right)\dd(r^{2})\right]\dd\boldsymbol{x}\nonumber \\
 & =\int_{\reals^{n}}U\left(R^{2}-\sum_{i=1}^{n}x_{i}^{2}\right)\dd\boldsymbol{x}\nonumber \\
 & =\Vol\left\{ \boldsymbol{x}\colon\sum_{i=1}^{n}x_{i}^{2}\le R^{2}\right\} ,
\end{align}
where $U(\cdot)$ is the unit step function. Thus, the integral of
$S_{n}(r)$ across the interval $[0,R]$ yields the volume of a hyper-sphere
of radius $R$, and so, $S_{n}(r)$ is the surface area of a hyper-sphere
of radius $r$. We next represent the Dirac delta function as the
inverse Fourier transform of the unit function, i.e., 
\begin{equation}
\delta(t)=\frac{1}{2\pi}\int_{-\infty}^{\infty}e^{j\omega t}\dd\omega,
\end{equation}
and so, referring to Chapter \ref{sec: CMoT}, the surface area of
sphere of radius $\sqrt{ns}$ is given as follows. Let $\vartheta>0$
be some positive real, to be chosen shortly. Then, 
\begin{align}
S_{n}(\sqrt{ns}) & =2\sqrt{ns}\int_{\reals^{n}}\dd\boldsymbol{x}\cdot\delta\left(ns-\sum_{i=1}^{n}x_{i}^{2}\right)\nonumber \\
 & \trre[=,a]2\sqrt{ns}e^{n\vartheta s}\int_{\reals^{n}}\dd\boldsymbol{x}\exp\left\{ -\vartheta\sum_{i=1}^{n}x_{i}^{2}\right\} \cdot\delta\left(ns-\sum_{i=1}^{n}x_{i}^{2}\right)\nonumber \\
 & =2\sqrt{ns}e^{n\vartheta s}\int_{\reals^{n}}\dd\boldsymbol{x}\exp\left\{ -\vartheta\sum_{i=1}^{n}x_{i}^{2}\right\} \int_{-\infty}^{+\infty}\frac{\dd\omega}{2\pi}\exp\left\{ j\omega\left(ns-\sum_{i=1}^{n}x_{i}^{2}\right)\right\} \nonumber \\
 & =\sqrt{ns}e^{n\vartheta s}\int_{-\infty}^{+\infty}\frac{\dd\omega}{\pi}e^{j\omega ns}\int_{\reals^{n}}\dd\boldsymbol{x}\exp\left\{ -(\vartheta+j\omega)\sum_{i=1}^{n}x_{i}^{2}\right\} \nonumber \\
 & =\sqrt{ns}e^{n\vartheta s}\int_{-\infty}^{+\infty}\frac{\dd\omega}{\pi}e^{j\omega ns}\left[\int_{\reals}\dd xe^{-(\vartheta+j\omega)x^{2}}\right]^{n}\nonumber \\
 & \trre[=,b]\sqrt{ns}e^{n\vartheta s}\int_{-\infty}^{+\infty}\frac{\dd\omega}{\pi}e^{j\omega ns}\left(\frac{\pi}{\vartheta+j\omega}\right)^{n/2}\nonumber \\
 & =\sqrt{ns}\pi^{n/2-1}\int_{-\infty}^{+\infty}\dd\omega\exp\left\{ n\left[(\vartheta+j\omega)s-\frac{1}{2}\ln(\vartheta+j\omega)\right]\right\} \nonumber \\
 & =\sqrt{ns}\cdot\pi^{n/2-1}\cdot\frac{1}{j}\cdot\int_{\vartheta-j\infty}^{\vartheta+j\infty}\dd z\exp\left\{ n\left[zs-\frac{1}{2}\ln z\right]\right\} ,
\end{align}
where in $(a)$ we have multiplied the expression by $e^{n\vartheta s}$
outside the integral and by $e^{-\vartheta\sum_{i}x^{2}}$ inside
the integral, but $e^{-\vartheta\sum_{i}x^{2}}=e^{-n\vartheta s}$
wherever the delta function of the integrand does not vanish, and
so, these two multiplications cancel each other. This step is crucial
for the subsequent steps. In $(b)$ we have applied complex Gaussian
integration. In this case, we have 
\begin{equation}
f(z)=zs-\frac{1}{2}\ln z,
\end{equation}
and the integration is along an arbitrary vertical straight line $\real\{z\}=\vartheta$.
We select this straight line to cross the saddle-point, that is, $\vartheta=z_{0}=\frac{1}{2s}$,
where 
\begin{equation}
f(z_{0})=\frac{1}{2}\ln(2es)
\end{equation}
and 
\begin{equation}
f''(z_{0})=2s^{2}.
\end{equation}
Once again, the axis is vertical, and so, 
\begin{equation}
S_{n}(\sqrt{ns})\sim\sqrt{ns}\cdot\pi^{n/2-1}\cdot\frac{1}{j}\cdot e^{j\pi/2}\cdot\exp\left\{ \frac{n}{2}\ln(2es)\right\} \cdot\sqrt{\frac{2\pi}{2s^{2}n}}=\frac{(2\pi es)^{n/2}}{\sqrt{\pi s}},
\end{equation}
which agrees with the result of Chapter \ref{sec: CMoT}. Note that
the representation of $\delta\left(ns-\sum_{i=1}^{n}x_{i}^{2}\right)$
as an inverse Fourier transform converted the integrand into an exponential
function of $(ns-\sum_{i=1}^{n}x_{i}^{2})$, which is a product form
and hence can be represented as a product of identical integrals,
which is actually one-dimensional integral raised to the power of
$n$. }
\end{example}
Note that in the above derivation, when we shifted the vertical integration
path from the imaginary axis, $\{z\colon\real\{z\}=0\}$, to the parallel
vertical line $\{z\colon\real\{z\}=\vartheta\}$, we have actually
replaced the inverse Fourier transform by the inverse Laplace transform.
By the same token, we can handle the volume of the $n$-dimensional
hyper-sphere as 
\begin{equation}
V_{n}(ns)=\int_{\reals^{n}}U\left(ns-\sum_{i=1}^{n}x_{i}^{2}\right)\dd\boldsymbol{x}
\end{equation}
with the representation of the unit step function as the inverse Laplace
transform of $1/z$, which amounts to substituting 
\begin{equation}
U\left(ns-\sum_{i=1}^{n}x_{i}^{2}\right)=\frac{1}{2\pi j}\int_{\real\{z\}=\vartheta}\frac{\dd z}{z}\exp\left\{ z\left(ns-\sum_{i=1}^{n}x_{i}^{2}\right)\right\} ,
\end{equation}
and interchanging the order of the integration. The saddle-point approximation
of this expression is very similar to the above. We next demonstrate
how this is done in the context of assessing a probability of a large-deviations
event.
\begin{example}[Large deviations]
\emph{\label{exa: largedeviations}This example delves into a topic
that was extensively studied by Bahadur and Rao \cite{BR60}. Here,
we offer a partial exposition to illustrate the application of the
saddle-point method. Consider a set of IID RVs $X_{1},X_{2},\ldots,X_{n}$,
all of which are independent copies of a real RV $X$ with a PDF $p(x)$.
Additionally, let $A$ be a constant greater than the expected value
of $X$. We aim to evaluate the probability of a large-deviations
event, namely, $\{\sum_{i=1}^{n}X_{i}\ge nA\}$, utilizing the saddle-point
method. Introducing $\theta$ as an arbitrary positive real number,
we have: 
\begin{align}
\Pr\left\{ \sum_{i=1}^{n}X_{i}\ge nA\right\}  & =\int_{\reals^{n}}U\left(\sum_{i=1}^{n}x_{i}-nA\right)\prod_{i=1}^{n}p(x_{i})\dd\boldsymbol{x}\nonumber \\
 & =\int_{\reals^{n}}\frac{1}{2\pi j}\int_{\real\{z\}=\theta}\frac{\dd z}{z}\exp\left\{ z\left(\sum_{i=1}^{n}x_{i}-nA\right)\right\} \cdot\prod_{i=1}^{n}p(x_{i})\dd\boldsymbol{x}\nonumber \\
 & =\frac{1}{2\pi j}\int_{\real\{z\}=\theta}\frac{e^{-znA}}{z}\dd z\int_{\reals^{n}}\prod_{i=1}^{n}\left[p(x_{i})e^{zx_{i}}\right]\dd\boldsymbol{x}\nonumber \\
 & =\frac{1}{2\pi j}\int_{\real\{z\}=\theta}\frac{e^{-znA}}{z}\dd z\left[\int_{\reals}p(x)e^{zx}\dd x\right]^{n}\nonumber \\
 & =\frac{1}{2\pi j}\int_{\real\{z\}=\theta}\frac{\dd z}{z}\exp\left\{ n\left[\ln\left(\int_{\reals}p(x)e^{zx}\dd x\right)-zA\right]\right\} ,
\end{align}
and we can apply}\footnote{\emph{There is a non-trivial issue concerning the non-analyticity
of the logarithmic function, whose argument, $\int_{\reals}p(x)e^{zx}\dd x$,
may surround the origin infinitely many times while $z$ exhausts
the vertical line $\real\{z\}=\theta$, because the origin is a singular
point of the logarithmic function. This requires to pass among different
branches of the logarithmic function along the journey from $\theta-j\infty$
to $\theta+j\infty$. This issue is discussed in detail in \cite{me23b}.}}\emph{ the saddle-point method with $g(z)=1/z$ and 
\begin{equation}
f(z)=\ln\left(\int_{\reals}p(x)e^{zx}\dd x\right)-zA.
\end{equation}
Consider the function $f$ confined to the reals, namely, $f(s)$,
where $s\in\reals$. Since $f(s)$ is a convex function, it can be
shown that its derivative vanishes uniquely at some finite real $s=s_{\star}>0$,
provided that $A<x_{\max}\triangleq\sup_{\{x\colon p(x)>0\}}x$. Then,
$z=s_{\star}$ is a saddle-point of $f$. Let us first assume that
$p$ is such that $z=s_{\star}$ is the only saddle-point of $f$
in the entire complex plane (shortly, we also address situations where
this is not the case). In this case, a simple application of the saddle-point
method suggests to select $\theta=s_{\star}$, where the axis is vertical,
and so, 
\begin{align}
\mbox{\ensuremath{\Pr}}\left\{ \sum_{i=1}^{n}X_{i}\ge nA\right\}  & \sim\frac{1}{s_{\star}}\cdot\frac{e^{j\pi/2}}{2\pi j}\cdot\exp\left\{ n\left[\ln\left(\int_{\reals}p(x)e^{s_{\star}x}\dd x\right)-s_{\star}A\right]\right\} \cdot\sqrt{\frac{2\pi}{nV(s_{\star})}}\nonumber \\
 & =\frac{\exp\left\{ n\left[\ln\left(\int_{\reals}p(x)e^{s_{\star}x}\dd x\right)-s_{\star}A\right]\right\} }{s_{\star}\sqrt{2\pi nV(s_{\star})}},
\end{align}
where $V(s)=f''(s)=\Var_{s}\{X\}$, with the latter being defined
as the variance of $X$ WRT the PDF that is proportional to $p(x)e^{sx}$,
i.e., the titled PDF. It is worth highlighting the intriguing similarity
between the exponential term 
\begin{equation}
\exp\left\{ n\left[\ln\left(\int_{\reals}p(x)e^{s_{\star}x}\dd x\right)-s_{\star}A\right]\right\} ,
\end{equation}
and the Chernoff bound, as $s_{\star}$ minimizes $f(s)$ over the
real numbers. At the same time, $z=s_{\star}$ is determined as the
saddle-point that dominates the integration along the vertical line
defined by $\real\{z\}=s_{\star}$. This observation aligns with the
modulus theorem: Given that $z=s_{\star}$ minimizes $|e^{nf(z)}|=e^{nf(s)}$
horizontally along the real line, it maximizes $|e^{nf(z)}|$ along
the vertical direction of the integration path. While the exponential
behavior of the saddle-point approximation mirrors that of the Chernoff
bound, known for its exponential tightness, it is noteworthy that
the former provides a more refined characterization, including the
correct pre-exponential factor, which is given by $1/[s_{\star}\sqrt{2\pi nV(s_{\star})}]$. }

\emph{As previously mentioned, in the earlier derivation, we made
the assumption that $z=s_{\star}$ represents the sole saddle-point
of the function $f$ across the entire complex plane. However, this
assumption does not hold universally. Let us consider a scenario in
which $X$ is a lattice RV, implying that $X$ can only assume values
that are integer multiples of a constant $\Delta>0$, that is, 
\begin{equation}
p(x)=\sum_{i=-\infty}^{\infty}\alpha_{i}\delta(x-i\Delta),
\end{equation}
where $\delta(\cdot)$ is the Dirac delta function and $\{\alpha_{i}\}$
are non-negative reals which sum up to unity. Consider the vertical
line of integration, $z=s_{\star}+j\omega$, $-\infty<\omega<\infty$.
In this scenario, it becomes evident that if $s_{\star}$ is a saddle-point
of $e^{nf(z)}$, then so are the points $s_{\star}+j\Omega k$, where
$k$ ranges over all integers ($k=0,\pm1,\pm2,...$), and $\Omega$
is defined as $\Omega=2\pi/\Delta$. This is due to the periodic nature
of $|e^{nf(z)}|$, which is equivalent to $e^{n\real\{f(z)\}}$, along
the vertical direction with a period of $\Omega$. Indeed, 
\begin{align}
\real\{f(s_{\star}+jk\Omega)\} & =\real\left\{ \ln\left[\int_{\reals}p(x)e^{(s_{\star}+jk\Omega)x}\dd x\right]-(s_{\star}+jk\Omega)A\right\} \nonumber \\
 & =\real\left\{ \ln\left[\sum_{i=-\infty}^{\infty}\alpha_{i}e^{(s_{\star}+jk\Omega)i\Delta}\right]\right\} -s_{\star}A\nonumber \\
 & =\real\left\{ \ln\left[\sum_{i=-\infty}^{\infty}\alpha_{i}e^{s_{\star}i\Delta}e^{jki\Omega\Delta}\right]\right\} -s_{\star}A\nonumber \\
 & =\real\left\{ \ln\left[\sum_{i=-\infty}^{\infty}\alpha_{i}e^{s_{\star}i\Delta}e^{j2\pi ik}\right]\right\} -s_{\star}A\nonumber \\
 & =\real\left\{ \ln\left[\sum_{i=-\infty}^{\infty}\alpha_{i}e^{s_{\star}i\Delta}\right]\right\} -s_{\star}A\nonumber \\
 & =\real\{f(s_{\star})\}.
\end{align}
In such a situation, during the integration along the line $\real\{z\}=s_{\star}$,
the contributions from all saddle-points, $s_{\star}+jk\Omega$ for
$k=0,\pm1,\pm2,\ldots$, carry equal significance, collectively dominating
the exponential rate of the integral. This has a notable impact on
the pre-exponential factor, which now needs to be adjusted to reflect
this collective contribution. Therefore, the modified pre-exponential
factor is given by: 
\begin{align}
\frac{1}{\sqrt{2\pi nV(s_{\star})}}\cdot\sum_{k=-\infty}^{\infty}\frac{e^{-jk\Omega An}}{s_{\star}+jk\Omega} & =\sqrt{\frac{2\pi}{nV(s_{\star})}}\cdot\frac{1}{2\pi}\int_{-\infty}^{\infty}e^{-j\omega nA}\cdot\frac{1}{s_{\star}+j\omega}\cdot\left[\sum_{k=-\infty}^{\infty}\delta(\omega-k\Omega)\right]\dd\omega\nonumber \\
 & \trre[=,a]\sqrt{\frac{2\pi}{nV(s_{\star})}}\cdot\left\{ \left[e^{-s_{\star}t}U(t)\right]\star\left[\frac{1}{\Omega}\sum_{k=-\infty}^{\infty}\delta\left(t-\frac{2\pi k}{\Omega}\right)\right]\right\} \Bigg|_{t=-nA}\nonumber \\
 & =\frac{1}{\Omega}\sqrt{\frac{2\pi}{nV(s_{\star})}}\sum_{k=-\infty}^{\infty}e^{-s_{\star}(-nA-2\pi k/\Omega)}U\left(-nA-\frac{2\pi k}{\Omega}\right)\nonumber \\
 & =\frac{1}{\Omega}\sqrt{\frac{2\pi}{nV(s_{\star})}}\cdot\exp\left\{ -s_{\star}\left[(-nA)\mod\left(\frac{2\pi}{\Omega}\right)\right]\right\} \cdot\sum_{k=0}^{\infty}e^{-s_{\star}\cdot2\pi k/\Omega}\nonumber \\
 & =\sqrt{\frac{2\pi}{nV(s_{\star})}}\cdot\frac{\exp\left\{ -s_{\star}\left[(-nA)\mod\left(\frac{2\pi}{\Omega}\right)\right]\right\} }{\Omega(1-e^{-2\pi s_{\star}/\Omega})}\nonumber \\
 & =\sqrt{\frac{2\pi}{nV(s_{\star})}}\cdot\frac{\Delta\exp\{-s_{\star}[(-nA)\mod\Delta]\}}{2\pi(1-e^{-s_{\star}\Delta})}\nonumber \\
 & =\sqrt{\frac{1}{2\pi nV(s_{\star})}}\cdot\frac{\Delta\exp\{-s_{\star}[(-nA)\mod\Delta]\}}{1-e^{-s_{\star}\Delta}},
\end{align}
where in $(a)$ we have used the fact that inverse Fourier transform
of the product of two frequency-domain functions is equal to the convolution
between the individual inverse Fourier transforms. The oscillatory
factor in the numerator, $\exp\{-s_{\star}[(-nA)\mod\Delta]\}$, illustrates
the granularity inherent in the probability quanta related to the
lattice-like nature of the involved RVs (also discussed in \cite{me23b}).
It is worth noting that the non-lattice scenario can be considered
as a specific case of the lattice scenario, where $\Delta\to0$. }
\end{example}
Our final example pertains to the enumeration of codewords within
a hyper-cubical lattice subject to an $L_{1}$ power constraint. The
motivation here is to evaluate the coding rate of a hyper-cubical
lattice code. In a nutshell, when the hyper-cubes are exceptionally
small, this count approximates the ratio between the volume of the
$L_{1}$ hyper-sphere defining the power constraint and the volume
of the hyper-cube. However, the saddle-point method provides a more
precise estimation.
\begin{example}[Number of codewords of a power-limited lattice code]
\emph{\label{exa: latticecodes}Let us examine a hyper-cubical lattice
code, where the codewords take the form of $(k_{1}\Delta,k_{2}\Delta,\ldots,k_{n}\Delta)$,
with $\Delta>0$ given, $\{k_{i}\}$ being integers, and adhering
to the $L_{1}$ power constraint $\Delta\sum_{i=1}^{n}|k_{i}|\le nQ$.
What is the number $M$ of lattice codewords that can be found? We
can establish the following sequence of equalities: 
\begin{align}
M & =\sum_{k_{1}=-\infty}^{\infty}\ldots\sum_{k_{n}=-\infty}^{\infty}U\left[nQ-\Delta\sum_{i=1}^{n}|k_{i}|\right]\nonumber \\
 & =\sum_{k_{1}=-\infty}^{\infty}\ldots\sum_{k_{n}=-\infty}^{\infty}\frac{1}{2\pi j}\int_{\real\{z\}=\theta}\frac{\dd z}{z}\exp\left\{ z\left[nQ-\Delta\sum_{i=1}^{n}|k_{i}|\right]\right\} \nonumber \\
 & =\frac{1}{2\pi j}\int_{\real\{z\}=\theta}\frac{\dd ze^{nQz}}{z}\sum_{k_{1}=-\infty}^{\infty}\ldots\sum_{k_{n}=-\infty}^{\infty}\exp\left\{ -\Delta z\sum_{i=1}^{n}|k_{i}|\right\} \nonumber \\
 & =\frac{1}{2\pi j}\int_{\real\{z\}=\theta}\frac{\dd ze^{nQz}}{z}\left[\sum_{k=-\infty}^{\infty}\exp\{-\Delta z|k|\}\right]^{n}\nonumber \\
 & =\frac{1}{2\pi j}\int_{\real\{z\}=\theta}\frac{\dd ze^{nQz}}{z}\left[\frac{e^{\Delta z}+1}{e^{\Delta z}-1}\right]^{n}\nonumber \\
 & =\frac{1}{2\pi j}\int_{\real\{z\}=\theta}\frac{\dd z}{z}\exp\left\{ n\left[Qz-\ln\tanh\left(\frac{\Delta z}{2}\right)\right]\right\} .
\end{align}
Thus, the saddle-point method can be applied with $g(z)=1/z$ and
\begin{equation}
f(z)=Qz-\ln\tanh\left(\frac{\Delta z}{2}\right)\equiv Qz-\ln\sinh\left(\frac{\Delta z}{2}\right)+\ln\cosh\left(\frac{\Delta z}{2}\right).
\end{equation}
The derivative of $f$ vanishes at 
\begin{equation}
z=s_{\star}=\frac{1}{\Delta}\ln\left(\frac{\Delta}{Q}+\sqrt{\frac{\Delta^{2}}{Q^{2}}+1}\right),
\end{equation}
but similarly as in Example \ref{exa: largedeviations}, here too,
$\real\{f(z)\}$ is periodic in the vertical direction with period
$\Omega=2\pi/\Delta$, and so, there are infinitely many saddle-points
$\{s_{\star}+jk\Omega,~k=0,\pm1,\pm2,\ldots\}$, and $M$ is exponentially
$e^{nf(s_{\star})}$ with the same pre-exponential factor as in the
lattice case of Example \ref{exa: largedeviations}, except that $(-nA)\mod\Delta$
is replaced by $(nQ)\mod\Delta$ and $V(s_{\star})$ is replaced by
$|f''(s_{\star})|$. Therefore, the coding rate (in nats per channel
use) is of the form, 
\begin{equation}
R=\frac{\ln M}{n}=f(s_{\star})-\frac{\ln n}{2n}+o\left(\frac{\ln n}{n}\right),
\end{equation}
with 
\begin{equation}
f(s_{\star})=\frac{Q}{\Delta}\ln\left(\frac{\Delta}{Q}+\sqrt{\frac{\Delta^{2}}{Q^{2}}+1}\right)+\ln\left(\frac{\Delta}{Q}+\sqrt{\frac{\Delta^{2}}{Q^{2}}+1}+1\right)-\ln\left(\frac{\Delta}{Q}+\sqrt{\frac{\Delta^{2}}{Q^{2}}+1}-1\right).
\end{equation}
It is easy to verify that when $\Delta/Q\ll1$, the exponential factor,
$e^{nf(s_{\star})}$ is approximately $\frac{(2eQ)^{n}}{\Delta^{n}}$,
which is exponentially the ratio between volume of the $L_{1}$-hyper-sphere
of `radius' $nQ$ and the volume of the hyper-cube, $\Delta^{n}$.
We skip the details of calculating $f''(s_{\star})$ for the pre-exponent. }

\emph{In conclusion, we note that a similar calculation for the more
traditional $L_{2}$ power constraint involves dealing with the infinite
summation $\sum_{k}e^{-z\Delta^{2}k^{2}}$ (instead of $\sum_{k}e^{-\Delta z|k|}$
as in our previous analysis). Although this expression lacks an apparent
closed-form representation, the same fundamental behavior persists:
The rate remains primarily determined by the log-volume ratio, subtracting
$\frac{\ln n}{2n}$, with some negligible terms. }
\end{example}
In Example \ref{exa: hyperspheresurface}, we witnessed the formidable
capability of the saddle-point method in assessing type class measures
without the need for the $\epsilon$-inflation technique employed
in Chapter \ref{sec: CMoT}. When confronted with the task of integrating
over $\boldsymbol{x}$ a function of the form $f(\sum_{i=1}^{n}x_{i}^{2})$,
we can conveniently rewrite this as an equivalent integral over $f(r)S_{n}(r)$
WRT $r$. This transformation effectively replaces the $n$-dimensional
integration with a one-dimensional integration, which, in certain
cases, can be well-approximated using either the Laplace method or
the saddle-point method.

In Chapter \ref{sec: CMoT}, we explored more intricate type classes
defined as intersections between hyper-sphere surfaces and hyper-planes,
such as $\sum_{i=1}^{n}x_{i}=nc$. Evaluating the Lebesgue measure
of such objects involves integrating a product of delta functions,
specifically $\delta(ns-\sum_{i}x_{i}^{2})\cdot\delta(nc-\sum_{i}x_{i})$.
To compute this measure, we represent each delta function as an inverse
Laplace transform separately, each with its own complex integration
variable, \emph{i.e.}, 
\begin{align}
 & \int_{\reals^{n}}\delta\left(ns-\sum_{i=1}^{n}x_{i}^{2}\right)\cdot\delta\left(nc-\sum_{i=1}^{n}x_{i}\right)\dd\boldsymbol{x}\nonumber \\
 & =\int_{\reals^{n}}\left[\frac{1}{(2\pi j)^{2}}\int_{\theta-j\infty}^{\theta+j\infty}\int_{\nu-j\infty}^{\nu+j\infty}\dd z_{1}\dd z_{2}\exp\left\{ z_{1}\left(ns-\sum_{i=1}^{n}x_{i}^{2}\right)+z_{2}\left(nc-\sum_{i=1}^{n}x_{i}\right)\right\} \right]\dd\boldsymbol{x}\nonumber \\
 & =\frac{1}{(2\pi j)^{2}}\int_{\theta-j\infty}^{\theta+j\infty}\int_{\nu-j\infty}^{\nu+j\infty}\dd z_{1}\dd z_{2}\int_{\reals^{n}}\exp\left\{ z_{1}\left(ns-\sum_{i=1}^{n}x_{i}^{2}\right)+z_{2}\left(nc-\sum_{i=1}^{n}x_{i}\right)\right\} \dd\boldsymbol{x}\nonumber \\
 & =\frac{1}{(2\pi j)^{2}}\int_{\theta-j\infty}^{\theta+j\infty}\int_{\nu-j\infty}^{\nu+j\infty}\dd z_{1}\dd z_{2}e^{n(zs+z'c)}\left[\int_{\reals}\exp\left\{ -(z_{1}x^{2}+z_{2}x)\right\} \dd x\right]^{n}\nonumber \\
 & =\frac{1}{(2\pi j)^{2}}\int_{\theta-j\infty}^{\theta+j\infty}\int_{\nu-j\infty}^{\nu+j\infty}\dd z_{1}\dd z_{2}e^{n(z_{1}s+z_{2}c)}\left[\exp\left\{ \frac{z_{2}^{2}}{4z_{1}^{2}}\right\} \sqrt{\frac{\pi}{z_{1}}}\right]^{n}\nonumber \\
 & =\frac{\pi^{n/2}}{(2\pi j)^{2}}\int_{\theta-j\infty}^{\theta+j\infty}\int_{\nu-j\infty}^{\nu+j\infty}\dd z_{1}\dd z_{2}\exp\left\{ n\left[z_{1}s+z_{2}c+\frac{z_{2}^{2}}{4z_{1}^{2}}-\frac{\ln z_{1}}{2}\right]\right\} ,
\end{align}
where $\theta$ and $\nu$ are arbitrary positive reals. In cases
like this, an extension of the saddle-point method to the multivariate
setting is required, as outlined in \cite{Neuschel14}. Building on
these insights, if we encounter the need to integrate a function of
the form $f(\sum_{i=1}^{n}x_{i}^{2},\sum_{i=1}^{n}x_{i})$, we can
transform it into a two-dimensional integration of $f$ multiplied
by the Lebesgue measure of the corresponding type class, following
a similar procedure to what was just described. These considerations
are applicable to types defined by any fixed number of constraints,
including those related to conditional types (e.g., constraints involving
$\sum_{i=1}^{n}x_{i}y_{i}$) and constraints associated with Gauss--Markov
types (such as constraints specifying values of $\sum_{i=1}^{n}x_{i}x_{i-\ell}$
for $\ell=1,2,\ldots,k$). Notably, the saddle-point method allows
for the combination of constraints, even those involving $\sum_{i=1}^{n}x_{i}y_{i}$
and $\sum_{i=1}^{n}x_{i}x_{i-\ell}$. This capability resolved an
outstanding challenge posed in \cite{me93} and was successfully addressed
in \cite{HSMM19}, particularly in the context of the Gaussian intersymbol
interference channel, thanks to the versatility of the saddle-point
method.

Extending this generality further, instead of linear and quadratic
constraints, situations may arise with constraints involving combinations
of empirical means of arbitrary functions, denoted as $\sum_{i=1}^{n}\phi_{j}(x_{i})$
for $j=1,2,\ldots,k$. The associated saddle-point integration in
these cases will involve exponential functions of linear combinations
of these statistics. It is important to note that the coefficients
of these linear combinations can be complex in general. In essence,
this entails working with exponential families characterized by complex
parameters.

\subsection{Applications}

The saddle-point method has found extensive applications in various
disciplines, including probability theory, mathematical statistics,
and physics, with notable usage in statistical physics. While less
common in the information theory community, there have been exceptions
in the last two decades.

In Example \ref{exa: univcoding}, we demonstrated how the Laplace
integration method can be effectively employed to approximate Bayesian
mixtures of memoryless sources, particularly relevant to universal
source coding \cite{Davisson73}, \cite{KT81}. Schwartz also utilized
this approximation to derive a model order estimator from a Bayesian
perspective within a sequence of nested parametric families \cite{Schwartz78}.

Several researchers have applied the Laplace and saddle-point methods
to obtain more refined bounds on the error probability of channel
coding and decoding, including characterizations of the pre-exponential
factor, in addition to the exponential one. Notable contributors to
this area include Atlu\u{g} and Wagner \cite{AW14}, Font-Segura,
Vasquez-Vilar, Martinez, and Guill\'{e}n i F\`{a}bregas \cite{FVMFL18},
Honda \cite{Honda18}, Martinez and Guill\'{e}n i F\`{a}bregas \cite{MF11a},
\cite{MF11b}, and Scarlett, Martinez and Guill\'{e}n i F\`{a}bregas
\cite{SMF14}. These methods have also been applied to derive sharper
bounds on the probability of error in binary hypothesis testing \cite{VFKL18}.

Furthermore, the saddle-point and Laplace methods have been applied
to finite block-length analysis and higher-order asymptotics of achievable
coding rates. Researchers like Anade, Gorce, Mary, and Perlaza \cite{AGMP20},
Erseghe \cite{Erseghe16}, Moulin \cite{Moulin17}, Polyanskiy \cite{Polyanskiy10},
Tan and Tomamichel \cite{TT14}, and Yavas, Kostina, and Wigger \cite{YKE22}
have contributed to this area.

In the work by Huleihel, Salamatian, Merhav, and M\'{e}dard \cite{HSMM19},
the saddle-point approximation was applied to assess the log-volume
of a conditional Gaussian type class related to the Gaussian intersymbol
interference channel, with implications for mismatched universal decoding.
This addressed an open problem from \cite{me93}.

In \cite{me23b}, the saddle-point approximation was used to refine
the evaluation of the probability that a randomly selected codeword
would fall within a sphere of a specified radius from a given source
vector, based on a given distortion measure. The precise pre-exponential
factor allowed for the characterization of redundancy rates. In \cite{me11a},
the method was applied to lossless data compression in the context
of the set partitioning problem.

Lastly, in \cite[Section 4.7]{MM09}, M\'{e}zard and Montanari establish
a valuable link between the saddle-point method, Sanov's theorem,
and the method of types, providing further insights into the connections
between these powerful techniques.

\newpage{}

\section{The Type Class Enumeration Method \label{sec:TCE}}

\subsection{Introduction}

In the previous chapters, we considered probabilistic properties of
a single random vector, or a finite collection of vectors. We developed
a generalized version of \emph{the method of types} \cite{Csiszar98,csiszar2011information},
which can be succinctly summarized by two main properties: For a vector
of length $n$: (1) The effective number of types is sub-exponential
with $n$; (2) The size of a type class is $|{\cal T}_{n}(Q)|\doteq e^{nH(Q)}$,
where $H(Q)$ is the entropy (but can also be replaced by a differential
entropy). The main consequence of these two properties is that the
probability of observing a type $Q$ from a memoryless source with
distribution $P$ is $\Pr[\boldsymbol{X}\in{\cal T}_{n}(Q)]\doteq e^{-nD(Q||P)}$.
In this chapter, we ascend one hierarchical level, and consider analysis
of \emph{coding} problems, and specifically the problem of evaluating
the \emph{error exponent} in coded systems. Such problems involve
an \emph{exponential} number of random vectors, and so, additional
analytical tools are required. 

Starting from Shannon \cite{shannon1948mathematical}, the common
method of proving achievability results in information theory is via
\emph{random-coding} analysis, in which the error probability is averaged
over an ensemble of randomly selected codebooks. While the random-coding
argument was originally invoked to find the \emph{capacity} $C$ of
noisy channels \cite{CT06}, it was broadly adapted to other settings.
In this chapter, we will focus on error exponent analysis \cite{elias1955coding},
\cite[Chapters 7-9]{fano1961transmission}, \cite[Chapter 5]{gallager1968information},
\cite{csiszar1981graph,csiszar1977new}, which is a refined performance
measure of coded systems. The error exponent refers to the largest
exponential decay rate of the error probability of a sequence of codes
at increasing blocklength $n$, for a given rate $R$ below the capacity
$C$. Since the error probability of the optimal codebook can be \emph{upper
bounded} by the average of the error probability over an ensemble
of random codebooks, the error exponent can be \emph{lower bounded}
by the \emph{random-coding error exponent} --- the exponential decay
rate of the ensemble-average error probability. 

Moreover, the random-coding error exponent is interesting as \emph{a
paradigm on its own right}, since it is by now well-established that
random codes, or \emph{random-like} codes (e.g., turbo codes \cite{berrou1993near}
and low-density parity-check (LDPC) codes \cite{gallager1962low};
see \cite{richardson2008modern}) are highly efficient \cite{chung2001design}.
In fact, in some applications, the codebook is routinely redrawn at
random, for example, in order to preserve the security of the transmitted
information. So, when a communication system uses such a random code,
it is the random-coding error probability (or exponent) that is a
relevant measure to the long-term performance of the system, rather
than just serving as a lower bound to the best achievable exponent. 

The analysis of the random-coding error exponent has lead to the proposal
and usage of an ample of analytical bounding methods. We next outline
several of them, in order to contrast them later on with our type
of techniques. 

First, the error probability of the optimal ML decoder can be upper
bounded by the error probability of simpler, sub-optimal decoder.
For example, the error probability of the typicality decoder \cite[Chapter 7]{CT06}
decays to zero at all rates below capacity, just as the ML decoder.
So, analyzing the typicality decoder can be used to prove lower bounds
on the capacity. However, this decoder has poor performance in terms
of the error exponent. 

Second, as popularized by Gallager \cite{gallager1965simple}, Jelinek
\cite{jelinek1968probabilistic} and Forney \cite{forney1968exponential},
the use of convexity properties and \emph{Jensen-style} inequalities.
These include, for example, the inequality $\E[Z^{\rho}]\leq(\E[Z])^{\rho}$
for a non-negative RV and $0\leq\rho\leq1$, or the power distribution
inequality 
\begin{equation}
\left(\sum_{j}a_{j}\right)^{\rho}\leq\sum_{j}a_{j}^{\rho}\label{eq: power distribution inequality}
\end{equation}
(see \cite[Appendix 3A]{viterbi2009principles} for a comprehensive
list of such inequalities). 

Third, the use of \emph{Chernoff-style }bounds, in which an indicator
of an error event, based on likelihoods, is replaced by their ratio.
For example, a pairwise error event of an ML decoder over the channel
$W$ from $\boldsymbol{x}$ to $\boldsymbol{y}$ is upper bounded
as 
\begin{equation}
\I\left\{ W(\boldsymbol{y}|\boldsymbol{x}_{j})\geq W(\boldsymbol{y}|\boldsymbol{x}_{i})\right\} \leq\left[\frac{W(\boldsymbol{y}|\boldsymbol{x}_{j})}{W(\boldsymbol{y}|\boldsymbol{x}_{i})}\right]^{\lambda}
\end{equation}
for any $\lambda\geq0$. 

Fourth, \emph{refined union bounds}, in which the simple union bound
over events $\{{\cal A}_{j}\}$ is replaced by a quantity lower than
the sum of probabilities of each event. These bounds include, \emph{a
truncated union bound} 
\begin{equation}
\Pr\left[\bigcup_{j}{\cal A}_{j}\right]\leq\min\left[1,\sum_{j}\Pr[{\cal A}_{j}]\right],\label{eq: clipped union bound}
\end{equation}
\emph{a union bound with a power parameter} $0\leq\rho\leq1$ (also
known as \emph{Gallager's union bound} \cite[p. 136]{gallager1968information})
\begin{equation}
\Pr\left[\bigcup_{j}{\cal A}_{j}\right]\leq\left(\sum_{j}\Pr[{\cal A}_{j}]\right)^{\rho},\label{eq: Gallager union bound}
\end{equation}
or \emph{a union bound with intersection} of an event ${\cal G}$
\begin{equation}
\Pr\left[\bigcup_{j}{\cal A}_{j}\right]\leq\sum_{j}\Pr\left[{\cal A}_{j}\cap{\cal G}\right]+\Pr\left[{\cal G}^{c}\right],
\end{equation}
where ${\cal G}^{c}$ is the complement of ${\cal G}$. As an illustrative
example, such a union bound can be used to bound the probability of
an error event in channel coding, since this event is a union of the
events that one of the alternative codewords is decoded. The above
bounding methods then lead to tractable, computable, bounds on the
random-coding exponent, and other quantities of interest. 

Nonetheless, in typical channel coding problems, codebooks with a
positive coding rate $R$ have an exponential number of codewords
$e^{nR}$, and so, the analysis of the error probability involves
evaluation of the probability of a union of an \emph{exponential}
number of events. In some cases, it can be shown that a bound obtained
via these methods is actually tight. For example, in the simple case
of a point-to-point discrete memoryless channel (DMC), Gallager has
shown that its random-coding error exponent, obtained using \eqref{eq: Gallager union bound},
is tight, by \emph{lower }bounding the error probability \cite{gallager1973random}.
However, there is no \emph{general} claim that these bounding methods
lead to the \emph{exact} random-coding error exponent\emph{, }that
is, that the final result is the true exponential decay rate of the
expected error probability over the random ensemble of codebooks.
In fact, in various scenarios they are \emph{strictly loose}. 

In this chapter, we introduce the \emph{type class enumeration method
}(TCEM) of random codes, which is an original viable alternative or
complement to the aforementioned techniques. It is a principled method,
whose main virtue is that it preserves exponential tightness along
all steps of the derivation of the exponent. It is therefore guaranteed
to obtain the \emph{exact} exponent. The TCEM achieves that by refraining
from using the various bounding techniques mentioned above, and thus
avoiding the need to optimize over various parameters (which cannot
always be done in a closed-form), and leads to explicit expressions.
More often than not, it does so in a ``single-pass'', \emph{i.e.},
without separately lower and upper bounding the random-coding error
exponent. Consequently, ensemble-tight random-coding exponents can
be obtained in a multitude of coding problems. Moreover, as mentioned,
and as we shall survey, in coding problem that go beyond basic ones,
the error exponents obtained by the TCEM are oftentimes strictly larger
than those achieved using the above bounding techniques. 

For this chapter, we recall the usual notation convention for an equality
or an inequality in the exponential scale: For two positive sequences
$\{a_{n}\}$ and $\{b_{n}\}$, the notation $a_{n}\doteq b_{n}$ means
that $\lim_{n\to\infty}\frac{1}{n}\log\frac{a_{n}}{b_{n}}=0$, and
$a_{n}\dot{\leq}b_{n}$ means that $\lim_{n\to\infty}\frac{1}{n}\log\frac{a_{n}}{b_{n}}\leq0$,
and similarly for $a_{n}\dot{\geq}b_{n}$. Accordingly, $a_{n}\doteq1$
means that $a_{n}$ is sub-exponential, and $a_{n}\doteq e^{-n\infty}$
means that $a_{n}$ decays at a super-exponential rate (e.g., double-exponentially). 

The main idea of the TCEM is that each codeword can be categorized
according to a \emph{joint type} (empirical distribution) with an
additional length-$n$ vector, and that and the union bound is exponentially
tight for a union of the polynomial number of events. Indeed, for
$k_{n}$ events $\{{\cal E}_{i}\}_{i=1}^{k_{n}}$
\begin{equation}
\max_{1\leq m\leq k_{n}}\Pr\left[{\cal E}_{m}\right]\leq\Pr\left[\bigcup_{m=1}^{k_{n}}{\cal E}_{m}\right]\leq k_{n}\cdot\max_{1\leq m\leq k_{n}}\Pr\left[{\cal E}_{m}\right].\label{eq: tightness of the union bound for poly events}
\end{equation}
and so if $k_{n}\doteq1$ then
\begin{equation}
\Pr\left[\bigcup_{m=1}^{k_{n}}{\cal E}_{m}\right]\doteq\max_{1\leq m\leq k_{n}}\Pr\left[{\cal E}_{m}\right].
\end{equation}
Therefore, the analysis of a coding problem can be based on a \emph{type
class enumerator} (TCE), which counts the number of randomly selected
codewords in a properly defined type class. For illustration, one
may recall that for binary symmetric channels (BSCs), the \emph{distance
spectrum of a codebook}, namely, the number of pairs of codewords
at each of the $n+1$ possible Hamming distances, plays an important
role in determining its error probability (e.g., \cite[Chapter 2]{macwilliams1977theory}).
Indeed, a specific form of TCEs for BSCs was used by \cite{barg2002random}
to analyze various random-coding exponents. The TCEM can be thought
of as a considerable generalization of this fundamental idea. 

In the TCEM, the codebook is drawn at random, and consequently, the
TCEs are RVs. The random-coding error exponent thus depends on their
probabilistic and statistical properties, such as moments or tail
bounds. Each TCE is typically a binomial RV $N\sim\text{Binomial}(e^{nA},e^{-nB})$
(or a close variant of such variables), defined by $e^{nA}$ independent
trials for belonging to a type class, each with success probability
$e^{-nB}$. It exhibits an interesting \emph{phase transition} at
$A=B$: If the the number of trials dominates the success rate, $A>B$,
then the TCE is tightly concentrated around its exponentially large
expected value $e^{n(A-B)}$ (double-exponential concentration). We
refer to these as \emph{typically populated} types. Otherwise, if
$B>A$, then the TCE is typically zero, and the probability that it
is strictly positive is exponentially less than $e^{-n(B-A)}$. We
refer to these as \emph{typically empty} types. The transition between
these regimes is sharp, and is rooted in a statistical-mechanical
perspective on random coding. This perspective is based on an analogy
to Derrida's random energy model (REM) \cite{derrida1980random_B,derrida1980random_A,derrida1981random}
\cite[Chapters 5 and 6]{MM09}, which is a spin glass model with high
degree of disorder, and which is well known in the literature of statistical
physics of magnetic materials. The phase transition in the REM model
is analogous to the one exhibited for the TCEs, and we refer the reader
to \cite{merhav2008relations} and \cite[Chapter 6]{me09} for a thorough
exposition. 

The TCEM hence involves the following steps: (1) Expressing the error
probability (or other quantity of interest) using properly defined
TCEs. (2) Evaluating the necessary probabilistic and statistical properties
of the TCEs (moments or tail probabilities). (3) Plugging in these
properties in the expression for the error probability, and evaluating
the resulting expression. (4) Developing an efficient procedure to
compute the exponent. This last step is equally important, since in
some cases, the resulting expression for the exponent may appear involved
or challenging to compute. We show in Appendix \ref{sec:Computation-of-the}
how efficient methods can be developed. 

For simplicity of exposition, we focus in this chapter on DMCs, for
which the standard method of types \cite{Csiszar98,csiszar2011information}
is applicable. However, given the generalized method of types described
in Chapter \ref{sec: CMoT}, these ideas can be extended to other
channels, including Gaussian channels (which have continuous alphabets)
and channels with memory, without requiring a substantial modification. 

The outline of this chapter is as follows. For methodological reasons,
our first step will invoke the TCEM for problems in which error exponents
are already well-established, to wit, error exponents for DMCs (random-coding
\cite{csiszar1981graph,csiszar1977new} and expurgated \cite[Section V]{gallager1965simple}
\cite{jelinek1968evaluation}) and the correct decoding exponent for
rates above capacity \cite{arimoto1973converse,dueck1979reliability}.
This will exemplify the technique of the TCEM in a familiar setting,
and serve as a basis for the rest of the chapter. We will then derive
the basic statistical and probabilistic properties of TCEs, to wit,
tail probabilities and moments. Afterwards, we will demonstrate the
TCEM in more advanced settings, namely: (1) The error exponent of
superposition coding in a broadcast asymmetric DMC for the optimal
bin-index decoder. (2) The random-binning error exponent of distributed
compression \cite{slepian1973noiseless}. (3) The random-coding error
exponents of generalized decoders, such as Forney's erasure/list decoder
\cite{forney1968exponential} and a generalized version of the likelihood
decoder \cite{yassaee2013technique}. (4) The error exponent of the
typical random code \cite{barg2002random}. 

At the last section of this chapter, we will survey the wide applicability
of the TCEM, and its ability to provide exact random-coding exponents
in a multitude of information-theoretic problems: The problem could
be channel coding or source coding problem; the problem could involve
a single user and point-to-point channels, or multiple users operating
in a distributed manner over a network \cite{el2011network}; the
code could have a fixed length, be a convolutional/trellis code \cite{viterbi2009principles,johannesson2015fundamentals,ungerboeck1982channel},
or have variable encoding length (with feedback) \cite{burnashev1976data};
the decoder could be the optimal ML decoder, the universal maximum
mutual information (MMI) decoder \cite{goppa1975nonprobabilitistic},
a mismatched decoder, an erasure/list decoder \cite{forney1968exponential}
that is allowed to output an erasure or more than a single codeword,
a list decoder that outputs a list of possible codewords \cite{elias1957list,wozencraft1958list},
a bin-index decoder, which is the optimal ML decoder in which the
codeword is only known to belong to a bin; a likelihood decoder which
randomly decodes a message based on a posterior probability distribution
\cite{yassaee2013technique}; a joint detector-decoder that is required
to make a decision in addition to decoding the message \cite{wang2011error};
and more. Moreover, beyond the random-coding error exponent, other
exponents can also be derived using the TCEM, e.g., the error exponent
of the typical random code \cite{barg2002random} and large-deviations
from this typical code \cite{tamir2020large,truong2023concentration}. 

\subsection{The Type Class Enumeration Method for Basic Coding Problems \label{subsec:Type-Class-Enumerators-in-Coding--Basic}}

To obtain a quick glance on the underlying ideas, we first consider
the basic problems of the random-coding and expurgated exponents for
a DMC, and then the correct decoding exponent (for rates above capacity).
Along the way, we will introduce several useful techniques, such as
the summation--maximization equivalence, tail integration, and, later
on, exponential tightness of the union bound for pairwise independent
events. For the sake of convenience, we begin with a short background
of classic error exponents for DMCs.

\subsubsection{A Short Background: Error Exponents of DMCs}

Consider a DMC $W$ with input alphabet ${\cal X}$ and output alphabet
${\cal Y}$, and a codebook ${\cal C}_{n}=\{\boldsymbol{x}_{m}\}$
whose codewords $\boldsymbol{x}_{m}\in{\cal X}^{n}$ have blocklength
$n$, and it has rate $R$, that is $|{\cal C}_{n}|=e^{nR}$.\footnote{Throughout, we will ignore integer constraints on large quantities
such as $e^{nR}$ (which should be $\lceil e^{nR}\rceil$), since
these do not affect any of the analyses or the results.} The holy grail of error exponent analysis \cite[Chapter 10]{csiszar2011information},
\cite[Chapter 5]{gallager1968information}, is to find the maximum
achievable error exponent achieved at any rate $R$, also known as
the \emph{reliability function, $E^{*}(R)$}. This establishes the
existence of a sequence of codes $\{{\cal C}_{n}^{*}\}$ of rate $R$,
whose error probability decays with the maximal exponent\footnote{It is unclear if the following limit exist \cite[Exercise 10.7]{csiszar2011information},
and so we take the conservative definition of limit-superior.}
\begin{equation}
E^{*}(R)=\sup_{\{{\cal C}_{n}\}}\limsup_{n\to\infty}-\frac{1}{n}\ln P_{\mathsf{e}}({\cal C}_{n}^{*}),\label{eq: reliability function}
\end{equation}
where $P_{\mathsf{e}}({\cal C}_{n})$ is the error probability of
the codebook ${\cal C}_{n}$ (for a given, implicit, decoding rule).
As expected, pointing out a particular sequence of codes achieving
the reliability function is a formidable problem. The random-coding
argument shows that $E^{*}(R)$ is lower bounded by the exponent achieved
by random codes. Specifically, we consider a random ensemble in which
each codeword $\boldsymbol{X}_{m}\in{\cal X}^{n}$ is chosen randomly,
independent of all other codewords, and in identical way: In the IID
ensemble, each symbol of the codeword is drawn independently from
some distribution $P_{X}$, and in the fixed-composition ensemble,
each codeword is chosen uniformly at random from a type class ${\cal T}_{n}(P_{X})$.
While both ensembles can be analyzed using the TCEM, we will focus
on the latter since it is more common when invoking the method of
types, and since it typically leads to larger random-coding exponents.
The average error probability for a random codebook $\mathfrak{C}_{n}$
chosen from the ensemble will be denoted by $\overline{P}_{\mathsf{e}}\triangleq\E[P_{\mathsf{e}}(\mathfrak{C}_{n})]$.
For a given ensemble, the random-coding error exponent of rate $R$
is then given by 
\begin{equation}
E_{\text{rc}}(R)\triangleq\lim_{n\to\infty}-\frac{1}{n}\ln\E\left[P_{\mathsf{e}}(\mathfrak{C}_{n})\right],\label{eq: random coding definition}
\end{equation}
whenever the limit exists, for which it holds that $E^{*}(R)\geq E_{\text{rc}}(R)$. 

The random-coding error exponent was studied by two different schools.
First, an approach lead by Gallager \cite[Chapter 5]{gallager1968information},
which is based on analytical techniques such as refined union bounds,
and later, by Csisz{\'a}r, K{\"o}rner and Marton \cite{csiszar2011information,csiszar1981graph,csiszar1977new},
who developed and used the method of types \cite{Csiszar98} to this
problem. Since the TCEM is based on the method of types, we will next
describe the latter \cite[Chapter 10]{csiszar2011information}. For
a DMC $W$, and a fixed-composition input $P_{X}$, this random-coding
error exponent takes the form $E_{\text{rc}}(R)=\max_{P_{X}}E_{\text{rc}}(R,P_{X})$,
where, with a slight abuse of notation, 
\begin{equation}
E_{\text{rc}}(R,P_{X})\triangleq\min_{Q_{Y|X}}\left\{ D(Q_{Y|X}||W|P_{X})+\left[I(P_{X}\times Q_{Y|X})-R\right]_{+}\right\} .\label{eq: random coding exponent CKM}
\end{equation}
It was also shown that this exponent can be achieved using the MMI
decoder, and does not require the optimal ML decoder. In parallel,
it was proved that the \emph{sphere packing bound} \cite{fano1961transmission,blahut1974hypothesis,Haroutunian1968estimates,shannon1967lowerI}
$E_{\text{sp}}(R)\triangleq\max_{P_{X}}E_{\text{sp}}(R,P_{X})$, where
\begin{equation}
E_{\text{sp}}(R,P_{X})\triangleq\min_{Q_{Y|X}\colon\;I(P_{X}\times Q_{Y|X})\leq R}D(Q_{Y|X}||W|P_{X}),
\end{equation}
is an upper bound on the reliability function $E^{*}(R)\leq E_{\text{sp}}(R)$.
Remarkably, there exists a critical rate $R_{\text{cr}}$ such that
for any $R\geq R_{\text{cr}}$ it holds that 
\begin{equation}
E^{*}(R)=E_{\text{rc}}(R)=E_{\text{sp}}(R),
\end{equation}
and so at high rates, the reliability function is exactly known, and
random-coding is optimal. At low rates, $R<R_{\text{cr}}$, however,
the ensemble-average error probability may be highly affected by codes
with large error probability. This has lead to the idea of \emph{expurgating}
the ensemble from these codes, and to the development of the expurgated
exponent. The expurgated exponent $E_{\text{ex}}(R)$ is a lower bound
on the reliability function $E^{*}(R)\geq E_{\text{ex}}(R)$, and
improves on the random-coding error exponent at low rates. Let the
Bhattacharyya distance between $\boldsymbol{x},\tilde{\boldsymbol{x}}\in{\cal X}^{n}$
be defined by 
\begin{equation}
d_{\text{B}}(\boldsymbol{x},\tilde{\boldsymbol{x}})\triangleq-\ln\sum_{\boldsymbol{y}\in{\cal Y}^{n}}\sqrt{W(\boldsymbol{y}|\boldsymbol{x})\cdot W(\boldsymbol{y}|\tilde{\boldsymbol{x}})}.
\end{equation}
Since it only depends on the joint type, $Q_{X\tilde{X}}=\hat{Q}_{\boldsymbol{x},\tilde{\boldsymbol{x}}}$,
we also denote, with a slight abuse of notation, $d_{\text{B}}(Q_{X\tilde{X}})$
as the Bhattacharyya between some $(\boldsymbol{x},\tilde{\boldsymbol{x}})\in{\cal T}_{n}(Q_{X\tilde{X}})$.
The expurgated exponent \cite{Csiszar98,csiszar2011information,csiszar1981graph,csiszar1977new}
is given by $E_{\text{ex}}(R)=\max_{P_{X}}E_{\text{ex}}(R,P_{X})$
where 
\begin{equation}
E_{\text{ex}}(R,P_{X})\triangleq\min_{Q_{X\tilde{X}}\colon\;Q_{X}=Q_{\tilde{X}}=P_{X},\;I(Q_{X\tilde{X}})\leq R}\left[d_{\text{B}}(Q_{X\tilde{X}})+I(Q_{X\tilde{X}})\right]-R.\label{eq: expurgated exponent}
\end{equation}
Also remarkably, the celebrated result of Shannon, Gallager and Berlekamp
\cite{shannon1967lowerII} showed that it is tight at zero rate $E^{*}(0)=E_{\text{ex}}(0)$,
and also used that to derive an improved upper bound at intermediate
rates, known as the \emph{straight line bound}.

\subsubsection{Error Exponents of a DMC via Type Class Enumeration \label{subsec:random-coding error exponent}}

We next show how to derive the random-coding error exponent $E_{\text{rc}}(R)$
via the TCEM. As usual, we fix the transmitted codeword $\boldsymbol{X}_{1}=\boldsymbol{x}$
and the output vector $\boldsymbol{y}$, and then write the probability
that one of the $e^{nR}-1$ (random) competing codewords in $\mathfrak{C}_{n}\backslash\{\boldsymbol{X}_{1}\}$
is decoded instead of $\boldsymbol{X}_{1}$. This amounts to
\begin{equation}
\overline{P}_{\mathsf{e}}=\sum_{\boldsymbol{x}\in{\cal X}^{n}}\sum_{\boldsymbol{y}\in{\cal Y}^{n}}\Pr\left[\boldsymbol{X}_{1}=\boldsymbol{x}\right]W(\boldsymbol{y}|\boldsymbol{x})\cdot\Pr\left[\bigcup_{m=2}^{e^{nR}}\text{\ensuremath{\boldsymbol{X}}}_{m}\text{ has higher score than }\text{\ensuremath{\boldsymbol{X}}}_{1}=\boldsymbol{x}\right].\label{eq: random coding exponent initial}
\end{equation}
The next conceivable step is to further bound the inner probability
by a union bound, and as said, while a naive union bound fails, the
clipped union bound \eqref{eq: clipped union bound} or Gallager's
union bound \eqref{eq: Gallager union bound} both lead to the exact
random-coding error exponent in this basic setting. However, the TCEM
proceeds differently.

Let us denote by $Q_{XY}$ a generic joint type of $(\boldsymbol{x},\boldsymbol{y})$,
where $Q_{X}=P_{X}$ matches the type of the fixed-composition ensemble,
and where for brevity, henceforth, we will often make this implicit.
We further consider the class of $\alpha$-decoders, which decide
using a score function $\alpha(Q_{XY})$ that depends only on the
joint type of the output vector $\boldsymbol{y}$ and the candidate
codeword. Specifically, if $\hat{Q}_{\boldsymbol{x},\boldsymbol{y}}$
is the joint type of $(\boldsymbol{x},\boldsymbol{y})$ then the decoded
codeword is a maximizer of $\alpha(\hat{Q}_{\boldsymbol{x}_{j},\boldsymbol{y}})$.
Let the expected log-likelihood of a joint type $Q_{XY}$ be 
\begin{equation}
f(Q_{XY})\triangleq\E_{Q}\left[\ln W(Y|X)\right].
\end{equation}
It can be easily noted that the MMI decoder, $\alpha(Q_{XY})=I(Q_{XY})$,
and the ML decoder, $\alpha(Q_{XY})=f(Q_{XY})$, are both $\alpha$-decoders.
We now introduce a proper TCE. 
\begin{defn}[TCE for random-coding exponent]
\label{def: TCE random coding} For a codebook ${\cal C}_{n}$, an
output vector $\boldsymbol{y}$, and a joint type $Q_{XY}$ such that
$\hat{Q}_{\boldsymbol{y}}=Q_{Y}$, let
\begin{equation}
N_{\boldsymbol{y}}(Q_{XY},{\cal C}_{n})\triangleq\left|\left\{ m>1\colon(\boldsymbol{X}_{m},\boldsymbol{y})\in{\cal T}_{n}(Q_{XY})\right\} \right|.\label{eq: enumerator for simple random coding}
\end{equation}
The TCE $N_{\boldsymbol{y}}(Q_{XY},{\cal C}_{n})$ counts the number
of incorrect codewords in ${\cal C}_{n}$ whose joint type with $\boldsymbol{y}$
is $Q_{XY}$. By the method of types, when $\boldsymbol{X}_{m}\sim\text{Uniform}[{\cal T}_{n}(Q_{X})]$
then 
\begin{equation}
\Pr\left[(\boldsymbol{X}_{m},\boldsymbol{y})\in{\cal T}_{n}(Q_{XY})\right]=k_{n}\cdot e^{-nI(Q_{XY})}
\end{equation}
for some $k_{n}\doteq1$. So, for a random codebook $\mathfrak{C}_{n}=\{\boldsymbol{X}_{m}\}$,
\begin{equation}
N_{\boldsymbol{y}}(Q_{XY},\mathfrak{C}_{n})=\sum_{m=2}^{e^{nR}}\I\left\{ (\boldsymbol{X}_{m},\boldsymbol{y})\in{\cal T}_{n}(Q_{XY})\right\} \sim\text{Binomial}\left(e^{nR}-1,k_{n}\cdot e^{-nI(Q_{XY})}\right).\label{eq: simple random coding enumerator}
\end{equation}
Since any $\boldsymbol{X}_{m}$ has a unique joint type with $\boldsymbol{y}$,
then viewed as a collection of TCEs, it holds that 
\begin{equation}
\{N_{\boldsymbol{y}}(Q_{XY})\}_{Q_{XY}}\sim\text{Multinomial}\left(e^{nR},\{p(Q_{XY})\}_{Q_{XY}}\right)
\end{equation}
 where $p(Q_{XY})\doteq e^{-nI(Q_{XY})}.$
\end{defn}
Since the probability measure of $N_{\boldsymbol{y}}(Q_{XY},\mathfrak{C}_{n})$
only depends on $\boldsymbol{y}$ through its type, for brevity, we
will omit both $\mathfrak{C}_{n}$ and $\boldsymbol{y}$ from the
notation of TCEs (with a slight abuse of notation). We then have 
\begin{equation}
\overline{P}_{\mathsf{e}}=\sum_{Q_{XY}}\Pr\left[(\boldsymbol{X}_{1},\boldsymbol{y})\in{\cal T}_{n}(Q_{XY})\right]\cdot\Pr\left[\bigcup_{\tilde{Q}_{XY}\colon Q_{Y}=\tilde{Q}_{Y},\;\alpha(\tilde{Q}_{XY})\geq\alpha(Q_{XY})}\I\left\{ N(\tilde{Q}_{XY})\geq1\right\} \right].\label{eq: random coding exponent first}
\end{equation}
The substantial difference between this bound and \eqref{eq: random coding exponent initial},
is that its inner probability is a union over a polynomial number
of types, rather than an exponential number of codewords. For such
a union of polynomial number of events, even the regular union bound
is exponentially tight \eqref{eq: tightness of the union bound for poly events},
and therefore
\begin{align}
\overline{P}_{\mathsf{e}} & \doteq\max_{Q_{XY}}\max_{\tilde{Q}_{XY}\colon Q_{Y}=\tilde{Q}_{Y},\;\alpha(\tilde{Q}_{XY})\geq\alpha(Q_{XY})}\Pr\left[(\boldsymbol{X}_{1},\boldsymbol{y})\in{\cal T}_{n}(Q_{XY})\right]\cdot\Pr\left[N(\tilde{Q}_{XY})\geq1\right]\nonumber \\
 & \doteq\max_{Q_{XY}}\max_{\tilde{Q}_{XY}\colon Q_{Y}=\tilde{Q}_{Y},\;\alpha(\tilde{Q}_{XY})\geq\alpha(Q_{XY})}\exp\left[-n\cdot D(Q_{XY}||P_{X}\times W)\right]\cdot\Pr\left[N(\tilde{Q}_{XY})\geq1\right].\label{eq: random coding exponent second}
\end{align}

In the next section, we will derive various properties of TCEs, and
specifically, tight tail bounds on $\Pr[N(\tilde{Q}_{XY})\geq1]$.
After inserting them back to \eqref{eq: random coding exponent second}
we will obtain 
\begin{equation}
\overline{P}_{\mathsf{e}}\doteq\exp\left[-n\cdot E_{\text{rc},\alpha}(R)\right],\label{eq: random coding exponent third}
\end{equation}
where 
\begin{equation}
E_{\text{rc},\alpha}(R,P_{X})\triangleq\min_{Q_{Y|X},\tilde{Q}_{Y|X}}D(Q_{Y|X}||W|P_{X})+\left[I(P_{X}\times\tilde{Q}_{Y|X})-R\right]_{+},\label{eq: random coding exponent expression}
\end{equation}
for which the minimization is over the set 
\begin{equation}
\left\{ Q_{Y|X},\;\tilde{Q}_{Y|X}\colon(P_{X}\times Q_{Y|X})_{Y}=(P_{X}\times\tilde{Q}_{Y|X})_{Y},\;\alpha(P_{X}\times\tilde{Q}_{Y|X})\geq\alpha(P_{X}\times Q_{Y|X})\right\} .
\end{equation}
This recovers a similar bound from \cite{csiszar1981graph} obtained
in a different way. For example, if $\alpha(Q_{XY})$ is the MMI rule,
the input the minimization is over $\{I(P_{X}\times\tilde{Q}_{Y|X})\geq I(P_{X}\times Q_{Y|X})\}$.
We recover the \emph{random-coding error exponent }\eqref{eq: random coding exponent CKM}.
If $\alpha(Q_{XY})=f(Q_{XY})\triangleq\E_{Q}[\ln W(Y|X)]$ is the
ML rule, then we achieve the same exponent. Indeed, since the ML is
the optimal decoder in terms of error probability, its error probability
can only be lower. On the other hand, $\tilde{Q}_{XY}=Q_{XY}$ belongs
to the set of inner minimization, and so the exponent cannot be larger
(see \cite[Proof of Lemma 4]{csiszar1981graph} for a direct proof,
which does not utilize the optimality of the ML rule). 

For general decoding scores, the random coding $E_{\text{rc},\alpha}(R,P_{X})$
is lower than the standard random-coding error exponent, and on the
face of it, is difficult to compute. Indeed, the clipping operation
is the result of the phase transition of the TCE at $R=I(\tilde{Q}_{XY})$.
This leads to an exponent expression which, in general, involves two
optimization problems, one for $I(\tilde{Q}_{XY})\leq R$ and the
other one for $I(\tilde{Q}_{XY})>R$. This is problematic since the
constraint $I(\tilde{Q}_{XY})>R$ is not convex, and so the resulting
minimization problem is not a convex optimization problem, even if
the constraint $\{\alpha(P_{X}\times\tilde{Q}_{Y|X})\geq\alpha(P_{X}\times Q_{Y|X})\}$
is convex (which occurs when $\alpha(Q_{XY})$ is linear in $Q_{XY}$,
as for the ML decoder). Nonetheless, in Appendix \ref{sec:Computation-of-the}
we show a method to efficiently compute the exponent that only requires
solving convex optimization problems (assuming $\alpha(Q)$ is linear).
As we have seen, the phase transition ($I>R$ or $I<R$) holds generally
for TCEs, and so this issue occurs for almost any exponent derived
by this method, which may take much more complicated form. Nonetheless,
methods similar to the one described in Appendix \ref{sec:Computation-of-the}
can usually be developed to efficiently compute the exponent, even
for such complicated scenarios. 

We next move on to shortly discuss the expurgated exponent. Assuming
that the ML decoder is used, the pairwise error probability for two
codewords $\boldsymbol{x}$ and $\tilde{\boldsymbol{x}}$ is upper
bounded by the \emph{Bhattacharyya bound} (e.g., \cite[Problem 10.20]{csiszar2011information})
as
\begin{equation}
P_{\mathsf{e}}(\boldsymbol{x},\tilde{\boldsymbol{x}})\triangleq\Pr\left[W(\boldsymbol{Y}|\tilde{\boldsymbol{x}})\geq W(\boldsymbol{Y}|\boldsymbol{x})\right]\leq\exp\left[-n\cdot d_{\text{B}}(\boldsymbol{x},\tilde{\boldsymbol{x}})\right].\label{eq: pairwise Bhattacharyya bound}
\end{equation}
Thus, for a given code ${\cal C}_{n}$, the regular union bound implies
that 
\begin{equation}
P_{\mathsf{e}}({\cal C}_{n})\leq\frac{1}{e^{nR}}\sum_{m=1}^{e^{nR}}\sum_{\tilde{m}=1}^{e^{nR}}\I\{\tilde{m}\neq m\}\cdot\exp\left[-n\cdot d_{\text{B}}(\boldsymbol{x}_{m},\boldsymbol{x}_{\tilde{m}})\right].
\end{equation}
We next introduce a proper TCE for the expurgated exponent.
\begin{defn}[TCE for expurgated exponent]
\label{def: TCE expurgated} For a joint type $Q_{X\tilde{X}}$,
a codebook ${\cal C}_{n}$, and a codeword index $m=1,\ldots,e^{nR}$,
let
\begin{equation}
\overline{N}_{m}(Q_{X\tilde{X}},\mathfrak{C}_{n})\triangleq\left|\left\{ \tilde{m}\colon m\neq\tilde{m},\;(\boldsymbol{x}_{m},\boldsymbol{x}_{\tilde{m}})\in{\cal T}_{n}(Q_{X\tilde{X}})\right\} \right|,\label{eq: exponent for expurgated}
\end{equation}
count the number of codewords in the codebook ${\cal C}_{n}$ which
have a joint type $Q_{X\tilde{X}}$ with $\boldsymbol{x}_{m}$. By
the method of types, when $\boldsymbol{X}_{\tilde{m}}\sim\text{Uniform}[{\cal T}_{n}(Q_{X})]$
then 
\begin{equation}
\Pr\left[(\boldsymbol{X}_{\tilde{m}},\boldsymbol{x}_{m})\in{\cal T}_{n}(Q_{X\tilde{X}})\right]=k_{n}\cdot e^{-nI(Q_{X\tilde{X}})}
\end{equation}
for some $k_{n}\doteq1$. So, for a random codebook $\mathfrak{C}_{n}=\{\boldsymbol{X}_{m}\}$
it holds that 
\begin{multline}
\overline{N}_{m}(Q_{X\tilde{X}},\mathfrak{C}_{n})\triangleq\sum_{\tilde{m}=1}^{e^{nR}}\I\{\tilde{m}\neq m\}\cdot\I\left\{ (\boldsymbol{X}_{\tilde{m}},\boldsymbol{X}_{m})\in{\cal T}_{n}(Q_{X\tilde{X}})\right\} \\
\sim\text{Binomial}\left(e^{nR}-1,k_{n}\cdot e^{-nI(Q_{X\tilde{X}})}\right).\label{eq: expurgated enumerator}
\end{multline}
It should be noted that $\{\overline{N}_{m}(Q_{X\tilde{X}},\mathfrak{C}_{n})\}_{m=1}^{e^{nR}}$
is a collection of an exponential number of \emph{dependent} RVs. 
\end{defn}
As for the TCE for random-coding, we will omit $\mathfrak{C}_{n}$
from the notation of TCEs (with a slight abuse of notation). Evidently,
the above upper bound can be expressed using the TCEs as 
\begin{equation}
P_{\mathsf{e}}({\cal C})\leq\frac{1}{e^{nR}}\sum_{m=1}^{e^{nR}}\sum_{Q_{X\tilde{X}}}\overline{N}_{m}(Q_{X\tilde{X}})\cdot\exp\left[-n\cdot d_{\text{B}}(Q_{X\tilde{X}})\right].\label{eq: expurgated exponent first}
\end{equation}
In Appendix \ref{sec:The-Expurgated-Exponent}, we show how the bound
\eqref{eq: expurgated exponent first} and the properties of $\{\overline{N}_{m}(Q_{X\tilde{X}})\}_{m=1}^{e^{nR}}$
derived in the next section can be used to derive the classic expurgated
exponent \eqref{eq: expurgated exponent}. 

\subsubsection{The Correct Decoding Exponent of a DMC \label{subsec:The-Correct-Decoding}}

One of the first demonstrations of the TCEM was for the \emph{correct}
decoding exponent of a DMC for rates above capacity \cite{merhav2008relations}.
Following Arimoto \cite{arimoto1973converse}, the correct decoding
error probability begins with the identity
\begin{align}
P_{\mathsf{c}}({\cal C}_{n}) & =\frac{1}{e^{nR}}\sum_{\boldsymbol{y}\in{\cal Y}^{n}}\max_{m}W(\boldsymbol{y}|\boldsymbol{X}_{m})\nonumber \\
 & =\lim_{\beta\to\infty}\frac{1}{e^{nR}}\sum_{\boldsymbol{y}\in{\cal Y}^{n}}\left[\sum_{m}W^{\beta}(\boldsymbol{y}|\boldsymbol{X}_{m})\right]^{1/\beta}.\label{eq: Arimoto's identity}
\end{align}
Let us consider the TCE
\begin{equation}
N_{\boldsymbol{y}}(Q_{XY})\triangleq\left|\left\{ m\geq1\colon(\boldsymbol{x}_{m},\boldsymbol{y})\in{\cal T}_{n}(Q_{XY})\right\} \right|,
\end{equation}
which is only slightly different from the random-coding TCE of Definition
\ref{def: TCE random coding}, and so, we abuse the notation and denote
them similarly. Assuming an ensemble of random codebooks, we next
evaluate the ensemble-average of the correct decoding probability.
To this end, we fix a finite $\beta$ and $\boldsymbol{y}$, and write
the ensemble average using TCEs as
\begin{align}
\E\left\{ \left[\sum_{m}W^{\beta}(\boldsymbol{y}|\boldsymbol{X}_{m})\right]^{1/\beta}\right\}  & =\E\left\{ \left[\sum_{Q_{XY}}N(Q_{XY})\cdot e^{n\beta f(Q_{XY})}\right]^{1/\beta}\right\} \nonumber \\
 & \doteq\E\left\{ \left[\max_{Q_{XY}}N(Q_{XY})\cdot e^{n\beta f(Q_{XY})}\right]^{1/\beta}\right\} \nonumber \\
 & =\E\left\{ \left[\max_{Q_{XY}}N^{1/\beta}(Q_{XY})\cdot e^{nf(Q_{XY})}\right]\right\} \nonumber \\
 & \doteq\E\left[\sum_{Q_{XY}}N^{1/\beta}(Q_{XY})\cdot e^{nf(Q_{XY})}\right]\nonumber \\
 & =\sum_{Q_{XY}}\E\left[N^{1/\beta}(Q_{XY})\right]\cdot e^{nf(Q_{XY})},\label{eq: max-to-sum-equivalence}
\end{align}
exploiting, \emph{twice}, the fact that the number of types is polynomial
in $n$ to interchange a summation with a maximum in both directions.
We refer to this as the \emph{summation--maximization equivalence},
which is frequently used to manipulate probabilities to a form that
allows for a direct substitution of TCE moments. As can be seen, \eqref{eq: max-to-sum-equivalence}
requires evaluating the fractional $1/\beta$ moment of the TCE $N(Q_{XY})$.
Next, since $|{\cal T}_{n}(Q_{Y})|\doteq e^{nH(Q_{Y})}$, we obtain
\begin{equation}
\overline{P}_{\mathsf{c}}\triangleq\E\left[P_{\mathsf{c}}(\mathfrak{C}_{n})\right]\doteq\lim_{\beta\to\infty}\sum_{Q_{XY}}\E\left[N^{1/\beta}(Q_{XY})\right]\cdot e^{n[f(Q_{XY})+H(Q_{Y})-R]}.\label{eq: correct decoding first}
\end{equation}
In the next section we will evaluate the moments $\E[N^{1/\beta}(Q_{XY})]$,
and using this results 
\begin{equation}
\overline{P}_{\mathsf{c}}\doteq\exp\left[-n\cdot E_{\text{c}}(R,P_{X})\right]
\end{equation}
where 
\begin{equation}
E_{\text{c}}(R,P_{X})\triangleq\min\left\{ E_{-}(R,P_{X}),E_{+}(R,P_{X})\right\} ,
\end{equation}
with 
\begin{equation}
E_{-}(R,P_{X})\triangleq\min_{Q_{XY}\colon\;I(Q_{XY})>R}\left[I(Q_{XY})-f(Q_{XY})-H(Q_{Y})\right]
\end{equation}
as well as
\begin{align}
E_{+}(R,P_{X}) & \triangleq\lim_{\beta\to\infty}\min_{Q_{XY}\colon\;I(Q_{XY})\leq R}\left[\frac{1}{\beta}I(Q_{XY})-\frac{1}{\beta}R+R-f(Q_{XY})-H(Q_{Y})\right]\nonumber \\
 & =\lim_{\beta\to\infty}\min_{Q_{XY}\colon\;I(Q_{XY})\leq R}\left[R-f(Q_{XY})-H(Q_{Y})\right].
\end{align}
Therefore, 
\begin{align}
E_{\text{c}}(R,P_{X}) & =\min_{Q_{XY}}\left[\max\{R,I(Q_{XY})\}-f(Q_{XY})-H(Q_{Y})\right]\nonumber \\
 & =\min_{Q_{XY}}\left\{ D(Q_{XY}||P_{X}\times W)+\left[R-I(Q_{XY})\right]_{+}\right\} ,
\end{align}
where the equality uses the identity $-f(Q_{XY})=D(Q||P_{X}\times W)-I(Q_{XY})+H(Q_{Y})$.
We also remark that in all the above expressions, $Q_{X}=P_{X}$ is
implicitly assumed. 

This bound recovers the K{\"o}rner--Dueck exponent \cite{dueck1979reliability},
which is known to be optimal (after minimizing over the input distribution
$P_{X}$). In \cite[Chapter 6]{me09} this example was studied in
detail for a BSC, and compared with Arimoto's approach in \cite{arimoto1973converse}.
Arimoto started as in \eqref{eq: Arimoto's identity}, but continued
by upper bounding these moments using Jensen's inequality, which interchanges
between the expectation operator and the $1/\beta$-power. As was
shown in \cite[Sec. 3]{merhav2008relations}, for a BSC with crossover
probability $p$, the correct-decoding exponent is 
\begin{align}
\overline{P}_{\mathsf{c}} & \doteq\exp\left[-n\cdot D(\delta_{\text{GV}}(R)||p)\right]\nonumber \\
 & =\exp\left[-n\cdot\left(\delta_{\text{GV}}(R)\ln\frac{1}{p}+\left(1-\delta_{\text{GV}}(R)\right)\ln\frac{1}{1-p}\right)-H(\delta_{\text{GV}}(R))\right]\label{eq: exact correct decoding exponent BSC}
\end{align}
where $\delta_{\text{GV}}(R)$ is the (smaller) solution to $\ln2-H(\delta)=R$
(recall that $H(q)\triangleq-q\ln q-(1-q)\ln(1-q)$ is the binary
entropy function). In comparison, if one takes $\beta\to\infty$,
then after using Jensen's inequality, the following bound is obtained:

\begin{equation}
\overline{P}_{\mathsf{c}}\dot{\leq\exp\left[-n\cdot\left(\min\left\{ \ln\frac{1}{p},\ln\frac{1}{1-p}\right\} \right)-H(\delta_{\text{GV}}(R))\right]}.\label{eq: Jensen correct decoding exponent BSC}
\end{equation}
Evidently, the exponent in \eqref{eq: Jensen correct decoding exponent BSC}
is \emph{strictly} smaller than the exact exponent in \eqref{eq: exact correct decoding exponent BSC}.
After maximizing over $\beta$, as well as minimizing over the input
distribution $P_{X}$, Arimoto's exponent matches that of K{\"o}rner--Dueck
\cite{dueck1979reliability}. However, this is not guaranteed in advance,
and also requires optimization over $\beta$. Indeed, in more complicated
settings, the optimization over parameters (such as $\beta$) for
derivations that are based on Jensen inequality cannot be performed
analytically, and even if so, they lead to strictly sub-optimal bounds. 

\subsection{Probabilistic and Statistical Properties of Type Class Enumerators
\label{subsec:Probabilistic-Properties-of}}

In the previous section, we have shown how to analyze basic coded
systems via the TCEM. In this section, we turn to analyze the probabilistic
and statistical properties of TCEs. Motivated by the discussion up
until now, we let $n$ be the blocklength, and let $A,B>0$ be two
constants. We will be interested in the properties of $N\sim\text{Binomial}(k_{n}'\cdot e^{nA},k_{n}''\cdot e^{-nB})$,
where $k_{n}'\cdot e^{nA}$ is the number of trials, and $k_{n}''\cdot e^{-nB}$
is the success probability, and where $k_{n}'\doteq k_{n}''\doteq1$,
and specifically, in tail probabilities and moments. In order to avoid
the polynomial pre-factors, we will analyze in what follows open intervals
of $A$ and $B$, for which it can be assumed that $k_{n}'=k_{n}''=1$.
However, as the exponents of the tail probabilities and moments are
continuous functions the results also hold for closed intervals. As
we shall next see, such binomial RVs experience a phase transition
at $B=A$, and therefore we will separate the analysis to the cases
of $B<A$ and $B>A$. As a side note, for $A=B$, the distribution
of such a binomial RV tends to that of a Poisson. 

We begin with the tail probabilities of $N$. 
\begin{thm}
\label{thm: tail probabilities of N}Assume that $N\sim\text{\emph{Binomial}}(e^{nA},e^{-nB})$
and $\lambda\in\reals$. Then, the upper tail is
\begin{equation}
\lim_{n\to\infty}-\frac{1}{n}\ln\Pr\left[N>e^{n\lambda}\right]=\begin{cases}
\left[B-A\right]_{+}, & \left[A-B\right]_{+}\geq\lambda\\
\infty, & \text{elsewhere}
\end{cases},\label{eq: exponent of upper tail}
\end{equation}
and the lower tail is
\begin{equation}
\lim_{n\to\infty}-\frac{1}{n}\ln\Pr\left[N<e^{n\lambda}\right]=\begin{cases}
0, & A-B<\lambda\\
\infty, & A-B>\lambda
\end{cases}.\label{eq: left tail populated type}
\end{equation}
\end{thm}
The proof of Theorem \ref{thm: tail probabilities of N}, as well
as all other theorems in this section, is deferred to Appendix \ref{sec:Proofs-for-TCE-properties}.

We continue with the moments of the $N$. While the first moment (expected
value) of $N$ is trivially given by $\E[N]=e^{n(A-B)}$, error exponent
analysis typically requires to evaluate general moments, which are
possibly fractional.
\begin{thm}
\label{thm: The moments of N}Assume that $N\sim\text{\emph{Binomial}}(e^{nA},e^{-nB})$.
Then, for $s>0$
\begin{equation}
\E\left[N^{s}\right]\doteq\begin{cases}
e^{n(A-B)s}, & A>B\\
e^{-n(B-A)}, & A<B
\end{cases}.\label{eq: moments of enumerators}
\end{equation}
\end{thm}
Importantly, for $A<B$ the moment is asymptotically independent of
$s>0$. 

In various advanced settings, the analysis of the TCEM also requires
to evaluate probabilistic and statistical properties of a \emph{pair}
of dependent TCEs, or even a family $\{N_{j}\}_{j=1}^{k_{n}}$ of
sub-exponential number $k_{n}\doteq1$ of TCEs, which are possibly
\emph{dependent}. For example, let $(U_{1},U_{2})$ be a pair of dependent
Bernoulli RVs so that $\Pr[U_{j}=1]=e^{-nB_{j}}$ for $j=1,2$. Typically,
$U_{j}$ are indicators for disjoint events, e.g., $U_{1}$ is the
event in which a random vector belongs to some type class, and $U_{2}$
belongs to a different type class. Thus, only one at most of the $U_{j}$
is $1$. Now, assume that we draw $e^{nA}$ IID RVs from the distribution
of $(U_{1},U_{2})$, and let $N_{j}$ denote the corresponding number
of successes for $U_{j}$, for $j=1,2$. While strictly speaking the
TCEs are dependent RVs, we next show they are asymptotically independent
in the regime we consider. Indeed, let us condition on the event that
$N_{1}=e^{n\nu}$ for some $\nu\in[0,A)$. Then, $N_{2}|N_{1}=e^{n\nu}\sim\text{Binomial}(e^{nA}-e^{n\nu},\Pr[U_{1}=1|U_{2}=0])$.
Evidently, the number of trials is $e^{nA}-e^{n\nu}\sim[1-o(1)]\cdot e^{nA}$
and the success probability is
\begin{align}
\Pr\left[U_{1}=1|U_{2}=0\right] & =\frac{\Pr\left[U_{1}=1,U_{2}=0\right]}{\Pr\left[U_{2}=0\right]}\nonumber \\
 & =\frac{\Pr\left[U_{1}=1\right]}{1-e^{-nB_{2}}}\nonumber \\
 & \sim e^{-nB_{1}}.\label{eq: conditional success probability}
\end{align}
So, up to factors that tend to $1$, the parameters of the conditional
binomial $N_{1}$ are exactly as those of the unconditional binomial.
It is easy to generalize the above argument to a sub-exponential number
of TCEs, that is, $\{N_{j}\}_{j=1}^{k_{n}}$ where $k_{n}\doteq1$,
thus showing that they are asymptotically independent. 

Indeed, if we consider the full set of TCEs, \emph{i.e.}, $\{N(Q)\}$
of all possible types, for which it must hold that $\sum N(Q)=e^{nA}$,
then each trial is successful for exactly one of the types. Thus,
the joint distribution of $\{N(Q)\}\sim\text{Multinomial}(e^{nA},\{p_{Q}\}_{Q})$,
where $p_{Q}\doteq e^{-nB(Q)}$. It is well-known that for a large
number of trials, the multinomial distribution tends to an \emph{independent}
Poisson distribution \cite[Theorem 5.6]{mitzenmacher2017probability}.
The joint distribution of TCEs in the context of a superposition codebook
(see Section \ref{subsec:Type-Class-Enumerators-in-Coding-advanced})
were analyzed in the arXiv version of \cite[Appendix D]{weinberger2019reliability}.
As we will also see in the setting of superposition coding, it is
sometimes required to analyze the probability of an intersection of
upper tail events of TCEs $\{N_{j}\}_{j=1}^{k_{n}}$, when $k_{n}\doteq1$
and where $N_{j}\sim\text{Binomial}(e^{nA_{j}},e^{-nB_{j}})$. This
is addressed by the following theorem. 
\begin{thm}
\label{thm: Intersection of tail events}Assume that $N_{j}\sim\text{\emph{Binomial}}(e^{nA_{j}},e^{-nB_{j}})$
for $j=1,\ldots,k_{n}$ and $k_{n}\doteq1$. Assume that $\lambda\in\reals$,
and $\lambda\neq A_{j}-B_{j}$ for all $j=1,\ldots,k_{n}$. Then,
\begin{equation}
\Pr\left[\bigcap_{j=1}^{k_{n}}\left\{ N_{j}\leq e^{n\lambda}\right\} \right]\doteq\I\left\{ \min_{1\leq j\leq k_{n}}\left\{ B_{j}-A_{j}+[\lambda]_{+}\right\} >0\right\} .
\end{equation}
\end{thm}

\subsection{Type Class Enumerators in Advanced Coding Problems \label{subsec:Type-Class-Enumerators-in-Coding-advanced}}

In this section, we demonstrate how the TCEM can be used in various
advanced coding problems. We show how the error probability in these
settings can be expressed via properly defined TCEs, which share similar
properties to the TCE analyzed in the previous section. For brevity,
we will not state here the resulting exponents -- they can be found
in the cited references -- and they provide the exact exponent for
the random ensemble of interest. Moreover, oftentimes the exponents
of the TCEM are also the best possible known, and are strictly better
than exponent bounds obtained via classic bounding techniques.

\subsubsection{Superposition Coding }

We begin with the asymmetric broadcast channel (or a broadcast channel
with a degraded message set) \cite{bergmans1973random,cover1972broadcast,gallager1974capacity,korner1977general,korner1980universally},
which is a prototypical example for a multiuser channel \cite{el2011network}.
This setting introduces new aspects, for which the TCEM is especially
useful in deriving exact random-coding error exponents. We focus on
a simple version of this setting, in which a single transmitter wishes
to communicate different messages to two receivers with different
channels, and so possibly different point-to-point capacities. The
first channel is thus referred to as the \emph{strong user} channel,
and the second as the \emph{weak user} channel\emph{. }We denote the
strong user (resp. weak user) channel by $W_{y}$ (resp. $W_{z}$),
which is from the input alphabet is ${\cal X}$ to the output alphabet
${\cal Y}$ (resp. ${\cal Z}$). 

Rather than drawing a regular random code for this channel, Cover
\cite{cover1972broadcast} and then Bergman \cite{bergmans1973random},
proposed to use \emph{superposition coding}, or an \emph{hierarchical
codebook}. In this coding method, the rate is split as $R=R_{z}+R_{y}$,
and the message is thus determined by two indices. In the random coding
regime, the codebook is constructed as follows. An auxiliary alphabet
${\cal U}$ is chosen along with a joint input type $P_{UX}$. Then,
\emph{$e^{nR_{z}}$ cloud centers} $\tilde{\mathfrak{C}}_{n}=\{\boldsymbol{U}_{i}\}_{i=1}^{e^{nR_{z}}}$
are drawn from the fixed-composition ensemble of input type $P_{U}$.
For each cloud center, a sub-codebook $\mathfrak{C}_{n}(i)=\{\boldsymbol{X}_{ij}\}_{j=1}^{e^{nR_{y}}}$
of \emph{satellite codewords} is chosen uniformly at random from the
conditional type class ${\cal T}_{n}(P_{X|U}|\boldsymbol{U}_{i})$.
Alternatively, this sub-codebook is referred to as a \emph{bin} \cite{merhav2014exact}.
The random codebook is then $\mathfrak{C}_{n}=\bigcup_{i=1}^{e^{nR_{z}}}\mathfrak{C}_{n}(i)=\{\boldsymbol{X}_{ij}\}$
which has size $e^{n(R_{y}+R_{z})}$ and thus is capable of sending
messages at a total rate of $R=R_{y}+R_{z}$. The weak user is only
intended to decoded the sub-codebook, that is, to decide which sub-codebook
$\{{\cal C}_{n}(i)\}$ contains the transmitted message, and thus
achieve a rate of $R_{z}$ (the rate of the common message, indexed
by $i$). The strong user decodes the codeword and achieves the total
rate $R$ (the common and the private message, indexed by $(i,j)$). 

Let us focus on the weak user. For a given hierarchical codebook ${\cal C}_{n}$,
the ML decoder, which minimizes the error probability of the weak
user, uses the likelihood 
\begin{equation}
W_{z}\left(\boldsymbol{Z}|{\cal C}_{n}(i)\right)\triangleq\Pr\left[\boldsymbol{Z}=\boldsymbol{z}|i\right]=\frac{1}{e^{nR_{y}}}\sum_{j=1}^{e^{nR_{y}}}W_{z}(\boldsymbol{z}|\boldsymbol{x}_{ij})=\frac{1}{e^{nR_{y}}}\sum_{j=1}^{e^{nR_{y}}}e^{n\cdot\alpha_{z}(\hat{Q}_{\boldsymbol{z},\boldsymbol{x}_{ij}})}\label{eq: index bin decoding}
\end{equation}
with the choice $\alpha_{z}(Q)=f_{z}(Q)\triangleq\E_{Q}[\ln W_{z}(\boldsymbol{Z}|\boldsymbol{X})]$,
that is, using the true channel likelihood. One can also replace this
choice with a channel-independent one, e.g., that of the MMI rule
$\alpha_{z}(Q)=I(Q_{XZ})$. In any case, the score of this decoder
for a single message $i$ is comprised of a sum over an \emph{exponential}
number of $e^{nR_{y}}$ satellite codewords . The effect of such a
decoding rule on the error exponent analysis is substantial. Indeed,
for the point-to-point channel considered before, an error event from
$\boldsymbol{x}$ to $\tilde{\boldsymbol{x}}$ given output $\boldsymbol{y}$
occurs for the ML decoder whenever $W(\boldsymbol{y}|\tilde{\boldsymbol{x}})\geq W(\boldsymbol{y}|\boldsymbol{x})$.
This event can be expressed using the corresponding joint types as
$f(\hat{Q}_{\tilde{\boldsymbol{x}}\boldsymbol{y}})\geq f(\hat{Q}_{\boldsymbol{x}\boldsymbol{y}})$,
and directly leads to a simple constraint $f(\tilde{Q}_{XY})\geq f(Q_{XY})$
in \eqref{eq: random coding exponent second} (when $\alpha(\cdot)=f(\cdot)$).
However, $W_{z}(\cdot|{\cal C}_{n}(i))$ is \emph{not} a memoryless
channel, and so the event of making an error from $i$ to $\tilde{i}$,
that is, $W_{z}(\boldsymbol{Z}|{\cal C}_{n}(\tilde{i}))\geq W_{z}(\boldsymbol{Z}|{\cal C}_{n}(i))$
cannot be expressed a simple relation between types as before. A naive
use of a union bound or Jensen-type inequalities to analyze this sum
of exponential number of terms typically fails in providing the exact
exponent. The TCEM ameliorates this by partitioning the summation
over the $e^{nR_{y}}$ private codewords of the strong user according
to their joint type $\hat{Q}_{\boldsymbol{z},\boldsymbol{x}_{ij}}$,
thus transforming the sum over an exponential number of likelihoods
to a sum over a polynomial number of average likelihoods. To show
this, we next evaluate the ensemble-average error probability of the
weak user. Nonetheless, we will do this in a slightly different way
compared to standard channel coding, in order to demonstrate another
technique. 

We assume, without loss of generality, that the first codeword $(1,1)$
is transmitted, and thus fix $(\boldsymbol{u}_{1},{\cal C}_{n}(1))$
as well as the output vector $\boldsymbol{z}$. The error probability
conditioned on these RVs is given by 
\begin{equation}
\Pr\left[\bigcup_{i=2}^{e^{nR_{z}}}\left\{ W_{z}(\boldsymbol{z}|{\cal \mathfrak{C}}_{n}(i))\geq W_{z}(\boldsymbol{z}|{\cal C}_{n}(1))\right\} \right],\label{eq: conditional error probability ABC first}
\end{equation}
where ${\cal \mathfrak{C}}_{n}(i)$ is the random code for the $i$th
bin. Now, we use the fact that the truncated union bound is exponentially
tight for pairwise independent events. That is, if $\{{\cal E}_{m}\}$
are \emph{pairwise independent} events then (e.g., \cite[p. 109, Lemma A.2]{shulman2003communication},
or from de Caen's inequality \cite{de1997lower})
\begin{equation}
\frac{1}{2}\cdot\min\left\{ \sum_{m}\Pr\left[{\cal E}_{m}\right],1\right\} \leq\Pr\left[\bigcup_{m}{\cal E}_{m}\right]\leq\min\left\{ \sum_{m}\Pr\left[{\cal E}_{m}\right],1\right\} .\label{eq: tightness of clipped union bound}
\end{equation}
Importantly, the number of events is arbitrary and could be exponential,
while still preserving exponential tightness. This bound can be further
generalized, as was shown in \cite[Appendix A]{scarlett2016multiuser}.
Exploiting this result and the fact that the events in \eqref{eq: conditional error probability ABC first}
are independent, we obtain that the probability is exponentially equal
to 
\begin{equation}
\min\left\{ e^{nR_{z}}\cdot\Pr\left[W_{z}(\boldsymbol{z}|{\cal \mathfrak{C}}_{n}(2))\geq W_{z}(\boldsymbol{z}|{\cal C}_{n}(1))\right],1\right\} .
\end{equation}
We may now focus on the inner probability, and as $\boldsymbol{z}$
and ${\cal C}_{n}(1)$ are fixed at this moment, we set, for brevity,
$s({\cal C}_{n}(1),\boldsymbol{z})=\frac{1}{n}\ln W_{z}(\boldsymbol{z}|{\cal C}_{n}(1))$.
We evaluate this probability in two steps. First, we condition on
$\boldsymbol{U}_{2}=\boldsymbol{u}_{2}$ and compute the conditional
probability according to a random choice of $\mathfrak{C}_{n}(2)$.
To this end, we define a proper TCE.
\begin{defn}[TCE for random-coding exponent of superposition codes]
\label{def: TCE superposition} For a superposition codebook ${\cal C}_{n}$,
a cloud center $\boldsymbol{u}$, an output vector $\boldsymbol{z}$
and a joint type $Q_{UXZ}$ such that $\hat{Q}_{\boldsymbol{u}\boldsymbol{z}}=Q_{UZ}$,
let 
\begin{equation}
N_{\boldsymbol{u},\boldsymbol{z}}(Q_{UXZ},{\cal C}_{n}(i))\triangleq\left|\left\{ 1\leq j\leq e^{nR_{y}}\colon(\boldsymbol{u},\boldsymbol{X}_{i,j},\boldsymbol{z})\in{\cal T}_{n}(Q_{UXZ})\right\} \right|.\label{eq: enumerator superposition}
\end{equation}
This TCE counts the number of codewords in a single bin ${\cal C}_{n}(i)$
of a superposition code defined by the cloud center $\boldsymbol{u}$,
which have a joint type $Q_{UXZ}$ with $\boldsymbol{z}$. By the
method of types, when $\boldsymbol{X}_{j}\sim\text{Uniform}[{\cal T}_{n}(Q_{X|U}|\boldsymbol{u})]$
\begin{equation}
\Pr\left[(\boldsymbol{u},\boldsymbol{X}_{j},\boldsymbol{z})\in{\cal T}_{n}(Q_{UXZ})\right]=k_{n}\cdot e^{-nI_{Q}(X;Z|U)}.
\end{equation}
for some $k_{n}\doteq1$. So, for a random codebook $\mathfrak{C}_{n}(i)=\{\boldsymbol{X}_{i,j}\}$,
\begin{equation}
N_{\boldsymbol{u},\boldsymbol{z}}(Q_{UXZ},\mathfrak{C}_{n}(i))=\sum_{j=1}^{e^{nR_{y}}}\I\left\{ (\boldsymbol{u},\boldsymbol{X}_{j},\boldsymbol{z})\in{\cal T}_{n}(Q_{UXZ})\right\} \sim\text{Binomial}\left(e^{nR_{y}},k_{n}\cdot e^{-nI_{Q}(X;Z|U)}\right).\label{eq: superposition enumerator}
\end{equation}
\end{defn}
As before, we will omit for brevity ${\cal C}_{n}(i)$ and $(\boldsymbol{u},\boldsymbol{z})$
from the notation of the TCE, as it does not affect its probability
distribution. With Definition \ref{def: TCE superposition}, the probability
of interest is 
\begin{align}
\Pr\left[W_{z}(\boldsymbol{z}|{\cal \mathfrak{C}}_{n}(2))\geq e^{ns}\right] & =\Pr\left[\sum_{Q_{UXZ}}N(Q_{UXZ})\geq e^{n\cdot s({\cal C}_{n}(1),\boldsymbol{z})}\right]\label{eq: weak user enumartor lower threshold}\\
 & \doteq\Pr\left[\max_{Q_{UXZ}}N(Q_{UXZ})\geq e^{n\cdot s({\cal C}_{n}(1),\boldsymbol{z})}\right]\nonumber \\
 & \doteq\max_{Q_{UXZ}}\Pr\left[N(Q_{UXZ})\geq e^{n\cdot s({\cal C}_{n}(1),\boldsymbol{z})}\right],
\end{align}
where we have used the summation--maximization equivalence \eqref{eq: max-to-sum-equivalence}.
Evidently, continuing evaluating this bound can be done by using the
tail probability of $N(Q_{UXZ})$, similarly to Section \ref{subsec:Probabilistic-Properties-of}. 

Given an exponentially-tight expression, we may average afterwards
over $\boldsymbol{U}_{2}$ by considering joint types $Q_{UZ}$ agreeing
with $\hat{Q}_{\boldsymbol{u}_{2},\boldsymbol{z}}$. At the next step,
it will be required to average over $(\boldsymbol{Z},\mathfrak{C}_{n}(1))$
and handling the randomness of $s(\mathfrak{C}_{n}(1),\boldsymbol{Z})$.
This again can be achieved with similarly defined TCEs. However, since
now the inequality is in a reversed direction compared to \eqref{eq: weak user enumartor lower threshold},
when proceeding this way, one encounters for some $t\in\reals$ the
following expression \cite[Proof of Theorem 1]{merhav2014exact} \cite[Section 5.1]{averbuch2018exact},
\begin{align}
 & \Pr\left[\sum_{Q_{UXZ}}N(Q_{UXZ})e^{n\cdot\alpha_{z}(Q_{UXZ})}\leq e^{n\cdot t}\right]\nonumber \\
 & \doteq\Pr\left[\bigcap_{Q_{UXZ}}\left\{ N(Q_{UXZ})\leq e^{n\cdot[t-\alpha_{z}(Q_{UXZ})]}\right\} \right]\nonumber \\
 & \doteq\I\left\{ \min_{Q_{UXZ}}\left\{ I_{Q}(X;Z|U)-R_{y}+\left[t-\alpha_{z}(Q_{UXZ})\right]_{+}\right\} >0\right\} 
\end{align}
where here $N(Q_{UXZ})$ is defined as in \eqref{eq: enumerator superposition},
yet for $\mathfrak{C}_{n}(1)$, and excluding $\boldsymbol{x}_{1,1}$,
and the last exponential inequality follows from Theorem \ref{thm: Intersection of tail events}.

\subsubsection{Distributed Compression and Random Binning }

Our next setting pertains to a source coding problem, and specifically,
the Slepian--Wolf (SW) problem of distributed lossless compression
\cite{slepian1973noiseless}. In this problem a source $X$ with a
finite alphabet ${\cal X}$ is given at the encoder side, and side-information
$Y$ of a finite alphabet ${\cal Y}$ that is at the decoder side.
The pair $(X,Y)$ is correlated, and follows a joint distribution
$P_{XY}$, and vectors $(\boldsymbol{X},\boldsymbol{Y})\in{\cal X}^{n}\times{\cal Y}^{n}$
are emitted from $P_{XY}$, with IID pairs of symbols. The source
vector $\boldsymbol{X}$ is compressed by assigning it to an index
$Z=f(\boldsymbol{X})$ of one of $e^{nR}$ possible bins, where $f\colon{\cal X}^{n}\to\{1,2,\ldots,e^{nR}\}$
is a called a \emph{binning} rule. Given the side-information $\boldsymbol{Y}=\boldsymbol{y}$
and the bin $Z=z$, the decoder decides that the source vector is
\begin{equation}
\hat{\boldsymbol{x}}(\boldsymbol{y},z)=\argmax_{\boldsymbol{x}\in f^{-1}(z)}\Pr\left[\boldsymbol{X}=\boldsymbol{x}|\boldsymbol{Y}=\boldsymbol{y}\right]\triangleq\argmax_{\boldsymbol{x}\in f^{-1}(z)}\Pr\left[\boldsymbol{x}|\boldsymbol{y}\right].
\end{equation}
For a giving binning rule $f$, the error probability is then given
by 
\begin{equation}
P_{\mathsf{e}}(f)\triangleq\sum_{\boldsymbol{x}\in{\cal X}^{n}}\sum_{\boldsymbol{y}\in{\cal Y}^{n}}\Pr\left[\boldsymbol{X}=\boldsymbol{x},\boldsymbol{Y}=\boldsymbol{y}\right]\cdot\I\left[\bigcup_{\boldsymbol{x}'\neq\boldsymbol{x}\colon\Pr[\boldsymbol{x}'|\boldsymbol{y}]\geq\Pr[\boldsymbol{x}|\boldsymbol{y}]}\left\{ f(\boldsymbol{x}')=f(\boldsymbol{x})\right\} \right].
\end{equation}
For a joint type $Q_{XY}$, let us further denote , the expected log-posterior
as
\begin{equation}
g(Q_{XY})\triangleq\E_{Q}\left[\ln P_{X|Y}(X|Y)\right],
\end{equation}
so that for any $(\boldsymbol{x},\boldsymbol{y})\in{\cal T}_{n}(Q_{XY})$
it holds that $\Pr[\boldsymbol{x}|\boldsymbol{y}]=e^{ng(Q_{XY})}$.

As expected, the optimal binning rule $f$ is infeasible to find.
However, whenever the compression rate $R$ is above the minimal required
rate $H(X|Y)$, the ensemble average of the error probability over
random choice of binning functions decays exponentially \cite{csiszar1981graph,ahlswede1982good,chen2008universal,gallager1976source,merhav2020universal,merhav2021more,oohama1994universal,weinberger2015optimum}.
In fact, the optimal error exponent for the SW problem is directly
related to the random-coding error exponent in channel coding (see,
e.g., \cite{ahlswede1982good,chen2008universal,weinberger2015optimum}).
Thus, we may evaluate the error probability averaged over random choice
of binning rules, referred to as \emph{random binning}. Accordingly,
the exponential decay of the average error probability is called the
random-binning error exponent. In simple random binning, the random
rule $F$ is such that the bin of any $\boldsymbol{x}\in{\cal X}^{n}$
is chosen uniformly at random from the $e^{nR}$ possible bins. Analogously
to TCEs for channel coding, we define the following TCE:
\begin{defn}[TCE for random-binning exponent]
\label{def: TCE random binning}For a binning rule $f$, a side-information
vector $\boldsymbol{y}$, an encoded index $z$, and a joint type
$Q_{XY}'$ such that $\hat{Q}_{\boldsymbol{y}}=Q_{Y}'$, let
\begin{equation}
\tilde{N}_{\boldsymbol{y},z}(Q_{XY}',f)\triangleq\left|\left\{ (\boldsymbol{x}',\boldsymbol{y})\in{\cal T}_{n}(Q_{XY}')\cap f^{-1}(z)\right\} \right|.\label{eq: enumerator for random binning}
\end{equation}
The TCE $\tilde{N}_{\boldsymbol{y},z}(Q_{XY},f)$ counts the number
of vectors in ${\cal T}_{n}(Q_{X})$ except for $\boldsymbol{x}$,
which have joint type $Q_{XY}'$ with $\boldsymbol{y}$, and that
are mapped to the index $z$. By the method of types, when 
\begin{equation}
|{\cal T}_{n}(Q_{XY}')|=k_{n}\cdot e^{nH(Q_{XY}')}
\end{equation}
for some $k_{n}\doteq1$. So, for a random binning rule $F(\boldsymbol{x}')\sim\text{Uniform}\{1,2,\ldots,e^{nR}\}$,
\begin{align}
\tilde{N}_{\boldsymbol{y},z}(Q_{XY}',F) & =\sum_{\boldsymbol{x}'\in{\cal T}_{n}(Q_{X})\colon\boldsymbol{x}'\neq\boldsymbol{x}}\I\left\{ (\boldsymbol{x}',\boldsymbol{y})\in{\cal T}_{n}(Q_{XY}')\cap F^{-1}(z)\right\} \nonumber \\
 & \sim\text{Binomial}\left(k_{n}\cdot e^{nH_{Q''}(X|Y)},e^{-nR}\right).\label{eq: random binning enumerator}
\end{align}
This TCE displays a plausible \emph{duality} with the TCE of channel
coding: In random coding the number of trials is fixed and the success
probability is type-dependent, whereas in random binning, it is the
other way round. 
\end{defn}
As in channel-coding analysis, we will simplify the notation to $\tilde{N}_{\boldsymbol{y},z}(Q_{XY}')$
or even just $\tilde{N}(Q_{XY}')$ when the TCE is an RV. Given Definition
\ref{def: TCE random binning}, the random-binning error probability
is then given by
\begin{align}
\overline{P}_{\mathsf{e}} & \triangleq\E\left[P_{\mathsf{e}}(F)\right]\nonumber \\
 & =\sum_{\boldsymbol{x}\in{\cal X}^{n}}\sum_{\boldsymbol{y}\in{\cal Y}^{n}}\Pr\left[\boldsymbol{X}=\boldsymbol{x},\boldsymbol{Y}=\boldsymbol{y}\right]\cdot\Pr\left[\bigcup_{Q_{XY}'\colon Q_{Y}'=\hat{Q}_{\boldsymbol{y}},\;g(Q_{XY}')\geq g(\hat{Q}_{\boldsymbol{x}\boldsymbol{y}})}\left\{ \tilde{N}_{\boldsymbol{y},f(\boldsymbol{x})}(Q_{XY}')\geq1\right\} \right]\nonumber \\
 & \doteq\sum_{Q_{XY}}\Pr\left[(\boldsymbol{X},\boldsymbol{Y})\in{\cal T}_{n}(Q_{XY})\right]\cdot\sum_{Q'\colon Q'_{Y}=Q_{Y},\;g(Q')\geq g(Q)}\Pr\left[\tilde{N}(Q_{XY}')\geq1\right]\nonumber \\
 & \doteq\max_{Q_{XY}}\max_{Q'_{XY}\colon Q'_{Y}=Q_{Y},\;g(Q'_{XY})\geq g(Q_{XY})}e^{-nD(Q_{XY}||P_{XY})}\cdot\Pr\left[\tilde{N}(Q'_{XY})\geq1\right].
\end{align}
The tail probability $\Pr[\tilde{N}(Q_{XY}')\geq1]$ can be analyzed
as in Section \ref{subsec:Probabilistic-Properties-of}, and this
results the exact random-binning exponent.

\subsubsection{Generalized Decoders \label{subsec:Generalized-Decoders}}

\paragraph{Erasure/List Decoders.}

An erasure decoder may either decode a message or declare an \emph{erasure},
that is, not to output any message. A list decoder, may output multiple
codewords, whose number is either fixed in advance, or varies according
to the channel output. As was shown by Forney \cite{forney1968exponential},
the optimal erasure decoder uses the posterior probability, rather
than the likelihood, to decide on its output. Concretely, consider
a codebook ${\cal C}_{n}=\{\boldsymbol{x}_{m}\}$, from which a codeword
$\boldsymbol{X}$ will be chosen under the uniform distribution. The
codeword is then transmitted over a channel $W$, and given an output
vector $\boldsymbol{Y}=\boldsymbol{y}$, the posterior probability
of the $m$th codeword is given by Bayes rule as 
\begin{equation}
\Pr\left[\boldsymbol{X}=\boldsymbol{x}_{m}|\boldsymbol{Y}=\boldsymbol{y}\right]=\frac{W\left[\boldsymbol{y}|\boldsymbol{x}_{m}\right]}{\sum_{m'}W\left[\boldsymbol{y}|\boldsymbol{x}_{m'}\right]}.
\end{equation}
If the maximal posterior over codewords is large enough, then the
maximizing codeword is decoded. Otherwise, an erasure is declared.
Equivalently, we may set a threshold $T>0$, so that the optimal erasure
decoder outputs message $m$ if $\boldsymbol{x}_{m}$ is the (unique)
codeword such that
\begin{equation}
\frac{W\left[\boldsymbol{y}|\boldsymbol{x}_{m}\right]}{\sum_{m'\neq m}W\left[\boldsymbol{y}|\boldsymbol{x}_{m'}\right]}>e^{nT}.\label{eq: erasure list decoder}
\end{equation}
The threshold $T$ determines the trade-off between two types of failure
events: An \emph{erasure event} ${\cal E}_{1}'({\cal C}_{n})$, in
which no codeword is decoded, or an \emph{undetected error event}
${\cal E}_{2}({\cal C}_{n})$, in which a wrong codeword is decoded.
The event ${\cal E}_{1}({\cal C}_{n})={\cal E}_{1}'({\cal C}_{n})\cup{\cal E}_{2}({\cal C}_{n})$
is then called the \emph{total-error event}. As can be seen, the score
of the decoder is a complicated function, since the denominator in
\eqref{eq: erasure list decoder} includes a summation over an exponential
number of likelihoods. 

To bound the probability of the total-error event, Forney has introduced
a parameter $s\in[0,1]$ and derived a Chernoff-style bound. As said,
this bounding method is not guaranteed to be tight, and indeed leads
to strictly sub-optimal exponents. The TCEM addresses the problem
of evaluating the probability of the total-error event, by using the
TCE $N(\tilde{Q}_{XY})$ of Definition \ref{def: TCE random coding}.
For a random codebook $\mathfrak{C}_{n}$, 
\begin{align}
 & \E\left\{ \Pr\left[{\cal E}_{1}(\mathfrak{C}_{n})\right]\right\} \nonumber \\
 & =\E\left\{ \Pr\left[{\cal E}_{1}(\mathfrak{C}_{n})|\boldsymbol{X}_{1}\text{ transmitted}\right]\right\} \nonumber \\
 & =\Pr\left[\sum_{m'>1}W\left[\boldsymbol{Y}|\boldsymbol{X}_{m'}\right]\geq e^{-nT}\cdot W\left[\boldsymbol{Y}|\boldsymbol{X}_{1}\right]\right]\nonumber \\
 & =\sum_{Q_{XY}}\Pr\left[(\boldsymbol{X}_{1},\boldsymbol{Y})\in{\cal T}_{n}(Q_{XY})\right]\cdot\Pr\left[\sum_{Q_{XY}'\colon\;\tilde{Q}_{Y}=Q_{Y}}N(Q_{XY}')e^{nf(Q_{XY}')}\geq e^{-nT}\cdot e^{-nf(Q_{XY})}\right].
\end{align}
Now, the inner probability may be evaluated by the summation--maximization
equivalence
\begin{align}
 & \Pr\left[\sum_{Q_{XY}'\colon\;Q_{Y}'=Q_{Y}}N(Q_{XY}')e^{nf(Q_{XY}')}\geq e^{-nT}\cdot e^{-nf(Q_{XY})}\right]\nonumber \\
 & \doteq\Pr\left[\max_{Q_{XY}'\colon\;Q_{Y}'=Q_{Y}}N(Q_{XY}')e^{nf(Q_{XY}')}\geq e^{-nT}\cdot e^{-nf(Q_{XY})}\right]\nonumber \\
 & =\Pr\left[\bigcup_{Q_{XY}'\colon\;Q_{Y}'=Q_{Y}}\left\{ N(Q_{XY}')e^{nf(Q_{XY}')}\geq e^{-nT}\cdot e^{-nf(Q_{XY})}\right\} \right]\nonumber \\
 & \doteq\max_{Q_{XY}'\colon\;Q_{Y}'=Q_{Y}}\Pr\left[N(Q_{XY}')\geq e^{-n[f(Q_{XY})-f(Q_{XY}')-T]}\right].
\end{align}
The derivation is completed by the exact exponential analysis of the
tail probability of $N(Q_{XY}')$ from Section \ref{subsec:Probabilistic-Properties-of}.
It can be shown that the resulting random-coding error exponent of
$\E[{\cal E}_{2}(\mathfrak{C}_{n})]$ is larger by exactly $T$ than
the exponent of the total-error event \cite{huleihel2016erasure}.
Thus, the exponent of the total-error event is equal to that of the
erasure event. When $T<0$, the decoder in \eqref{eq: erasure list decoder}
is a list decoder, with a \emph{variable} list size. In this regime,
the trade-off is between the exponent of the error event (where the
correct codeword is not in the output list), and the normalized logarithm
of the expected list size. It can be shown that the expressions for
these values are exactly as the ones of the total-error exponent and
undetected error exponent in the erasure regime $T>0$, and so the
analysis is identical, while just allowing $T<0$. 

\paragraph{Likelihood Decoders.}

Similarly to an erasure/list decoder, a likelihood decoder \cite{yassaee2013technique}
also uses the posterior probability. However, it outputs a \emph{random}
codeword based on this posterior, so the decoded message is $m$ with
probability 
\begin{equation}
\Pr\left[\boldsymbol{x}_{m}|\boldsymbol{y}\right]=\frac{W(\boldsymbol{Y}|\boldsymbol{x}_{m})}{\sum_{\tilde{m}=1}^{e^{nR}}W(\boldsymbol{Y}|\boldsymbol{x}_{\tilde{m}})}.\label{eq: ordinary likelihood decoder}
\end{equation}
More generally, one may choose a continuous function $g(Q_{XY})$
of joint types and an inverse-temperature $\beta>0$, and consider
a likelihood decoder that decodes message $m$ with probability 
\begin{equation}
\Pr\left[\boldsymbol{x}_{m}|\boldsymbol{y}\right]=\frac{\exp\left[n\beta g(\hat{Q}_{\boldsymbol{x}_{m},\boldsymbol{y}})\right]}{\sum_{\tilde{m}=1}^{e^{nR}}\exp\left[n\beta g(\hat{Q}_{\boldsymbol{x}_{\tilde{m}},\boldsymbol{y}})\right]}.\label{eq: generalized likelihood decoder}
\end{equation}
When $\beta=1$ and $g(Q_{XY})=f(Q_{XY})=\E_{Q}[\ln W(Y|X)]$ then
the \emph{ordinary likelihood decoder} \eqref{eq: ordinary likelihood decoder}
is reproduced. However, $g(\cdot)$ can be replaced by a choice that
is mismatched to the channel, or even by a universal function such
as $g(Q_{XY})=I(Q_{XY})$. Similar to finite-temperature decoding
\cite{rujan1993finite}, the parameter $\beta$ controls the ``amount
of stochasticity'' of the decoder: If $\beta\to\infty$ then the
decoder becomes deterministic, reproducing the score-based decoder
with $\alpha\equiv g$. As the temperature increases and $\beta$
decreases, the decoder becomes more random (at the extreme $\beta=0$,
the output codeword is chosen uniformly at random). On the upside,
The ensemble average error probability follows a remarkably simple
formula, given by
\begin{align}
\overline{P_{\mathsf{e}}}=\E\left[P_{\mathsf{e}}(\mathfrak{C}_{n})\right] & =\E\left[P_{\mathsf{e}}\left(\mathfrak{C}_{n}|\boldsymbol{X}_{1}\text{ transmitted}\right)\right]\nonumber \\
 & =\E\left[\frac{\sum_{m\geq2}\exp\left[n\beta g(\hat{Q}_{\boldsymbol{X}_{m},\boldsymbol{Y}})\right]}{\sum_{m}\exp\left[n\beta g(\hat{Q}_{\boldsymbol{X}_{m},\boldsymbol{Y}})\right]}\right].
\end{align}
On the downside, in this expression, \emph{both} the numerator and
denominator contain an exponential number of codewords, and this makes
its analysis more challenging. Following the TCEM, let us condition
on the joint type $(\boldsymbol{X}_{1},\boldsymbol{Y})\in{\cal T}_{n}(Q_{XY})$.
Then, using the TCE of Definition \ref{def: TCE random coding}, 
\begin{equation}
\overline{P_{\mathsf{e}}}=\sum_{Q_{XY}}\Pr\left[(\boldsymbol{X}_{1},\boldsymbol{Y})\in{\cal T}_{n}(Q_{XY})\right]\cdot e(Q_{XY})
\end{equation}
where, for any given $Q_{XY}$, 
\begin{equation}
e(Q_{XY})=\E\left[\frac{\sum_{\tilde{Q}_{XY}\colon\tilde{Q}_{Y}=Q_{Y}}N(\tilde{Q}_{XY})e^{n\beta g(\tilde{Q}_{XY})}}{e^{n\beta g(Q_{XY})}+\sum_{\tilde{Q}_{XY}\colon\tilde{Q}_{Y}=Q_{Y}}N(\tilde{Q}_{XY})e^{n\beta g(\tilde{Q}_{XY})}}\right].
\end{equation}
Using the tail-integration identity $\E[X]=\int_{0}^{\infty}\Pr[X\geq t]\dd t$,
which holds for any non-negative RV. $X$, as well as the summation--maximization
equivalence, we obtain 
\begin{align}
e(Q_{XY}) & \doteq\E\left[\min\left\{ \sum_{\tilde{Q}_{XY}\colon\tilde{Q}_{Y}=Q_{Y}}N(\tilde{Q}_{XY})e^{n[\beta g(\tilde{Q}_{XY})-\beta g(Q_{XY})]},1\right\} \right]\nonumber \\
 & =\int_{0}^{\infty}\Pr\left[\min\left\{ \sum_{\tilde{Q}_{XY}\colon\tilde{Q}_{Y}=Q_{Y}}N(\tilde{Q}_{XY})e^{n[\beta g(\tilde{Q}_{XY})-\beta g(Q_{XY})]},1\right\} \geq t\right]\dd t\nonumber \\
 & =\int_{0}^{1}\Pr\left[\sum_{\tilde{Q}_{XY}\colon\tilde{Q}_{Y}=Q_{Y}}N(\tilde{Q}_{XY})e^{n[\beta g(\tilde{Q}_{XY})-\beta g(Q_{XY})]}\geq t\right]\dd t\nonumber \\
 & =n\int_{0}^{\infty}e^{-n\theta}\cdot\Pr\left[\sum_{\tilde{Q}_{XY}\colon\tilde{Q}_{Y}=Q_{Y}}N(\tilde{Q}_{XY})e^{n[\beta g(\tilde{Q}_{XY})-\beta g(Q_{XY})]}\geq e^{-n\theta}\right]\dd\theta\nonumber \\
 & \doteq\max_{\tilde{Q}_{XY}\colon\tilde{Q}_{Y}=Q_{Y}}n\int_{0}^{\infty}e^{-n\theta}\cdot\Pr\left[N(\tilde{Q}_{XY})\geq e^{-n[\theta-\beta g(\tilde{Q}_{XY}+\beta g(Q_{XY}))]}\right]\dd\theta.
\end{align}
The derivation continuous by plugging into the integral the tight
exponent of the tail probability of $N(\tilde{Q}_{XY})$ from Section
\ref{subsec:Probabilistic-Properties-of}. Then, the exponential decay
rate of the integral can be determined using Laplace method from Chapter
\ref{sec: laplacesaddlepoint}. After averaging WRT to $(\boldsymbol{X}_{1},\boldsymbol{Y})$,
the resulting expression may be minimized over $Q_{XY}$ to obtain
the exact exponent of the ensemble-average error probability. 

\subsubsection{Error Exponent of the Typical Random Code}

Inspecting the definition of the random-coding error exponent \eqref{eq: random coding definition},
it is appears to be somewhat at odds with the goal of being a performance
measure for the exponent of a typical random codebook from the ensemble.
Indeed, a direct way is to evaluate the error exponent of the \emph{typical
random code} is by the exponent 
\begin{equation}
E_{\text{trc}}(R)\triangleq\E\left[-\frac{1}{n}\ln P_{\mathsf{e}}(\mathfrak{C}_{n})\right].
\end{equation}
This exponent leads to tighter bounds since Jensen's inequality assures
that $E_{\text{trc}}(R)\geq E_{\text{rc}}(R)$. The exponent $E_{\text{trc}}(R)$
is determined by \emph{typical} codebooks, whereas the random-coding
error exponent is actually dominated by \emph{unlikely} poor codebooks.
This can be seen from the following informal argument. Let ${\cal G}_{E}$
be the collection of codes $\{{\cal C}_{n}\}$ for which $P_{\mathsf{e}}({\cal C}_{n})\approx e^{-nE}$.
Then, approximating the values of $E$ by a discrete fine grid, 
\begin{equation}
\E[P_{\mathsf{e}}({\cal \mathfrak{C}}_{n})]\approx\sum_{E}\Pr\left[{\cal \mathfrak{C}}_{n}\in{\cal G}_{E}\right]\cdot e^{-nE}.
\end{equation}
This term is dominated by the largest term in the sum, yet the maximizer
may occur for $\tilde{E}$ in which $\Pr[{\cal \mathfrak{C}}_{n}\in{\cal G}_{\tilde{E}}]$
is exponentially small. By contrast, it holds that
\begin{equation}
E_{\text{trc}}(R)=\E\left[-\frac{1}{n}\ln P_{\mathsf{e}}({\cal C}_{n})\right]=\sum_{E}\Pr\left[{\cal \mathfrak{C}}_{n}\in{\cal G}_{E}\right]\cdot E.
\end{equation}
Thus, if there exists a value $E_{0}$ for which $\Pr[{\cal \mathfrak{C}}_{n}\in{\cal G}_{E}]\to1$
then $E_{\text{trc}}(R)=E_{0}$ (and such $E_{0}$ does exist). Evidently,
the averaging of the normalized logarithm over the error probability
mitigates the effect of high error-probability codebooks on the ensemble
average. 

Despite this obvious advantage, the error exponent of the typical
random code was considered to be more difficult to evaluate than the
random-coding error exponent, and thus was somewhat ignored in the
traditional developments of bounds on the reliability function. In
\cite{barg2002random}, Barg and Forney have evaluated the error exponent
of the typical random code for the BSC (and credit \cite{bassalygo1991simple}
for inspiration). The derivation is sufficiently simple to be done
directly, and involves the typical \emph{distance spectrum} of the
code, given by $\{\overline{N}(d)\}_{d=0}^{n}$ where
\begin{equation}
\overline{N}(d)\triangleq\left|\left\{ m_{1},m_{2}\colon m_{1}\neq m_{2}\;d_{\text{H}}(\boldsymbol{x}_{m_{1}},\boldsymbol{x}_{m_{2}})=d\right\} \right|\label{eq: Hamming distance enumerator}
\end{equation}
where $d_{\text{H}}(\cdot,\cdot)$ is the Hamming distance. The typical
behavior of $\overline{N}(d)$ over the ensemble was then determined,
and the error probability was then tightly bounded by the union bound
as
\begin{equation}
P_{\mathsf{e}}({\cal C}_{n})\leq\sum_{d=0}^{n}\overline{N}(d)\cdot e^{-d\cdot D(\frac{1}{2}||p)}.\label{eq: bound on error probability using distance spectrum}
\end{equation}
(for the typical random exponent, the tightness of this bound is credited
by Barg and Forney to \cite{d1980bounds}). 

If we generalize this expression to general DMCs, then the derivation
of the error exponent of the typical random code will involve the
following TCE:
\begin{defn}[TCE for the error exponent of the typical random code]
\label{def: TCE typical}For a codebook ${\cal C}_{n}$, and a joint
type $Q_{X\tilde{X}}$, let

\begin{align}
\overline{N}(Q_{X\tilde{X}},{\cal C}_{n}) & =\left|\left\{ (m,\tilde{m})\colon m\neq\tilde{m},\;(\boldsymbol{x}_{m},\boldsymbol{x}_{\tilde{m}})\in{\cal T}_{n}(Q_{X\tilde{X}})\right\} \right|\nonumber \\
 & =\sum_{m=1}^{e^{nR}}\overline{N}_{m}(Q_{X\tilde{X}},{\cal C}_{n})\nonumber \\
 & =\sum_{m=1}^{e^{nR}}\sum_{\tilde{m}=1}^{e^{nR}}\I\{\tilde{m}\neq m\}\cdot\I\left\{ (\boldsymbol{X}_{\tilde{m}},\boldsymbol{X}_{m})\in{\cal T}_{n}(Q_{X\tilde{X}})\right\} ,\label{eq: typical enumerator}
\end{align}
where $\overline{N}_{m}(Q_{X\tilde{X}})$ is as defined in \eqref{eq: exponent for expurgated}.
The TCE $\overline{N}(Q_{X\tilde{X}},{\cal C}_{n})$ counts the number
of pairs of codewords in the codebook which have a joint type $Q_{X\tilde{X}}$.
\end{defn}
For a random codebook, the TCE $\overline{N}(Q_{X\tilde{X}},\mathfrak{C}_{n})$
has no simple probabilistic description. As for previous TCEs we abbreviate
the notation to $\overline{N}(Q_{X\tilde{X}})$ when it is an RV.
The derivation of the properties $\overline{N}(Q_{X\tilde{X}})$ introduces
a technical challenge due to dependencies between pairs of codewords.
Indeed, $\overline{N}(Q_{X\tilde{X}})$ counts the number of successes
in $e^{nR}(e^{nR}-1)\doteq e^{2nR}$ trials, and the success probability
of each trial is $\doteq e^{-nI(Q_{X\tilde{X}})}$. However, the trails
are \emph{not} mutually independent, and so $\overline{N}(Q_{X\tilde{X}})$
is \emph{not} a binomial RV. This dependence can be demonstrated by
the following extreme example: Let $Q_{X}$ be uniform over ${\cal X}$
and let $Q_{X\tilde{X}}$ be the joint type such equals to $1/|{\cal X}|$
whenever $x=\tilde{x}$ and $0$ otherwise. Then, without any prior
knowledge, for every $\tilde{m}\neq m$, $\Pr[\boldsymbol{X}_{m}=\boldsymbol{X}_{\tilde{m}}]=\Pr[(\boldsymbol{X}_{m},\boldsymbol{X}_{\tilde{m}})\in{\cal T}_{n}(Q_{X\tilde{X}})]\doteq\exp[-nI(Q_{X\tilde{X}})]$
where $I(Q_{X\tilde{X}})=\ln|{\cal X}|$. Now, conditioned on $\boldsymbol{X}_{1}=\boldsymbol{X}_{2}$
and $\boldsymbol{X}_{2}=\boldsymbol{X}_{3}$ it holds with probability
$1$ that $\boldsymbol{X}_{1}=\boldsymbol{X}_{3}$. Nonetheless, this
dependence is ``weak'', and as we will show, some of its asymptotic
properties can be shown to be indifferent to this dependence, and
match those of a regular $\text{Binomial}(e^{2nR},e^{-nI(Q_{X\tilde{X}})})$
distribution. 

The required moments and tail properties of $\overline{N}(Q_{X\tilde{X}})$
were evaluated as follows. In \cite[Theorem 3]{tamir2020large}, it
was determined that for any $s\in\reals$
\begin{equation}
\lim_{n\to\infty}-\frac{1}{n}\ln\Pr\left[\overline{N}(Q_{X\tilde{X}})\geq e^{ns}\right]=\begin{cases}
[I(Q_{X\tilde{X}})-2R]_{+} & [2R-I(Q_{X\tilde{X}})]_{+}>s\\
\infty, & [2R-I(Q_{X\tilde{X}})]_{+}<s
\end{cases},
\end{equation}
which is the same upper tail behavior as for a $\text{Binomial}(e^{2nR},e^{-nI(Q_{X\tilde{X}})})$.
This is intuitively justified because the events $\I\{(\boldsymbol{x}_{m_{1}},\boldsymbol{x}_{m_{2}})\in{\cal T}_{n}(Q_{X\tilde{X}})\}$
and $\I\{(\boldsymbol{x}_{m_{3}},\boldsymbol{x}_{m_{4}})\in{\cal T}_{n}(Q_{X\tilde{X}})\}$
are pairwise-independent, even if $m_{1}=m_{3}$, and as the overall
dependence between all events is fairly low. The proof of this upper
tail bound is based on bounding its integer moments, and showing that
for any $k\in\mathbb{N}$
\begin{equation}
\E\left[\overline{N}^{k}(Q_{X\tilde{X}})\right]\dot{\leq}\begin{cases}
e^{nk[2R-I(Q_{X\tilde{X}})]} & I(Q_{X\tilde{X}})<2R\\
e^{n[2R-I(Q_{X\tilde{X}})]}, & I(Q_{X\tilde{X}})>2R
\end{cases}.\label{eq: moment bounds for full codebook enumerator}
\end{equation}
The proof is then completed by applying Markov's inequality with an
arbitrarily large $k$. For $k=1$, the bound of \eqref{eq: moment bounds for full codebook enumerator}
readily follows by the linearity of the expectation. In turn, the
proof of \eqref{eq: moment bounds for full codebook enumerator} for
an arbitrary $k$ follows by a careful induction. For typically empty
types, the analysis is based on \emph{Janson's inequality} \cite[Theorem 9]{janson1998new}
for the probability of the event that $\overline{N}(Q_{X\tilde{X}})=0$. 

As for the lower tail, it was determined in \cite[Lemma 2]{tamir2020large}
that given $\epsilon\in(0,2R)$, if $I(Q_{X\tilde{X}})\leq2R-\epsilon$
then 
\begin{equation}
\lim_{n\to\infty}-\frac{1}{n}\ln\Pr\left[\overline{N}(Q_{X\tilde{X}})\leq e^{-n\epsilon}\cdot\E\left[\overline{N}(Q_{X\tilde{X}})\right]\right]=\infty
\end{equation}
(see a more accurate statement therein). The proof is inspired by
an investigation of \emph{typicality graphs} by Nazari \emph{et al.}\ \cite{nazari2010typicality}.
For typically populated TCEs ($2R\geq I(Q_{X\tilde{X}})$), the analysis
is based on a lower tail-bound form of Janson's inequality \cite[Theorem 3]{janson1998new}:
For our analysis, this Janson's bound is a tail bound for the sum
of possibly \emph{dependent} Bernoulli RVs, which is suitable for
settings in which each Bernoulli RV only depends on a small number
of other Bernoulli RVs. Another useful moment result that was derived,
is a bound on the correlation of TCE powers, given by \cite[Proposition 4]{tamir2020large}
as 
\begin{equation}
\E\left[\overline{N}^{k}(Q_{X\tilde{X}})N^{\ell}(Q_{XY})\right]\dot{\leq}F(R,Q_{XY},\ell)\cdot F(2R,Q_{X\tilde{X}},k),\label{eq: moment correlation bounds for full codebook enumerator and enumerator}
\end{equation}
where for a joint type $Q_{UV}$, $S\geq0$ and $j\in\mathbb{N}$,
\begin{equation}
F(S,Q_{UV},j)\triangleq\begin{cases}
e^{nj[R-I(Q)]}, & I(Q)<S\\
e^{n[R-I(Q)]}, & I(Q)>S
\end{cases}.
\end{equation}
The bound \eqref{eq: moment correlation bounds for full codebook enumerator and enumerator}
again shows that asymptotic independence of TCEs. The proof of \eqref{eq: moment correlation bounds for full codebook enumerator and enumerator}
generalizes the proof of \eqref{eq: moment bounds for full codebook enumerator}
and utilizes a double-induction on both $k$ and $\ell$.

\subsection{Applications}

In the last 15 years, the TCEM has found extensive applications in
diverse coding problems. We next briefly review these applications. 

Expurgated exponents were considered in \cite{merhav2013another,scarlett2014expurgated}
for both standard channel coding, as well as mismatched decoding and
under input constraints, utilizing Definition \ref{def: TCE expurgated}
and an expurgation argument based on TCEs (see Appendix \ref{sec:The-Expurgated-Exponent}). 

The TCEM was widely used in multiuser and network problems \cite{el2011network}.
For the broadcast channel \cite{bergmans1973random,cover1972broadcast,gallager1974capacity,korner1977general,korner1980universally},
random-coding and expurgated error exponents were derived using the
TCEM in \cite{merhav2014exact,averbuch2018exact,kaspi2010error,averbuch2019expurgated},
for various decoders, including the optimal bin-index decoder. For
the multiple-access channel (MAC), concurrently to the early development
stages of the TCEM, Nazari \emph{et al.\ }\cite{nazari2014error}
also used TCEs (referred to as ``packing functions''), but only
derived bounds as their derivation was only based on the expectation
and variance of the TCEs. Scarlett, Martinez and Guill\'{e}n i F\`{a}bregas
fully utilized the used the TCEM for the MAC in \cite{scarlett2016multiuser,scarlett2017mismatched}.
For the interference channel, random-coding error exponents for the
Han--Kobayashi scheme \cite{han1981new} and under the optimal ML
decoder, were derived in \cite{etkin2009error,huleihel2016random}.
For the wiretap channel \cite{wyner1975wire}, assuming a multi-coding
scheme \cite{wyner1975wire,csiszar1978broadcast,liang2009information},
the correct-decoding exponent of the eavesdropper (as in \cite{merhav2003large})
was derived in \cite{merhav2014exact_correct}, and the exponential
decay rate of the mutual information between the message and the eavesdropper
output vector (\emph{unnormalized} by the blocklength $n$) was derived
in \cite{parizi2016exact}; this later result refined previous bounds
in \cite{hayashi2006general,hayashi2011exponential,parizi2015secrecy}.
The dirty-paper \cite{costa1983writing} and the Gel'fand--Pinsker
\cite{gel1980coding} channels were analyzed using the TCEM in \cite{TM23},\footnote{By means of source-channel duality \cite{keshet2008channel,pradhan2003duality}
these results are also applicable to the Wyner--Ziv distributed lossy
compression problem \cite{wyner1976rate}. } improving the bounds of Moulin and Wang \cite{moulin2007capacity}.
The works above heavily rely on the idea of superposition coding and
index-bin decoding, which evidently has wide applicability in multiuser
problems \cite{el2011network}. Some settings in which it has not
been applied yet include the relay channel \cite{cover1979capacity},
and channels with feedback \cite{cover1981achievable,ozarow1984capacity}.
Similar derivations and results are therefore anticipated to these
settings too.

For source coding problems, the TCEM was mainly used in distributed
compression, and specifically for deriving exact random-binning exponents
\cite{merhav2015statistical,merhav2016universal}. For secure lossy
compression, the optimal trade-off between the excess-distortion exponent
of the legitimate receiver and the exiguous-distortion exponent of
the eavesdropper was derived in \cite{weinberger2017large}. For distributed
hypothesis testing \cite{ahlswede1986hypothesis,han1998statistical,han2006exponential,rahman2012optimality},
type-I and type II error exponents were derived for the \emph{quantization-and-binning}
scheme \cite{han1998statistical,shimokawa1994error} under optimal
decision rule in \cite{weinberger2019reliability}. 

Returning to channel coding, the TCEM was extensively used to derive
error exponents for generalized decoders. For Forney's erasure/list
decoder \cite{forney1968exponential}, random-coding and expurgated
error exponents were derived in \cite{huleihel2016erasure,merhav2008error,somekh2011exact,merhav2014erasure,weinberger2017simplified,ginzach2017random}
and by Cao and Tan \cite{cao2019exact} in a broadcast channel setting.
Among other results, this has shown that the exact exponents can be
arbitrarily large compared to Forney's bounds, and that, unlike for
ordinary decoding \cite{feder2002universal,merhav2007minimax}, they
are not universally achievable. Hayashi and Tan \cite{hayashi2015asymmetric}
used the TCEM for erasure decoding in the \emph{moderate deviation
regime }\cite{altuug2014moderate,polyanskiy2010moderate}. The exponents
of a decoder with a fixed list size were derived in \cite{merhav2014list}.
The result matches the celebrated converse bound of Shannon, Gallager
and Berlekamp \cite[Theorem 2]{shannon1967lowerI} and improved the
best lower bound previously known \cite[p. 196, Problem 27, part (a)]{csiszar2011information}.
For asynchronous sparse communication, \cite{chandar2012asynchronous,chandar2008optimal,tchamkerten2009communication,tchamkerten2012asynchronous},
false-alarm and mis-detection exponents were derived for joint codeword
detection and decoding in \cite{weinberger2014codeword}, improving
Wang \emph{et al.\ }\cite{wang2011error}. This result was generalized
in \cite{weinberger2017channel} to joint channel detection and decoding.
Exponents for generalized likelihood decoding \cite{yassaee2013technique}
were derived in \cite{merhav2017correction,merhav2017generalized}.
Additional settings include decoding for biometric identification
\cite[Chapter 5]{ignatenko2012biometric}, \cite{tuncel2009capacity,willems2003capacity}
and content identification \cite{dasarathy2011reliability,dasarathy2014upper},
for which the exponents for vector-quantized codewords were derived
in \cite{merhav2017reliability}, \cite{merhav2018ensemble}, and
error exponents for an alternative model for a biometric identification
system, which is based on secret key generation and a helper messages
during the enrollment phase \cite[Chapter 2]{ignatenko2012biometric},
in \cite{merhav2018ensemble}. Error exponents for the bee identification
problem \cite{tandon2019bee} were derived in \cite{tamir2021error}. 

The error exponent of the typical random code was derived using the
TCEM in \cite{merhav2018error}, for a broad class of generalized
likelihood decoders. One of the consequences of this analysis is that
a general relation of the form $E_{\text{trc}}(R,P_{X})\leq E_{\text{ex}}(2R,P_{X})+R$
holds for any $R$ and generalized likelihood decoder (for ML decoding
and the BSC, a similar relation with equality sign was shown in \cite{barg2002random}).
A Gallager--style exponent was developed in \cite{merhav2019lagrange}.
The results were then extended to the colored Gaussian channel in
\cite{me19}, to random time-varying trellis codes in \cite{merhav2019error_trellis},
and to typical SW codes in \cite{tamir2021trade}. In \cite{tamir2021universal},
the TCEM was used to establish that a stochastic MMI decoder, which
is a universal decoder, achieves the exponent of the typical random
code and the expurgated exponent. Finally, the concentration of the
random error exponent to its mean value, the error exponent of the
typical random code, was derived using the TCEM in \cite{tamir2020large},
with refinements by Truong \emph{et al.\ }in \cite{truong2023concentration},
and then by Truong and Guill\'{e}n i F\`{a}bregas. In this last
result, the TCEM was used for codewords that are drawn in a dependent
manner, for an ensemble based on the Gilbert--Varshamov construction,
previously suggested by Somekh-Baruch, Scarlett and Guill\'{e}n i
F\`{a}bregas \cite{somekh2019generalized}.

\newpage{}

\section{Manipulating Expectations of Nonlinear Functions of Random Variables
\label{sec:Manipulating-Expectations-of}}

\subsection{Introduction}

An often encountered challenge in information-theoretic analytical
derivations involves the necessity to assess the expected value of
a non-linear function applied to either a RV or a random vector. The
conventional approach typically involves resorting to upper and lower
bounds for the sought-after expectation, with the hope that these
bounds are sufficiently accurate, at least for guiding us toward correct
conclusions. When dealing with a non-linear function that exhibits
convexity or concavity, it seems natural to employ Jensen's inequality,
which yields an upper or lower bound, respectively. However, it is
worth noting that this bound may not always prove precise enough to
serve our intended purposes.

The primary aim of this chapter is to introduce a range of alternative
tools that have proven their utility in prior research. These alternative
tools can be broadly categorized into two main categories.

In the first category (Sections \ref{subsec: irlog} and \ref{subsec: irpower}
below), the focus is on achieving\emph{ exact results}. Here, the
fundamental approach involves leveraging \emph{integral representations}
of the non-linear function under consideration. In the second category
(Sections \ref{subsec: jensenCoM}, \ref{subsec: reversejensen} and
\ref{subsec: jensenlike} below), we turn to bounding techniques,
but these bounds are designed to be more refined and precise than
what traditional applications of Jensen's inequality typically yield.
In some cases, these bounds even extend in the opposite direction,
offering a comprehensive exploration of the problem at hand.

To provide the reader with a swift comprehension of the concept of
an integral representation, as discussed in the first category mentioned
earlier, let us delve into a straightforward example. Imagine we have
a set of IID zero-mean Gaussian RVs, $X_{1},X_{2},\ldots,X_{n}$,
each with a variance of $\sigma^{2}$. Our task is to compute the
expected value of $\E\{1/\sum_{i=1}^{n}X_{i}^{2}\}$. At first glance,
this expectation might appear insurmountable to compute precisely.
However, let us consider the integral representation of the function
$f(s)=1/s$ as 
\begin{equation}
\frac{1}{s}=\int_{0}^{\infty}e^{-st}\dd t.
\end{equation}
The concept is to employ this representation to tackle the current
problem by rearranging the order between the expectation and the integration,
much akin to our approach in Chapter \ref{sec: laplacesaddlepoint}:
\begin{align}
\E\left\{ \frac{1}{\sum_{i=1}^{n}X_{i}^{2}}\right\}  & =\E\left\{ \int_{0}^{\infty}\exp\left(-s\sum_{i=1}^{n}X_{i}^{2}\right)\dd t\right\} \nonumber \\
 & =\int_{0}^{\infty}\E\left\{ \exp\left(-t\sum_{i=1}^{n}X_{i}^{2}\right)\right\} \dd t\nonumber \\
 & =\int_{0}^{\infty}\left[\E\left\{ \exp(-tX_{1}^{2})\right\} \right]^{n}\dd t\nonumber \\
 & =\int_{0}^{\infty}\frac{\dd t}{(1+2\sigma^{2}t)^{n/2}}\nonumber \\
 & =\begin{cases}
\infty, & n\leq2\\
\frac{1}{(n-2)\sigma^{2}}, & n>2
\end{cases}.
\end{align}
Certainly, the example provided is quite elementary, but it is important
to note that this concept can be applied in a wide range of scenarios,
involving the presentation of the given function as the Laplace transform
(or any other linear transform) of another function, and the expectation
of the given non-linear function is represented as an integral of
an expression that involves an expectation for which there is a closed-form
expression, like the MGF.

Another family of integral representations relates to the following
identity, which is applicable to any positive RV $X$ (and can be
readily extended to encompass any RV with a well-defined expectation):
\begin{equation}
\E\{X\}=\int_{0}^{\infty}\Pr\left\{ X\ge t\right\} \dd t.
\end{equation}
In fact, this idea has already been used in Example \ref{exa: extremevaluestatistics}
as well as in Chapter \ref{sec:TCE}. Accordingly, if $f$ is non-negative
and monotonic, and hence invertible, we have 
\begin{equation}
\E\{f(X)\}=\int_{0}^{\infty}\mbox{\ensuremath{\Pr}}\{f(X)\ge t\}\dd t=\int_{0}^{\infty}\mbox{\ensuremath{\Pr}}\{X\ge f^{-1}(t)\}\dd t,
\end{equation}
which is often lends itself to closed-form analysis.

In the first two upcoming sections, we will explore certain integral
representations of two specific highly important functions in the
context of information-theoretic analyses: The logarithmic function
(in Section \ref{subsec: irlog}, which is based on reference \cite{MS20a})
and the power function (in Section \ref{subsec: irpower}, which builds
on reference \cite{MS20b}). Those integral representations are not
very familiar to many researchers in the information theory community,
but the essence of the approach remains as described above: Substitute
the expectation of the given non-linear function with an integral
of a function for which a closed-form expectation exists.

In the second category of tools explored in this chapter, which focuses
on modified versions of the Jensen inequality, we delve into three
distinct types of bounding techniques: 
\begin{enumerate}
\item Jensen's inequality combined with a change of measures (Section \ref{subsec: jensenCoM}),
where our exposition relies strongly on \cite{AM98}, \cite{me11b}
and \cite{me18}. 
\item Reverse Jensen inequalities (Section \ref{subsec: reversejensen}),
which summarizes the main findings on \cite{me22}. 
\item Jensen-like inequalities, where the convex/concave function is only
part of the expression and the supporting line is re-optimized (Section
\ref{subsec: jensenlike}), which is based on \cite{me23a}. 
\end{enumerate}
While these techniques provide bounds rather than exact results, as
seen in the first category, their applicability extends across a broader
range of scenarios. Furthermore, they often yield substantial improvements
compared to the bounds derived from the conventional Jensen inequality.

\subsection{An Integral Representation of the Logarithmic Function \label{subsec: irlog}}

In the realm of analytic derivations within various information theory
problem domains, it is a recurring necessity to compute expectations
and higher-order moments of expressions involving the logarithm of
a positive-valued RV, or more broadly, the logarithm of the sum of
multiple such RVs. Traditionally, when faced with such scenarios,
two prevalent methods come into play: One is to employ upper and lower
bounds on the desired expression, often utilizing established inequalities
like Jensen's inequality or other widely recognized mathematical techniques;
the other approach involves applying the Taylor series expansion of
the logarithmic function. In more contemporary practices, a modernized
strategy has emerged, known as the replica method (as detailed in
\cite[Chapter 8]{MM09}). This method, while not strictly rigorous,
has gained prominence and proven effective, having been adopted from
the field of statistical physics with notable success.

The objective of this section is to introduce an alternative approach
and illustrate its practicality in commonly encountered scenarios.
Specifically, we will explore the following integral representation
of the logarithmic function, 
\begin{equation}
\ln x=\int_{0}^{\infty}\frac{e^{-u}-e^{-ux}}{u}\dd u,\quad x>0,\label{eq: ir}
\end{equation}
which is easily proved by substituting $\int_{0}^{\infty}e^{-ut}\dd t$
for $1/u$ on the RHS and interchanging the order of the integration.
This representation finds its immediate utility in scenarios where
the argument of the logarithmic function is a positive-valued RV denoted
as $X$, and our goal is to compute the expectation, denoted as $\E\{\ln X\}$.
By assuming the validity of interchanging the expectation operator
with the integration over the variable $u$, we can simplify the calculation
of $\E\{\ln X\}$ into evaluating the MGF of $X$, which is often
a more straightforward task. This transformation allows us to express
it as: 
\begin{equation}
\E\left\{ \ln X\right\} =\int_{0}^{\infty}\left[e^{-u}-\E\{e^{-uX}\}\right]\frac{\dd u}{u}.\label{eq: ElnX}
\end{equation}
In particular, if $X_{1},\ldots,X_{n}$ are positive IID RVs, then
\begin{equation}
\E\left\{ \ln(X_{1}+\ldots+X_{n})\right\} =\int_{0}^{\infty}\left(e^{-u}-\bigl[\E\{e^{-uX_{1}}\}\bigr]^{n}\right)\frac{\dd u}{u}.\label{eq: Elnsum}
\end{equation}
This concept is not entirely novel, as it has been previously applied
in the realm of physics, as evidenced in sources such as \cite[Eq.\ (2.4) and beyond]{EN93},
\cite[Exercise 7.6, p.~140]{MM09}, and \cite[Eq.\ (12) and beyond]{SSRCM19}.
However, in the field of information theory, this approach is seldom
utilized, despite its potential significance. This significance arises
from the frequent requirement to compute logarithmic expectations
-- a common occurrence in numerous problem areas within information
theory. Furthermore, the integral representation \eqref{eq: ir} extends
its utility beyond mere expectation calculations; it also proves invaluable
in evaluating higher moments of $\ln X$, most notably, the second
moment or variance. This added functionality allows us to assess statistical
fluctuations around the mean, enhancing our analytical capabilities
in the field.

In \cite{MS20a}, the practicality of this approach was effectively
showcased across various application domains. These applications encompassed
areas such as entropy and differential entropy assessments, performance
analysis of universal lossless source codes, and the determination
of ergodic capacity for the Rayleigh SIMO channel. It is worth noting
that within some of these examples, we successfully computed variances
related to the pertinent RVs. In particular, in \cite[Proposition 2]{MS20a},
the following result is stated and proved: Let $X$ be a RV, and let
\begin{equation}
M_{X}(s)\triangleq\E\left\{ e^{sX}\right\} ,~~~~~~~s\in\reals,\label{eq: g}
\end{equation}
be the MGF of $X$. If $X\ge0$ with probability one, then 
\begin{equation}
\E\left\{ \ln(1+X)\right\} =\int_{0}^{\infty}\frac{e^{-u}\,[1-M_{X}(-u)]}{u}\;\mathrm{\dd}u,\label{eq: 2511a}
\end{equation}
and 
\begin{equation}
\mathrm{\Var}\left\{ \ln(1+X)\right\} =\int_{0}^{\infty}\int_{0}^{\infty}\frac{e^{-(u+v)}}{uv}\;\Bigl[M_{X}(-u-v)-M_{X}(-u)\,M_{X}(-v)\Bigr]\mathrm{\dd}u\mathrm{\dd}v.\label{eq: 2511b}
\end{equation}

It is worth highlighting an intriguing consequence of the integral
representation \eqref{eq: ElnX}. It transforms the calculation of
the expectation of the logarithm of $X$ into the expectation of an
exponential function of $X$. This transformation has an added benefit:
It simplifies expressions involving quantities like $\ln(n!)$ into
the integral of a summation of a geometric series, a form that is
readily expressible in closed form. Specifically,
\begin{align}
\ln(n!) & =\sum_{k=1}^{n}\ln k\nonumber \\
 & =\sum_{k=1}^{n}\int_{0}^{\infty}(e^{-u}-e^{-uk})\frac{\mathrm{\dd}u}{u}\nonumber \\
 & =\int_{0}^{\infty}\left(ne^{-u}-\sum_{k=1}^{n}e^{-uk}\right)\frac{\mathrm{\dd}u}{u}\nonumber \\
 & =\int_{0}^{\infty}e^{-u}\left(n-\frac{1-e^{-un}}{1-e^{-u}}\right)\frac{\mathrm{\dd}u}{u}.
\end{align}
For a positive integer-valued RV, denoted as $N$, the computation
of $\E\{\ln N!\}$ becomes a straightforward task, requiring only
the calculation of $\E\{N\}$ and the MGF, $\E\{e^{-uN}\}$. This
is useful, for example, when $N$ follows a Poisson distribution,
as shown in \cite{MS20a} in detail.

In \cite{MS20a}, the usefulness of the integral representation of
the logarithmic function is illustrated in several problem areas in
information theory, including graphs of numerical results. Here, we
briefly summarize two of the examples provided therein.

\subsubsection{Differential Entropy for Generalized Multivariate Cauchy Densities
\label{subsec:Cauchy}}

Let $(X_{1},\ldots,X_{n})$ be a random vector whose PDFs is of the
form 
\begin{equation}
f(x_{1},\ldots,x_{n})=\frac{C_{n}}{\left[1+\sum_{i=1}^{n}g(x_{i})\right]^{q}},~~~~(x_{1},\ldots,x_{n})\in\reals^{n},\label{eq: pdf gen Cauchy}
\end{equation}
for a given non-negative function $g$ and positive real $q$ such
that 
\begin{equation}
\int_{\reals^{n}}\frac{\mathrm{\dd}\boldsymbol{x}}{\left[1+\sum_{i=1}^{n}g(x_{i})\right]^{q}}<\infty.
\end{equation}
We term this category of density as ``generalized multivariate Cauchy,''
primarily because the multivariate Cauchy density arises as a specific
instance when $g(x)=x^{2}$ and $q=\tfrac{1}{2}(n+1)$. Employing
the Laplace transform relation, 
\begin{equation}
\frac{1}{s^{q}}=\frac{1}{\Gamma(q)}\int_{0}^{\infty}t^{q-1}e^{-st}\mathrm{\dd}t,~~~~~~q\ge1,\;\real(s)>0,
\end{equation}
$f$ can be displayed as a mixture of product-form PDFs: 
\begin{align}
f(x_{1},\ldots,x_{n}) & =\frac{C_{n}}{\left[1+\sum_{i=1}^{n}g(x_{i})\right]^{q}}\nonumber \\
 & =\frac{C_{n}}{\Gamma(q)}\int_{0}^{\infty}t^{q-1}e^{-t}\,\exp\left\{ -t\sum_{i=1}^{n}g(x_{i})\right\} \dd t.\label{eq: 2911a1}
\end{align}
Defining 
\begin{equation}
Z(t)\triangleq\int_{-\infty}^{\infty}e^{-tg(x)}\mathrm{\dd}x,\quad\forall t>0,\label{eq: Z}
\end{equation}
we obtain from \eqref{eq: 2911a1}, 
\begin{align}
1 & =\frac{C_{n}}{\Gamma(q)}\int_{0}^{\infty}t^{q-1}e^{-t}\int_{\reals^{n}}\exp\left\{ -t\sum_{i=1}^{n}g(x_{i})\right\} \dd x_{1}\ldots\dd x_{n}\dd t\nonumber \\
 & =\frac{C_{n}}{\Gamma(q)}\int_{0}^{\infty}t^{q-1}e^{-t}\left(\int_{-\infty}^{\infty}e^{-tg(x)}\dd x\right)^{n}\dd t\nonumber \\
 & =\frac{C_{n}}{\Gamma(q)}\int_{0}^{\infty}t^{q-1}e^{-t}Z^{n}(t)\dd t,
\end{align}
and so, 
\begin{equation}
C_{n}=\frac{\Gamma(q)}{{\displaystyle \int_{0}^{\infty}t^{q-1}e^{-t}Z^{n}(t)\dd t}}.\label{eq: C_n}
\end{equation}
Evaluating the differential entropy of $f$ involves deriving $\E\{\ln\bigl[1+\sum_{i=1}^{n}g(X_{i})\bigr]\}$.
Using \eqref{eq: 2511a} 
\begin{equation}
\E\Biggl\{\ln\left[1+\sum_{i=1}^{n}g(X_{i})\right]\Biggr\}=\int_{0}^{\infty}\frac{e^{-u}}{u}\left(1-\E\left\{ \exp\left[-u\sum_{i=1}^{n}g(X_{i})\right]\right\} \right)\dd u,
\end{equation}
and 
\begin{align}
 & \E\left\{ \exp\left[-u\sum_{i=1}^{n}g(X_{i})\right]\right\} \nonumber \\
 & =\frac{C_{n}}{\Gamma(q)}\int_{0}^{\infty}t^{q-1}e^{-t}\int_{\reals^{n}}\exp\left\{ -(t+u)\sum_{i=1}^{n}g(x_{i})\right\} \dd x_{1}\ldots\dd x_{n}\dd t\nonumber \\
 & =\frac{C_{n}}{\Gamma(q)}\int_{0}^{\infty}t^{q-1}e^{-t}Z^{n}(t+u)\dd t.
\end{align}
Thus, the joint differential entropy is given by 
\begin{align}
h(X_{1},\ldots,X_{n}) & =q\cdot\E\left\{ \ln\left[1+\sum_{i=1}^{n}g(X_{i})\right]\right\} -\ln C_{n}\nonumber \\[0.1cm]
 & =q\int_{0}^{\infty}\frac{e^{-u}}{u}\left(1-\frac{C_{n}}{\Gamma(q)}\int_{0}^{\infty}t^{q-1}e^{-t}Z^{n}(t+u)\dd t\right)\dd u-\ln C_{n}\nonumber \\[0.1cm]
 & =\frac{qC_{n}}{\Gamma(q)}\int_{0}^{\infty}\int_{0}^{\infty}\frac{t^{q-1}e^{-(t+u)}}{u}\;\Bigl[Z^{n}(t)-Z^{n}(t+u)\Bigr]\dd t\dd u-\ln C_{n}.\label{eq: diff ent. gen. cauchy}
\end{align}
For $g(x)=|x|^{\theta}$, with an arbitrary $\theta>0$, we obtain
from \eqref{eq: Z} that 
\begin{equation}
Z(t)=\frac{2\cdot\Gamma(1/\theta)}{\theta\cdot t^{1/\theta}}.\label{eq:spec. Z}
\end{equation}
In particular, for $\theta=2$ and $q=\tfrac{1}{2}(n+1)$, we get
the multivariate Cauchy density from \eqref{eq: pdf gen Cauchy}.
In this case, since $\Gamma\bigl(\tfrac{1}{2}\bigr)=\sqrt{\pi}$,
it follows from \eqref{eq:spec. Z} that $Z(t)=\sqrt{\frac{\pi}{t}}$
for $t>0$, and from \eqref{eq: C_n} 
\begin{equation}
C_{n}=\frac{\Gamma\left(\frac{n+1}{2}\right)}{\pi^{n/2}\,{\displaystyle \int_{0}^{\infty}t^{(n+1)/2-1}e^{-t}\,t^{-n/2}\dd t}}=\frac{\Gamma\left(\frac{n+1}{2}\right)}{\pi^{n/2}\,\Gamma\bigl(\tfrac{1}{2}\bigr)}=\frac{\Gamma\left(\frac{n+1}{2}\right)}{\pi^{(n+1)/2}}.\label{eq: cauchynormalization}
\end{equation}
Combining \eqref{eq: diff ent. gen. cauchy}, \eqref{eq:spec. Z}
and \eqref{eq: cauchynormalization} gives 
\begin{align}
h(X_{1},\ldots,X_{n}) & =\frac{n+1}{2\pi^{(n+1)/2}}\int_{0}^{\infty}\int_{0}^{\infty}\frac{e^{-(t+u)}}{u\sqrt{t}}\left[1-\left(\frac{t}{t+u}\right)^{n/2}\right]\dd t\dd u\nonumber \\[0.1cm]
 & \hphantom{===}+\frac{(n+1)\ln\pi}{2}-\ln\Gamma\left(\frac{n+1}{2}\right).\label{diff ent. gen. cauchy 2}
\end{align}
In this application example, we find it intriguing that \eqref{eq: diff ent. gen. cauchy}
offers what can be considered a ``single-letter expression.'' Remarkably,
the $n$-dimensional integral tied to the original expression of the
differential entropy $h(X_{1},\ldots,X_{n})$ is effectively replaced
by the two-dimensional integral in \eqref{eq: diff ent. gen. cauchy},
and notably, this replacement remains independent of the value of
$n$.

\subsubsection{Ergodic Capacity of the Rayleigh SIMO Channel \label{subsec: SIMO}}

Let us consider the SIMO channel with $L$ receive antennas. We make
the assumption that the channel transfer coefficients, denoted as
$h_{1},h_{2},\ldots,h_{L}$, are independent and follow a zero-mean,
circularly symmetric complex Gaussian distribution with variances
$\sigma_{1}^{2},\sigma_{2}^{2},\ldots,\sigma_{L}^{2}$. In this context,
the ergodic capacity of the SIMO channel, measured in nats per channel
use, is expressed as follows: 
\begin{equation}
C=\E\left\{ \ln\left(1+\rho\sum_{\ell=1}^{L}|h_{\ell}|^{2}\right)\right\} =\E\left\{ \ln\left(1+\rho\sum_{\ell=1}^{L}\bigl(f_{\ell}^{2}+g_{\ell}^{2}\bigr)\right)\right\} ,\label{eq: capacity SIMO}
\end{equation}
where $f_{\ell}:=\real\{h_{\ell}\}$, $g_{\ell}:=\imag\{h_{\ell}\}$,
and $\rho:=\frac{P}{N_{0}}$ is the signal-to-noise ratio (SNR). In
view of \eqref{eq: 2511a}, let 
\begin{equation}
X\triangleq\rho\,\sum_{\ell=1}^{L}(f_{\ell}^{2}+g_{\ell}^{2}).
\end{equation}
For all $u>0$, 
\begin{align}
M_{X}(-u) & =\E\biggl\{\exp\biggl(-\rho u\,\sum_{\ell=1}^{L}(f_{\ell}^{2}+g_{\ell}^{2})\biggr)\biggr\}\nonumber \\[0.1cm]
 & =\prod_{\ell=1}^{L}\biggl\{\E\Bigl\{ e^{-u\rho f_{\ell}^{2}}\Bigr\}\;\E\Bigl\{ e^{-u\rho g_{\ell}^{2}}\Bigr\}\biggr\}\nonumber \\[0.1cm]
 & =\prod_{\ell=1}^{L}\frac{1}{1+u\rho\sigma_{\ell}^{2}},\label{eq: MGF1}
\end{align}
where \eqref{eq: MGF1} holds since 
\begin{align}
\E\left\{ e^{-u\rho f_{\ell}^{2}}\right\}  & =\E\left\{ e^{-u\rho g_{\ell}^{2}}\right\} \nonumber \\[0.1cm]
 & =\int_{-\infty}^{\infty}\frac{\dd w}{\sqrt{\pi\sigma_{\ell}^{2}}}\,e^{-w^{2}/\sigma_{\ell}^{2}}\,e^{-u\rho w^{2}}\nonumber \\[0.1cm]
 & =\frac{1}{\sqrt{1+u\rho\sigma_{\ell}^{2}}}.
\end{align}
From \eqref{eq: 2511a}, \eqref{eq: capacity SIMO} and \eqref{eq: MGF1},
the ergodic capacity (in nats per channel use) is given by 
\begin{align}
C & =\E\left\{ \ln\left(1+\rho\sum_{\ell=1}^{L}\bigl(f_{\ell}^{2}+g_{\ell}^{2}\bigr)\right)\right\} \nonumber \\[0.1cm]
 & =\int_{0}^{\infty}\frac{e^{-u}}{u}\left(1-\prod_{\ell=1}^{L}\frac{1}{1+u\rho\sigma_{\ell}^{2}}\right)\dd u\nonumber \\[0.1cm]
 & =\int_{0}^{\infty}\frac{e^{-x/\rho}}{x}\left(1-\prod_{\ell=1}^{L}\frac{1}{1+\sigma_{\ell}^{2}x}\right)\dd x.\label{eq: erg. C}
\end{align}
Concerning the variance, owing to \eqref{eq: 2511b} and \eqref{eq: MGF1},
we have: 
\begin{align}
 & \mathrm{\Var}\left\{ \ln\left(1+\rho\sum_{\ell=1}^{L}[f_{\ell}^{2}+g_{\ell}^{2}]\right)\right\} \nonumber \\[0.1cm]
 & =\int_{0}^{\infty}\int_{0}^{\infty}\frac{e^{-(x+y)/\rho}}{xy}\,\left\{ \prod_{\ell=1}^{L}\frac{1}{1+\sigma_{\ell}^{2}(x+y)}-\prod_{\ell=1}^{L}\biggl[\frac{1}{(1+\sigma_{\ell}^{2}x)(1+\sigma_{\ell}^{2}y)}\biggr]\right\} \dd x\dd y.\label{var SIMO}
\end{align}
The capacity $C$ can be expressed as a linear combination of integrals
of the form 
\begin{align}
\int_{0}^{\infty}\frac{e^{-x/\rho}\dd x}{1+\sigma_{\ell}^{2}x} & =\frac{1}{\sigma_{\ell}^{2}}\int_{0}^{\infty}\frac{e^{-t}\dd t}{t+1/(\sigma_{\ell}^{2}\rho)}\nonumber \\[0.1cm]
 & =\frac{e^{1/(\sigma_{\ell}^{2}\rho)}}{\sigma_{\ell}^{2}}\int_{1/(\sigma_{\ell}^{2}\rho)}^{\infty}\frac{e^{-s}}{s}\dd s\nonumber \\[0.1cm]
 & =\frac{1}{\sigma_{\ell}^{2}}\cdot e^{1/(\sigma_{\ell}^{2}\rho)}\cdot E_{1}\biggl(\frac{1}{\sigma_{\ell}^{2}\rho}\biggr),
\end{align}
where $E_{1}(\cdot)$ is the (modified) exponential integral function,
defined as 
\begin{equation}
E_{1}(x):=\int_{x}^{\infty}\frac{e^{-s}}{s}\dd s,\quad\forall\,x>0.\label{eq: E1 func.}
\end{equation}

\subsection{An Integral Representation of the Power Function \label{subsec: irpower}}

In this section, we build upon the same approach as in Section \ref{subsec: irlog},
expanding the scope to introduce an integral representation for a
general moment of a non-negative RV, $X$. Specifically, we aim to
find an expression for $\E\{X^{\rho}\}$ where $\rho>0$. When $\rho$
is an integer, it is well-known that this moment can be computed as
the $\rho$-th order derivative of the MGF of $X$, evaluated at the
origin. However, our proposed integral representation, presented in
this work, applies to any non-integer positive value of $\rho$. Here
as well, it replaces the direct calculation of $\E\{X^{\rho}\}$ with
the integration of an expression involving the MGF of $X$. We refer
to this representation as an ``extension'' of the integral representation
of the logarithmic function discussed in Section \ref{subsec: irlog}.
This is because the latter can be derived as a special case of the
formula for $\E\{X^{\rho}\}$ by employing the identity: 
\begin{equation}
\E\left\{ \ln X\right\} =\lim_{\rho\to0}\frac{\E\{X^{\rho}\}-1}{\rho},
\end{equation}
or alternatively, the identity, 
\begin{equation}
\E\left\{ \ln X\right\} =\lim_{\rho\to0}\frac{\ln\left[\E\{X^{\rho}\}\right]}{\rho}.
\end{equation}
As in the previous section, here too, the proposed integral representation
is applied to a range of examples motivated by information theory
\cite{MS20b}. This application showcases how the representation streamlines
numerical evaluations. In particular, much like the case of the logarithmic
function, when employed to compute a moment of the sum of a large
number, denoted as $n$, of non-negative RVs, it becomes evident that
integration over one or two dimensions, as suggested by our integral
representation, is notably simpler than the alternative of integrating
over $n$ dimensions, as required in the direct calculation of the
desired moment. Additionally, single or double-dimensional integrals
can be promptly and accurately computed using built-in numerical integration
techniques.

In order to present the integral representation, we commence by defining
the Gamma and Beta functions as follows: 
\begin{equation}
\Gamma(u)\triangleq\int_{0}^{\infty}t^{u-1}e^{-t}\dd t,\quad u>0,\label{eq: Gamma}
\end{equation}
\begin{equation}
B(u,v):=\int_{0}^{1}t^{u-1}(1-t)^{v-1}\dd t=\frac{\Gamma(u)\,\Gamma(v)}{\Gamma(u+v)},\quad u,v>0.\label{eq: Beta}
\end{equation}
Let $X$ be a non-negative RV with an MGF $M_{X}(\cdot)$, and let
$\rho>0$ be a non-integer real. Then, as shown in \cite{MS20b},
\begin{align}
\E\{X^{\rho}\} & =\frac{1}{1+\rho}\sum_{\ell=0}^{\lfloor\rho\rfloor}\frac{\alpha_{\ell}}{B(\ell+1,\rho+1-\ell)}\nonumber \\
 & \hphantom{==}+\frac{\rho\cdot\sin(\pi\rho)\cdot\Gamma(\rho)}{\pi}\int_{0}^{\infty}\frac{1}{u^{\rho+1}}\Biggl(\sum_{j=0}^{\lfloor\rho\rfloor}\biggl\{\frac{(-1)^{j}\cdot\alpha_{j}}{j!}\cdot u^{j}\biggr\} e^{-u}-M_{X}(-u)\Biggr)\dd u,\label{eq: rho>0 non-int.}
\end{align}
where for all $j\in\{0,1,\ldots,\}$ 
\begin{align}
\alpha_{j} & \triangleq\E\left\{ (X-1)^{j}\right\} \label{eq: alpha_k}\\
 & =\frac{1}{j+1}\sum_{\ell=0}^{j}\frac{(-1)^{j-\ell}\cdot M_{X}^{(\ell)}(0)}{B(\ell+1,j-\ell+1)}.\label{eq: alpha_k2}
\end{align}
The proof of \eqref{eq: rho>0 non-int.} in \cite{MS20b} does not
apply to natural values $\rho$ (see \cite[Appendix A]{MS20b}, where
the denominators vanish). However, taking a limit in \eqref{eq: rho>0 non-int.}
where we let $\rho$ tend to an integer, and applying L'H\^{o}pital's
rule, one can reproduce the well-known result for integer $\rho$,
which is given in terms of the $\rho$-th order derivative of the
MGF at the origin. For $\rho\in(0,1)$, the above simplifies to: 
\begin{align}
\E\{X^{\rho}\} & =1+\frac{\rho}{\Gamma(1-\rho)}\int_{0}^{\infty}\frac{e^{-u}-M_{X}(-u)}{u^{1+\rho}}\dd u.\label{eq: rho-in-(0,1)}
\end{align}

In \cite{MS20b}, the profound utility of the integral representation
shines through in a comprehensive exploration across various domains
within information theory and statistics. These applications include
detailed investigations accompanied by graphical illustrations. The
showcased instances span a range of analytical inquiries, encompassing
randomized guessing, estimation errors, the R\'{e}nyi entropy of
$n$-dimensional generalized Cauchy distributions, and mutual information
calculations for channels featuring a specific jammer model. Here,
we will provide a succinct overview of two of these application examples,
focusing primarily on the easier scenario where $\rho\in(0,1)$ for
clarity and simplicity of exposition.

\subsubsection{Moments of Guesswork \label{subsec: randomaized guessing}}

Suppose we have a RV $X$ that assumes values from a finite alphabet
${\cal X}$. Let us explore a random guessing strategy wherein the
guesser submits a sequence of independent random guesses, drawn from
a specific probability distribution denoted as $\widetilde{P}(\cdot)$,
defined over ${\cal X}$. Consider any instance where $x\in{\cal X}$
represents a realization of $X$, and we have the guessing distribution
$\widetilde{P}$ at our disposal. In such a scenario, the RV $G$,
representing the number of independent guesses required to achieve
success, follows a geometric distribution:
\begin{equation}
\mbox{\ensuremath{\Pr}}\{G=k|x\}=\bigl[1-\widetilde{P}(x)\bigr]^{k-1}\,\widetilde{P}(x),\label{eq: 20200115a1}
\end{equation}
hence, the corresponding MGF is equal to 
\begin{align}
M_{G}(u|x) & =\sum_{k=1}^{\infty}e^{ku}\,\mbox{\ensuremath{\Pr}}\{G=k|x\}\nonumber \\
 & =\frac{\widetilde{P}(x)}{e^{-u}-\bigl(1-\widetilde{P}(x)\bigr)},\quad u<\ln\frac{1}{1-\widetilde{P}(x)}.\label{eq: 20200115a2}
\end{align}
For $\rho\in(0,1)$, it is shown in \cite{MS20b} that 
\begin{align}
\E\left\{ G^{\rho}|x\right\}  & =1+\frac{\rho}{\Gamma(1-\rho)}\int_{0}^{\infty}\frac{e^{-u}-e^{-2u}}{u^{\rho+1}\bigl[\bigl(1-\widetilde{P}(x)\bigr)^{-1}-e^{-u}\bigr]}\dd u.\label{eq: 20200116a1}
\end{align}
Consider the distribution of the RV $X$, denoted as $P$. To compute
the unconditional $\rho$-th moment using \eqref{eq: 20200116a1},
we average over all possible values of $X$. This yields the following
result for all $\rho$ in the open interval $(0,1)$: 
\begin{equation}
\E\left\{ G^{\rho}\right\} =1+\frac{\rho}{\Gamma(1-\rho)}\int_{0}^{1}\frac{1-z}{(-\ln z)^{\rho+1}}\sum_{x\in{\cal X}}\frac{P(x)\bigl(1-\widetilde{P}(x)\bigr)}{1-z\bigl(1-\widetilde{P}(x)\bigr)}\dd z,\label{eq: 20200509a1}
\end{equation}
where \eqref{eq: 20200509a1} is by changing the integration variable
according to $z=e^{-u}$.

In conclusion, equation \eqref{eq: 20200116a1} provides a computable
one-dimensional integral expression for the $\rho$-th guessing moment
for any $\rho>0$. This eliminates the necessity for numerical computations
involving infinite sums.

\subsubsection{Moments of Estimation Errors \label{subsec: estimation}}

Let $X_{1},\ldots,X_{n}$ be IID RVs with an unknown expectation $\theta$
to be estimated, and consider the simple sample-mean estimator, 
\begin{equation}
\widehat{\theta}_{n}=\frac{1}{n}\sum_{i=1}^{n}X_{i}.\label{eq: estimator}
\end{equation}
For $\rho\in(0,2)$, it is shown in \cite{MS20b} that 
\begin{align}
 & \E\left\{ \bigl|\widehat{\theta}_{n}-\theta\bigr|^{\rho}\right\} \nonumber \\
 & =1+\frac{\rho}{2\,\Gamma(1-\frac{1}{2}\,\rho)}\int_{0}^{\infty}\int_{-\infty}^{\infty}u^{-(\rho/2+1)}\Bigl[\tfrac{1}{2}\,e^{-u-|\omega|}-\frac{1}{2\sqrt{\pi u}}\;\phi_{X}^{\,n}\Bigl(\frac{\omega}{n}\Bigr)\,e^{-j\omega\theta-\omega^{2}/(4u)}\Bigr]\dd\omega\dd u,\label{eq: 20200111a1}
\end{align}
which is an exact, double-integral calculable expression for the $\rho$-th
moment of the estimation error of the expectation of $n$ IID RVs.

\subsection{Jensen's Inequality with a Change of Measure \label{subsec: jensenCoM}}

In this section, we advocate for the practical utility of combining
Jensen's inequality with a change of measure, effectively introducing
an additional degree of freedom for optimization. To illustrate this
concept concretely, let us consider a concave function $f$ and a
RV $X$ characterized by a PDF $p$, with its support set in ${\cal X}$.
Additionally, let $q$ represent another PDF, also with support in
${\cal X}$. We will use $\E_{p}\{\cdot\}$ and $\E_{q}\{\cdot\}$
to denote expectation operators WRT $p$ and $q$, respectively. Now,
let us delve into the following chain of simple inequalities: 
\begin{align}
f\left(\E_{p}\{X\}\right) & =f\left(\int_{{\cal X}}p(x)x\dd x\right)\nonumber \\
 & =f\left(\int_{{\cal X}}q(x)\cdot\frac{xp(x)}{q(x)}\dd x\right)\nonumber \\
 & \ge\int_{{\cal X}}q(x)f\left(\frac{xp(x)}{q(x)}\right)\nonumber \\
 & =\E_{q}\left\{ f\left(\frac{Xp(X)}{q(X)}\right)\right\} .
\end{align}
Given that the inequalities mentioned above are valid for any PDF
$q$ supported by ${\cal X}$, we possess the flexibility to maximize
the rightmost side of this chain, which can be expressed as: 
\begin{equation}
f\left(\E_{p}\{X\}\right)\ge\sup_{q\in{\cal Q}}\E_{q}\left\{ f\left(\frac{Xp(X)}{q(X)}\right)\right\} ,\label{eq: CoM}
\end{equation}
where ${\cal Q}$ is any class of PDFs with this support. Clearly,
when $p$ belongs to the set ${\cal Q}$, the choice of $q=p$ reduces
the inequality in (\ref{eq: CoM}) to the standard Jensen's inequality.
On the opposite end of the spectrum, if ${\cal Q}$ encompasses the
entire collection of PDFs over ${\cal X}$, and if $X$ is a positive
RV with $\E_{p}\{X\}<\infty$, then selecting $q(x)=\frac{xp(x)}{\E_{p}\{X\}}$
results in a trivial and uninformative identity. However, this highlights
that in such a scenario, the inequality in (\ref{eq: CoM}) essentially
becomes an equality, depicted as: 
\begin{equation}
f\left(\E_{p}\{X\}\right)=\sup_{\{q\colon\supp\{q\}={\cal X}\}}\E_{q}\left\{ f\left(\frac{Xp(X)}{q(X)}\right)\right\} .\label{eq: CoM1}
\end{equation}
In the sequel, we abbreviate suprema and infima over $\{q\colon\mbox{supp}\{q\}={\cal X}\}$
simply by writing $\sup_{q}$ and $\inf_{q}$, respectively. Likewise,
if $f$ is convex, we have 
\begin{equation}
f\left(\E_{p}\{X\}\right)\le\inf_{q\in{\cal Q}}\E_{q}\left\{ f\left(\frac{Xp(X)}{q(X)}\right)\right\} .
\end{equation}
The effectiveness of these inequalities hinges on our judicious selection
of ${\cal Q}$. As we have observed, ${\cal Q}$ should encompass
$p$ to ensure that the resultant bound, after optimizing over $q\in{\cal Q}$,
does not fall short of the standard Jensen's inequality. Conversely,
${\cal Q}$ should exclude the choice $q(x)=\frac{xp(x)}{\E_{p}\{X\}}$,
which renders the inequality uninformative. Ideally, the class ${\cal Q}$
should be well-suited for practical use, allowing for closed-form
optimization. This convenience would enable us to derive bounds that
significantly improve upon the standard Jensen's inequality, making
the approach both mathematically tractable and practically valuable.

Perhaps the most important special case of \eqref{eq: CoM1} pertains
to the case of $f(x)=\ln x$, where it becomes 
\begin{align}
\ln(\E_{p}\{X\}) & =\sup_{q}\E_{q}\left\{ \ln\left(\frac{Xp(X)}{q(X)}\right)\right\} \nonumber \\
 & =\sup_{q}\left\{ \E_{q}\{f(X)\}-D(q\|p)\right\} ,\label{eq: lncase}
\end{align}
which is intimately related to the Laplace principle \cite{DE97}
in large-deviations theory, or more generally, to Varadhan's integral
lemma \cite[Section 4.3]{DZ93} or the Donsker-Varadhan variational
principle.
\begin{example}
\emph{Let $U_{1},\ldots,U_{n}$ be drawn from a finite-alphabet memoryless
source $P$ and let $X=\exp\{\alpha\ell(U_{1},\ldots,U_{n})\}$, where
$\alpha>0$ is a given real parameter and $\ell(U_{1},\ldots,U_{n})$
is the length (in nats) of the compressed version of $(U_{1},\ldots,U_{n})$
under some given fixed-to-variable length lossless source code. Now,
a naive application of Jensen's inequality yields 
\begin{equation}
\E\{X\}=\E\left\{ \exp[\alpha\ell(U_{1},\ldots,U_{n})]\right\} \ge\exp\left\{ \alpha\E\{\ell(U_{1},\ldots,U_{n})\right\} \ge e^{\alpha nH(P)},
\end{equation}
where $H(P)$ is the per-symbol entropy of the source $P$. On the
other hand, considering $P^{n}$ and $Q^{n}$ as probability distributions
of $n$-vectors from the source, we have 
\begin{align}
\ln\left(\E_{P^{n}}\{X\}\right) & \ge\sup_{Q^{n}\in{\cal Q}}\left[\E_{Q^{n}}\{\ln X\}-D(Q^{n}\|P^{n})\right]\nonumber \\
 & =\sup_{Q^{n}\in{\cal Q}}\left[\alpha\E_{Q^{n}}\{\ell(U_{1},\ldots,U_{n})\}-D(Q^{n}\|P^{n})\right]\nonumber \\
 & \ge\sup_{Q^{n}\in{\cal Q}}\left[\alpha H(Q^{n})-D(Q^{n}\|P^{n})\right].
\end{align}
Now, rather than taking ${\cal Q}$ to be the class of all probability
distributions of $n$-vectors, let us take it to be the class of all
product form distributions, i.e., $Q^{n}(u_{1},\ldots,u_{n})=\prod_{i=1}^{n}Q(u_{i})$.
Since $P^{n}$ has a product form too, i.e., $P^{n}(u_{1},\ldots,u_{n})=\prod_{i=1}^{n}P(u_{i})$,
we readily obtain that the last expression reads 
\begin{equation}
\sup_{Q^{n}\in{\cal Q}}\left[\alpha H(Q^{n})-D(Q^{n}\|P^{n})\right]=n\cdot\sup_{Q}[\alpha H(Q)-D(Q\|P)],
\end{equation}
that yields the R\'{e}nyi entropy of order $\alpha$ pertaining to
$P$, which is an attainable lower bound to the exponential moment
of $\ell(U_{1},\ldots,U_{n})$, unlike the lower bound obtained from
the naive use of Jensen's inequality above. In other words, rather
than maximizing over the entire class of all probability distributions
of $n$-vectors, $\{Q^{n}\}$, we observe that the much smaller class
of memoryless probability distributions is large enough to obtain
a tight result. The same idea was used also in the converse part of
\cite{AM98} in the context of guessing, which is strongly related
to source coding.}
\end{example}
The identity \eqref{eq: lncase} has found extensive utility, not
only in this context but also in previous works such as \cite{me11b},
where it was applied to exponential moments of various loss functions,
and \cite{me18}, where it played a crucial role in establishing lower
bounds on exponential moments of estimation errors. Numerous references
within these two articles further emphasize the importance of \eqref{eq: lncase}.
However, it is vital to highlight a key takeaway message from this
section: The relation \eqref{eq: CoM} is not limited to the logarithmic
function alone; it holds true for any concave function (or convex
function with appropriate adjustments) and extends its applicability
beyond just the logarithmic case.
\begin{example}
\emph{To demonstrate another special case of combining Jensen's inequality
with a change of measure, consider the example of deriving an upper
bound to the expectation of the harmonic mean of $n$ positive RVs,
$X_{1},\ldots,X_{n}$, i.e., 
\begin{equation}
\E\left\{ \frac{n}{\sum_{i=1}^{n}1/X_{i}}\right\} .
\end{equation}
This expectation cannot be upper bounded by a direct application of
Jensen's inequality, because it provides a lower bound, 
\begin{equation}
\E\left\{ \frac{n}{\sum_{i=1}^{n}1/X_{i}}\right\} \ge\frac{n}{\sum_{i=1}^{n}\E\left\{ 1/X_{i}\right\} },
\end{equation}
rather than an upper bound, and moreover, it requires the expectations
of $1/X_{i}$ rather than those of $X_{i}$. However, consider the
following approach: Let $q=(q_{1},\ldots q_{n})$ be an arbitrary
probability vector, i.e., a set of $n$ positive numbers summing to
unity. Then, 
\begin{align}
\sum_{i=1}^{n}\frac{1}{X_{i}} & =\sum_{i=1}^{n}q_{i}\cdot\frac{1}{q_{i}X_{i}}\nonumber \\
 & \ge\frac{1}{\sum_{i=1}^{n}q_{i}\cdot(q_{i}X_{i})}\nonumber \\
 & =\frac{1}{\sum_{i=1}^{n}q_{i}^{2}X_{i}},
\end{align}
where the inequality stems from the (ordinary) Jensen inequality applied
to the convex function $f(u)=1/u$. Equivalently, 
\begin{equation}
\frac{n}{\sum_{i=1}^{n}1/X_{i}}\le n\cdot\sum_{i=1}^{n}q_{i}^{2}X_{i}.
\end{equation}
Since this inequality holds for every probability vector $q$, we
may minimize the RHS over $q$, to obtain 
\begin{equation}
\frac{n}{\sum_{i=1}^{n}1/X_{i}}\le n\cdot\min_{q}\sum_{i=1}^{n}q_{i}^{2}X_{i}.
\end{equation}
Taking the expectations of both sides, we get: 
\begin{align}
\E\left\{ \frac{n}{\sum_{i=1}^{n}1/X_{i}}\right\}  & \le n\cdot\E\left\{ \min_{q}\sum_{i=1}^{n}q_{i}^{2}X_{i}\right\} \nonumber \\
 & \le n\cdot\min_{q}\E\left\{ \sum_{i=1}^{n}q_{i}^{2}X_{i}\right\} \nonumber \\
 & =n\cdot\min_{q}\sum_{i=1}^{n}q_{i}^{2}\E\{X_{i}\}\nonumber \\
 & =\frac{n}{\sum_{i=1}^{n}1/\E\{X_{i}\}},
\end{align}
where the last inequality follows from the optimal choice of $q$,
which is according to 
\begin{equation}
q_{i}=\frac{1/\E\{X_{i}\}}{\sum_{j=1}^{n}1/\E\{X_{j}\}}.
\end{equation}
More generally, whenever the function $f(u)=u^{\rho}$ is convex (namely,
for $\rho\notin(0,1)$), we can similarly obtain the inequality 
\begin{equation}
\E\left\{ \left(\sum_{i=1}^{n}X_{i}\right)^{\rho}\right\} \le\left[\sum_{i=1}^{n}\left(\E\{X_{i}^{\rho}\}\right)^{1/\rho}\right]^{\rho}.
\end{equation}
Note that no assumptions were imposed on the dependence/independence
among the RVs $\{X_{i}\}$.}
\end{example}

\subsection{Reverse Jensen Inequalities \label{subsec: reversejensen}}

Frequently, applied mathematicians, and especially information-theorists,
encounter a rather vexing situation where Jensen's inequality seems
to operate in the opposite direction of their desired results. This
observation has spurred significant research efforts aimed at developing
various versions of the so-called \emph{reverse Jensen inequality}
(RJI). A myriad of articles, including, but not limited to, \cite{ABZ21,BAT20,BDP01,Dragomir10,Dragomir13,JP00,KKC20a,KKC20b,Simic09,WGFS21},
have delved into this topic, showcasing its rich and evolving landscape.
In the majority of these works, the derived inequalities find practical
applications in diverse fields. Examples include establishing valuable
relationships between arithmetic and geometric means, deriving reverse
bounds on entropy, KL divergence, and more generally, Csisz{\'a}r's
$f$-divergence. Additionally, these inequalities have been extended
to reverse versions of the H\"{o}lder inequality, among other applications.
In many of the aforementioned papers, the primary results manifest
in the form of an upper bound on the difference $\E\{f(X)\}-f(\E\{X\})$,
where $f$ denotes a convex function and $X$ is a RV. It is worth
noting that these upper bounds predominantly rely on global properties
of the function $f$, such as its range and domain, rather than on
the underlying PDF  of $X$ or its probability mass function in the
discrete case. Ideally, a desirable characteristic of an RJI would
be its ability to provide tight bounds when the PDF of $X$ is highly
concentrated around its mean, akin to the well-known property of the
standard Jensen inequality: 
\begin{equation}
\E\left\{ f(X)\right\} \ge f(\E\{X\}).
\end{equation}
Such tightness in the presence of concentration around the mean is
a hallmark of the ordinary Jensen inequality, and it would be advantageous
for RJIs to exhibit a similar behavior under such conditions.

In \cite{me22}, we extend the concepts introduced in \cite{WGFS21},
providing a fresh perspective on the RJI landscape. Our contributions
encompass several novel variants of RJI, and what sets these apart
is their ability to exhibit the desired property of tightness in cases
of measure concentration, a characteristic we consistently emphasize.

Our journey in this section commences from the same foundational point
as found in the proof of \cite[Lemma 1]{WGFS21}. However, the subsequent
course of our derivation takes a significantly different path. This
novel approach leads to notably tighter bounds, which prove to be
eminently tractable and analyzable in a multitude of scenarios, as
we amply demonstrate. Expanding the horizon of our research, we venture
into the realm of functions involving more than one variable. These
functions exhibit convexity (or concavity) in each variable individually,
although they may not possess this property jointly across all variables.
This extension broadens the applicability of our findings and enhances
their relevance to multifaceted real-world problems. Furthermore,
building upon similar underlying principles, we extend our investigations
to derive upper and lower bounds on the expectations of functions
that do not necessarily exhibit convexity or concavity across their
entire domain. These diverse contributions collectively enrich the
toolbox of RJIs and broaden their potential utility in a wide array
of practical and theoretical contexts.

We commence our exploration from a foundational point that closely
resembles \cite[Lemma 1]{WGFS21}. Let $f:\reals^{+}\to\reals$ be
a concave function with $f(x)\ge f(0)$ for every $x\ge0$. Let $X$
be a non-negative RV with a finite mean, $\E\{X\}=\mu$. Then, 
\begin{equation}
\E\{f(X)\}\ge\sup_{a>0}\left[\frac{\mu}{a}\cdot f(a)+\left(1-\frac{\mu}{a}\right)\cdot f(0)-\frac{f(a)-f(0)}{a}\cdot\E\left\{ X\cdot\I[X>a]\right\} \right],\label{basicinequality}
\end{equation}
where $\I[X>a]$ denotes the indicator function of event $\{X>a\}$.
This foundational inequality sets the stage for our subsequent derivations.
The primary challenge at this juncture is to evaluate the term: 
\begin{equation}
q(a)\equiv\E\left\{ X\cdot\I[X>a]\right\} .
\end{equation}
In straightforward cases, the exact calculation of $q(a)$ is achievable
through closed-form expressions. Examples include scenarios where
the PDF of $X$ follows uniform, triangular, or exponential distributions,
among others. However, for the majority of cases that pique our interest,
obtaining an exact, closed-form expression for $q(a)$ becomes a formidable
task, if not an impossibility. Consequently, we must rely on upper
bounds to further constrain the RHS of \eqref{basicinequality}.

In situations where the computation of $q(a)$ eludes an exact closed-form
expression, we introduce two fundamental alternative approaches for
bounding $q(a)$. Both approaches share a common feature: When the
RV $X$ tightly concentrates around its mean $\mu$, even slight deviations
of $a$ from $\mu$ result in small values for $q(a)$. This characteristic
ensures that our bounds closely approach the value of $f(\mu)$. The
selection between these two approaches depends on the specific problem
under consideration and the feasibility of obtaining closed-form expressions
for the moments involved, if such expressions exist at all.
\begin{enumerate}
\item \emph{The Chernoff approach.} The first approach is to upper bound
the indicator function, $\I\{x>a\}$ by the exponential function $e^{s(x-a)}$
($s\ge0$), akin to the Chernoff bound. This results in 
\begin{equation}
q(a)\le\inf_{s\ge0}\E\{Xe^{s(X-a)}\}=\inf_{s\ge0}\left[e^{-as}\E\{Xe^{sX}\}\right]=\inf_{s\ge0}\left[e^{-as}\Phi'(s)\right]\triangleq q_{\text{Chernoff}}(a),
\end{equation}
where $\Phi'(s)$ is the derivative of the MGF, $\Phi(s)\triangleq\E\{e^{sX}\}$.
Thus, (\ref{basicinequality}) is further lower bounded as 
\begin{equation}
\E\left\{ f(X)\right\} \ge\sup_{a>0}\left[\frac{\mu}{a}\cdot f(a)+\left(1-\frac{\mu}{a}\right)\cdot f(0)-\frac{f(a)-f(0)}{a}\cdot q_{\text{Chernoff}}(a)\right].\label{eq: Chernoff-approach}
\end{equation}
This bound proves to be particularly valuable when the RV $X$ possesses
a finite MGF, denoted as $\Phi(s)$, within a certain range of positive
$s$ values. Furthermore, it is essential that $\Phi(s)$ is differentiable
within this range. To ensure the practicality of this bound, it is
crucial that $q_{\text{Chernoff}}(a)$ can be expressed in a reasonably
straightforward closed-form manner. A slight variation of the Chernoff
approach involves bounding not just the indicator function factor
but the entire function $x\cdot\I[x>a]$ by an exponential function
of the form $a\cdot e^{s(x-a)}$. To ensure the effectiveness of this
approach, we choose $s$ such that the derivative WRT $x$ at $x=a$
is not less than 1. This ensures that the exponential function is
at least tangential to the function $x\cdot\I[x>a]$ as $x$ approaches
$a$ from above. Mathematically, this condition can be expressed as
$as\ge1$, which implies that $s$ should be greater than or equal
to $1/a$. Thus, 
\begin{equation}
q(a)\le a\cdot\inf_{s\ge1/a}\{e^{-as}\Phi(s)\}\triangleq\tilde{q}_{\text{Chernoff}}(a)
\end{equation}
which, of course, may replace $q_{\text{Chernoff}}(a)$ in \eqref{eq: Chernoff-approach}.
The usefulness of this version of the bound is essentially under the
same circumstances as those of $q_{\text{Chernoff}}(a)$. It has the
small advantage that there is no need to differentiate $\Phi(s)$,
but the range of the optimization over $s$ is somewhat smaller. 
\item \emph{The Chebychev}--\emph{Cantelli approach}. According to this
approach, the function $x\cdot\I[x>a]$ is upper bounded by a quadratic
function, in the spirit of the Chebychev--Cantelli inequality, \emph{i.e.},
\begin{equation}
x\cdot\I[x>a]\le\frac{a(x+s)^{2}}{(a+s)^{2}},
\end{equation}
where the parameter $s\ge0$ is optimized under the constraint that
the derivative at $x=a$, which is $2a/(a+s)$, is at least 1 (again,
to be at least tangential to the function itself at $x\downarrow a$),
which is equivalent to the requirement, $s\le a$. In this case, denoting
$\sigma^{2}=\Var\{X\}$, we get 
\begin{equation}
q(a)\le\frac{a\E\left\{ (X+s)^{2}\right\} }{(a+s)^{2}}=\frac{a\left[\sigma^{2}+(\mu+s)^{2}\right]}{(a+s)^{2}},
\end{equation}
which, when minimized over $s\in[0,a]$, yields 
\begin{equation}
s^{*}=\min\left\{ a,\frac{\sigma^{2}}{a-\mu}-\mu\right\} ,
\end{equation}
and then the best bound is given by 
\begin{equation}
q(a)\le q_{\text{Cheb-Cant}}(a)\triangleq\begin{cases}
\frac{\sigma^{2}+(a+\mu)^{2}}{4a}, & a<a_{\text{c}}\\
\frac{a\sigma^{2}}{\sigma^{2}+(a-\mu)^{2}}, & a\ge a_{\text{c}}
\end{cases},
\end{equation}
where $a_{\text{c}}\triangleq\sqrt{\sigma^{2}+\mu^{2}}$.
\end{enumerate}
The Chernoff approach often outperforms the Chebychev--Cantelli approach
in many scenarios. Let us consider an example to illustrate this point.
Suppose we have a RV $X$ expressed as the sum of $n$ IID RVs, $Y_{1},Y_{2},\ldots,Y_{n}$,
all with the same mean $\mu_{Y}$, variance $\sigma_{Y}^{2}$, and
MGF $\Phi_{Y}(s)$. In this case, we can readily calculate that $\mu=n\mu_{Y}$,
$\sigma^{2}=n\sigma_{Y}^{2}$, and $\Phi(s)=[\Phi_{Y}(s)]^{n}$. Additionally,
for the sake of simplicity, let us assume that $f(0)=0$. Now, if
we aim to apply the Chebychev--Cantelli approach, we typically end
up with a bound that relies on the variance of $X$ and its mean,
which are both multiplied by $n$. This often results in a relatively
loose bound due to the dependence on the sample size $n$. On the
other hand, when we employ the Chernoff approach, we leverage the
MGF of $X$ and, consequently, the MGF of $Y_{i}$, which remains
unchanged as $n$ grows. This approach frequently yields tighter bounds,
even when $n$ is substantial. Thus, in cases like this, the Chernoff
approach tends to be more effective in providing more accurate and
meaningful bounds.

Suppose, for example, that $X=\sum_{i=1}^{n}Y_{i}$, where $Y_{1},\ldots,Y_{n}$
are IID RVs, all having mean $\mu_{Y}$, variance $\sigma_{Y}^{2}$
and MGF $\Phi_{Y}(s)$. Then, of course, $\mu=n\mu_{Y}$, $\sigma^{2}=n\sigma_{Y}^{2}$,
and $\Phi(s)=[\Phi_{Y}(s)]^{n}$. For simplicity suppose also that
$f(0)=0$. In this case, the Chernoff approach yields 
\begin{align}
\E\left\{ f\left(\sum_{i=1}^{n}Y_{i}\right)\right\}  & \ge\frac{n\mu_{Y}}{a}\cdot f(a)-\frac{f(a)}{a}\inf_{s\ge0}\left\{ e^{-sa}\frac{\dd}{\dd s}[\Phi_{Y}(s)]^{n}\right\} \nonumber \\
 & =\frac{n\mu_{Y}}{a}\cdot f(a)-\frac{nf(a)}{a}\inf_{s\ge0}\left\{ e^{-sa}[\Phi_{Y}(s)]^{n-1}\Phi_{Y}'(s)\right\} \nonumber \\
 & =\frac{nf(a)}{a}\left[\mu_{Y}-\inf_{s\ge0}\left\{ e^{-sa}[\Phi_{Y}(s)]^{n}\cdot\frac{\dd\ln\Phi_{Y}(s)}{\dd s}\right\} \right].
\end{align}
Now, if $Y_{1},Y_{2},\ldots$ obey a large-deviations principle, the
second term in the square brackets tends to zero exponentially for
the choice $a=n(\mu_{Y}+\epsilon)$ with arbitrarily small $\epsilon>0$.
In this case, let $s^{*}>0$ be the maximizer of $[s(\mu+\epsilon)-\ln\Phi_{Y}(s)]$,
and denote $I(\epsilon)=s^{*}(\mu+\epsilon)-\ln\Phi_{Y}(s^{*})$.
Then, 
\begin{equation}
\E\left\{ f\left(\sum_{i=1}^{n}Y_{i}\right)\right\} \ge\frac{f[(\mu_{Y}+\epsilon)n]}{\mu_{Y}+\epsilon}\left[\mu_{Y}-e^{-nI(\epsilon)}\frac{\dd\ln\Phi_{Y}(s)}{\dd s}\bigg|_{s=s^{*}}\right].\label{eq: iid}
\end{equation}
For large enough $n$, the second term in the square brackets becomes
negligible, and the lower bound becomes arbitrarily close to $f[(\mu_{Y}+\epsilon)n]\cdot\mu_{Y}/(\mu_{Y}+\epsilon)$.
On the other hand, Jensen's upper bound is $f(\mu_{Y}n)$. In some
cases, the difference is not very large, at least for asymptotic evaluations.
For example, if $f(x)=\ln(1+x)$, which is a frequently encountered
concave function in information theory, $\ln[1+n(\mu_{Y}+\epsilon)]\ge\ln n+\ln(\mu_{Y}+\epsilon)$,
whereas $\ln(1+n\mu_{Y})\le\ln n+\ln(\mu_{Y}+1/n)$, which are very
close for large $n$ and small $\epsilon>0$.

In the Chebychev--Cantelli approach, on the other hand, we have $a_{\text{c}}=\sqrt{n^{2}\mu_{Y}^{2}+n\sigma_{Y}^{2}}\sim n\mu_{Y}$
for large $n$. Thus, if we take $a=n(\mu_{Y}+\epsilon)>a_{\text{c}}$,
we have 
\begin{equation}
q_{\text{Cheb-Cant}}\left[n(\mu_{Y}+\epsilon)\right]=\frac{n\sigma_{Y}^{2}}{n\sigma_{Y}^{2}+n^{2}\epsilon^{2}}=\frac{\sigma_{Y}^{2}}{\sigma_{Y}^{2}+n\epsilon^{2}},
\end{equation}
which tends to zero, but only at the rate of $1/n$, as opposed to
the exponential decay in the Chernoff approach. Still, for large $n$,
the main term of the bound becomes asymptotically tight, as before.

In spite of the superiority of the Chernoff approach relative to the
Chebychev--Cantelli approach, as we now demonstrated, one should
keep in mind that there are also situations where the RV $X$ does
not have an MGF (\emph{i.e.}, when the PDF of $X$ has a heavy tail),
yet it does have a mean and a variance. In such cases, the Chebychev--Cantelli
approach is applicable while the Chernoff approach is not. But even
when the MGF exists, in certain cases, the calculation of the first
and the second moment are easier than the calculation of the exponential
moment.

We summarize our main finding this section so far in the following
inequality: 
\begin{equation}
\E\{f(X)\}\ge\sup_{a>0}\left[\frac{\mu}{a}\cdot f(a)+\left(1-\frac{\mu}{a}\right)\cdot f(0)-\frac{f(a)-f(0)}{a}\cdot q_{\min}(a)\right],
\end{equation}
where 
\begin{equation}
q_{\min}(a)\triangleq\min\left\{ q_{\text{Chernoff}}(a),\tilde{q}_{\text{Chernoff}}(a),q_{\text{Cheb-Cant}}(a)\right\} .
\end{equation}

We now demonstrate the lower bound in two information-theoretic application
examples. More examples can be found in \cite{me22}. The first example
concerns channel capacity.
\begin{example}[Capacity of the Gaussian channel with random SNR]
\emph{ Consider a zero-mean, circularly symmetric complex Gaussian
channel whose SNR, $Z$, is a RV (e.g., due to fading), known to both
the transmitter and the receiver. The capacity is given by $C=\E\{\ln(1+gZ)\}$,
where $g$ is a certain deterministic gain factor and the expectation
is WRT the randomness of $Z$. For simplicity, let us assume that
$Z$ is distributed exponentially, i.e., 
\begin{equation}
p_{Z}(z)=\theta e^{-\theta z},~~~z\ge0,
\end{equation}
where the parameter $\theta>0$ is given. In this case, $f(x)=\ln(1+gx)$,
$\mu=1/\theta$ and $q(a)$ can be easily derived in closed form,
to obtain 
\begin{equation}
q(a)=\theta\cdot\int_{a}^{\infty}ze^{-\theta z}\dd z=\left(a+\frac{1}{\theta}\right)\cdot e^{-\theta a}.
\end{equation}
Consequently, 
\begin{align}
C & \ge\sup_{a\ge1/\theta}\frac{\ln(1+ga)}{a}\bigg[\frac{1}{\theta}-\left(a+\frac{1}{\theta}\right)\cdot e^{-a\theta}\bigg]\nonumber \\
 & =\sup_{s\ge1}\bigg[\frac{1-(s+1)e^{-s}}{s}\bigg]\cdot\ln\left(1+\frac{gs}{\theta}\right),
\end{align}
whereas the Jensen upper bound is $C\le\ln(1+g/\theta)$.}
\end{example}
The next example belongs to the realm of universal source coding. 
\begin{example}[Universal source coding]
\emph{\label{exa: univsrccoding}  Let us delve into the evaluation
of the expected code length linked with the universal lossless source
code developed by Krichevsky and Trofimov \cite{KT81}. In essence,
this code serves as a universal solution for encoding memoryless sources.
In the binary context, at each time step $t$, it systematically assigns
probabilities to the next binary symbol based on a biased version
of the empirical distribution derived from the source data observed
up to that point, denoted as $s_{1},s_{2},\ldots,s_{t}$. To be more
specific, let us examine the ideal code-length function (measured
in nats): 
\begin{equation}
L(s^{n})=-\sum_{t=0}^{n-1}\ln Q(s_{t+1}|s_{1},\ldots,s_{t}),\label{eq: KT81a}
\end{equation}
where 
\begin{equation}
Q(s_{t+1}=s|s_{1},\ldots,s_{t})=\frac{N_{t}(s)+1}{t+2},\label{eq: KT81b}
\end{equation}
and $N_{t}(s)$, $s\in\{0,1\}$, is the number of occurrences of the
symbol $s$ in $(s_{1},\ldots,s_{t})$. Therefore, 
\begin{align}
\E\left\{ L(S^{n})\right\}  & =\sum_{t=0}^{n-1}\ln(t+2)-\sum_{t=0}^{n-1}\E\{\ln[N_{t}(S_{t+1})+1]\}\nonumber \\
 & =\ln[(n+1)!]-\sum_{t=0}^{n-1}\E\left\{ \ln\left(1+\sum_{i=0}^{t}\I[S_{i}=S_{t+1}]\right)\right\} \nonumber \\
 & =\ln[(n+1)!]-p\cdot\sum_{t=0}^{n-1}\E\left\{ \ln\left(1+\sum_{i=0}^{t}\I[S_{i}=1]\right)\right\} -\nonumber \\
 & \hphantom{==}(1-p)\cdot\sum_{t=0}^{n-1}\E\left\{ \ln\left(1+\sum_{i=0}^{t}\I[S_{i}=0]\right)\right\} ,
\end{align}
where $\I[\cdot]$ are indicator functions of the corresponding events
and where $p$ and $1-p$ are the probabilities of `1' and `0', respectively.
To establish an upper bound for $\E\{L(S^{n})\}$, one can now use
\eqref{eq: iid} for lower bounds for each of the terms: $\E\{\ln(1+\sum_{i=0}^{t}\I[S_{i}=1])\}$
and $\E\{\ln(1+\sum_{i=0}^{t}\I[S_{i}=0])\}$. }
\end{example}

\subsection{Jensen-Like Inequalities \label{subsec: jensenlike}}

In this section, which summarizes the main findings of \cite{me23a},
we consider inequalities that are founded upon a fundamental insight
closely tied to the derivation of the ordinary Jensen inequality.
This insight revolves around the relationship between a given convex
function, denoted as $f(x)$, and the tangential affine function,
$\ell(x)=f(a)+f'(a)(x-a)$. Here, $a$ is an arbitrary value within
the domain of $x$, and $f'(a)$ represents the derivative of $f$
at the point $x=a$ (assuming the differentiability of $f$ at that
point). By strategically choosing $a$ to be $\E\{X\}$ (the expected
value of the RV $X$) and subsequently taking expectations of both
sides of the inequality $f(X)\ge\ell(X)$, we can effortlessly establish
the traditional Jensen inequality. This crucially hinges on the fact
that $a_{*}=\E\{X\}$ constitutes the optimal selection of $a$ in
the context of maximizing $\E\{\ell(X)\}$ across all potential values
of $a$. This, in turn, furnishes us with the most stringent lower
bound within the scope of lower bounds for $\E\{f(X)\}$. However,
it is worth noting that the optimal choice of $a$ may differ when
we are dealing with more intricate expressions where the expectation
needs to be lower bounded. For instance, one might seek to establish
a lower bound for $\E\{g[f(X)]\}$, where $g$ is a monotonically
non-decreasing function, or $\E\{f(X)g(X)\}$, where $g$ is a non-negative
and/or convex function, or perhaps a combination of these conditions
and more. In such cases, the optimal choice of $a$ could deviate
from $\E\{X\}$.

To illustrate this point, let us examine the lower bound of $\E\{f(X)g(X)\}$,
where $g$ is a non-negative function. In this scenario, we can establish
the following inequality: 
\begin{equation}
\E\left\{ f(X)g(X)\right\} \ge\E\left\{ [f(a)+f'(a)(X-a)]g(X)\right\} .
\end{equation}
By optimizing the RHS over the parameter $a$, we can easily determine
the optimal choice for $a$, denoted as $a_{*}$: 
\begin{equation}
a_{*}=\frac{\E\{Xg(X)\}}{\E\{g(X)\}}.
\end{equation}
This result leads to the inequality: 
\begin{equation}
\E\left\{ f(X)g(X)\right\} \ge f\left(\frac{\E\{Xg(X)\}}{\E\{g(X)\}}\right)\cdot\E\left\{ g(X)\right\} .\label{eq: prod}
\end{equation}
This inequality proves valuable, provided that we can readily compute
both $\E\{g(X)\}$ and $\E\{Xg(X)\}$ for the given function $g$.
Our first example concerns a function that is intimately related to
the Shannon entropy.
\begin{example}[An entropy-related function]
\emph{\label{exa: entropyrelated} Letting $f(x)=-\ln x$ and $g(x)=x$,
$x>0$, we obtain 
\begin{align}
\E\left\{ -X\ln X\right\}  & \ge-\E\{X\}\cdot\ln\frac{\E\{X^{2}\}}{\E\{X\}}\nonumber \\
 & =-\E\{X\}\cdot\ln(\E\{X\})-\E\{X\}\cdot\ln\left(1+\frac{\Var\{X\}}{[\E\{X\}]^{2}}\right).
\end{align}
Notice that the function $-x\ln x$ exhibits concavity, rather than
convexity. Nevertheless, we establish a lower bound, not an upper
one, on its expectation, thereby unveiling a RJI. The right-most side
of the expression comprises two components: The initial term represents
the standard Jensen upper bound for $\E\{-X\ln X\}$, while the second
term accounts for the gap. This gap is contingent not only upon the
expectation of $X$ but also on its variance, reflecting the fluctuations
around $\E\{X\}$. Clearly, in scenarios where $\Var\{X\}=0$, the
second term disappears --- a logical outcome, as a degenerate RV
 causes Jensen's inequality to hold with equality, eliminating any
gap. This inequality promptly finds application in deriving a lower
bound for the expected empirical entropy of a sequence generated by
a memoryless source. Such an application holds significance within
the realm of universal source coding, as detailed in \cite{KT81}
(see more details in \cite{me23a}). }
\end{example}
Another important example is associated with moments.
\begin{example}[Bounds on moments]
\emph{\label{exa: moments}  Let $s$ and $t$ be two real numbers
whose difference, $s-t$, is either negative or larger than unity.
Now, let $g(x)=x^{t}$, and $f(x)=x^{s-t}$. Then, 
\begin{align}
\E\left\{ X^{s}\right\}  & =\E\left\{ X^{t}X^{s-t}\right\} \nonumber \\
 & \ge\left(\frac{\E\{X^{t+1}\}}{\E\{X^{t}\}}\right)^{s-t}\cdot\E\{X^{t}\}\nonumber \\
 & =\frac{(\E\{X^{t+1}\})^{s-t}}{(\E\{X^{t}\})^{s-t-1}}.
\end{align}
In particular, for $t=1$ and $s\notin(1,2)$, this becomes 
\begin{equation}
\E\left\{ X^{s}\right\} \ge\frac{\left(\E\{X^{2}\}\right)^{s-1}}{\left(\E\{X\}\right)^{s-2}}=\left[\E\{X\}\right]^{s}\cdot\left(1+\frac{\Var\{X\}}{\left[\E\{X\}\right]^{2}}\right)^{s-1}
\end{equation}
which is, once again, a bound that depends only on the first two moments
of $X$. For $s\in(0,1)$, the function $x^{s}$ exhibits concavity,
resulting in a RJI. Conversely, when $s\le0$ or $s\ge2$, the function
$x^{s}$ is convex, giving rise to an enhanced version of Jensen's
inequality. In this enhanced version, the first term, $[\E\{X\}]^{s}$,
corresponds to the standard Jensen inequality, while the second factor
quantifies the degree of enhancement. This enhancement is contingent
on the relative fluctuation term, $\Var\{X\}/[\E\{X\}]^{2}$. Naturally,
the extent of improvement hinges on the variance of $X$. When the
variance dwindles to zero, there is no room for improvement since
the standard Jensen inequality attains equality. In contrast, a larger
variance results in a wider gap between the conventional Jensen bound,
$[\E\{X\}]^{s}$, and the enhanced counterpart. This underscores the
importance of optimizing the parameter $a$, as opposed to the default
choice of $a=\E\{X\}$ in the standard Jensen inequality.}
\end{example}
Another family of Jensen-like bounds is associated with the product
of two non-negative convex functions. Let both $f$ and $g$ be non-negative
convex functions of $x\ge0$. Then, 
\begin{align}
\E\left\{ f(X)g(X)\right\}  & \ge\E\left\{ [f(a)+f'(a)(X-a)]\cdot g(X)\right\} \nonumber \\
 & =\left[f(a)-af'(a)\right]\cdot\E\left\{ g(X)\right\} +f'(a)\E\left\{ Xg(X))\right\} \nonumber \\
 & \ge\left[f(a)-af'(a)\right]\E\left\{ [g(b)+g'(b)(X-b)]\right\} +\nonumber \\
 & \hphantom{===}f'(a)\E\left\{ X\left[g(c)+g'(c)(X-c)\right]\right\} ~~~~~~~f(a)\ge af'(a)\ge0\nonumber \\
 & =\left[f(a)-af'(a)\right]\cdot\left[g(b)-bg'(b)+g'(b)\E\{X\}\right]+\nonumber \\
 & \hphantom{===}f'(a)\left[(g(c)-cg'(c))\E\{X\}+g'(c)\E\left\{ X^{2}\right\} \right].
\end{align}
Maximizing the right-most side over $a$, $b$ and $c$, one obtains
the inequality: 
\begin{equation}
\E\left\{ f(X)g(X)\right\} \ge f\left(\frac{\E\{X\}\cdot g(\E\{X^{2}\}/\E\{X\})}{g(\E\{X\})}\right)\cdot g\left(\E\{X\}\right).
\end{equation}

\begin{example}[Second moment of Gaussian capacity]
\emph{\label{exa: capacityvariance} Consider the example of the
AWGN channel with a random SNR, denoted as $Z$. In this context,
we aim to bound the variance of the capacity, denoted as $c(Z)$,
as a means to assess the fluctuations, particularly for applications
like bounding the outage probability. The variance of $c(Z)$ can
be expressed as follows: 
\begin{equation}
\Var\left\{ c(Z)\right\} =\E\left\{ c^{2}(Z)\right\} -\left[\E\{c(Z)\}\right]^{2}=\E\left\{ \ln^{2}(1+gZ)\right\} -\left[\E\left\{ \ln(1+gZ)\right\} \right]^{2}.
\end{equation}
To establish an upper bound for $\Var\{c(Z)\}$, we can derive upper
bounds for both $\E\{\ln^{2}(1+gZ)\}$ and a lower bound for $\E\{\ln(1+gZ)\}$.
For the former, we can utilize the inequality presented here, employing
$f(z)=g(z)=\ln(1+gz)$. This yields the following upper bound, relying
solely on the first two moments of $Z$: 
\begin{equation}
\E\left\{ \ln^{2}(1+gZ)\right\} \le\ln\left(1+g\E\{Z\}\right)\cdot\ln\left(1+\frac{g\E\{Z\}\ln(1+g\E\{Z^{2}\}/\E\{Z\})}{\ln(1+g\E\{Z\})}\right).
\end{equation}
Notably, the function $\ln^{2}(1+gx)$ is neither convex nor concave.
Nevertheless, our approach provides an upper bound that can be easily
computed, given the ability to calculate the first two moments of
$Z$.}
\end{example}
These are just a few out of many more examples provided in \cite{me23a}.
The main features of the results on Jensen-like inequalities in general,
are the following. Firstly, in many instances, such as the one mentioned
above, we can analytically determine the optimal value of a parameter
(e.g., $a$ in the preceding discussion). However, in cases where
closed-form optimization is not feasible, we have two viable options:
(i) Perform numerical optimization or (ii) select an arbitrary value
for $a$ and derive a valid lower bound. It is important to note that
a well-informed choice for $a$ can potentially yield a robust lower
bound. Secondly, these inequalities offer two distinct types of bounds:
(i) Bounds that necessitate computing the first two moments (or equivalently,
the first two cumulants) of the RV $X$, and (ii) bounds that require
calculating the MGF of $X$ and its derivative, or equivalently, the
cumulant generating function  of $X$ and its derivative. These moment
calculations are often straightforward, especially in scenarios where
$X$ is represented as the sum of IID RVs --- a common occurrence
in information-theoretic applications. It should also be noted that
the classes of Jensen-like inequalities we explore provide ample flexibility
for deriving lower bounds on functions that may not be inherently
convex, some may even be concave. This opens the door to an alternative
approach for RJIs, different than those discussed in the Section \ref{subsec: reversejensen}.
This can be achieved by representing the given function within one
of the discussed categories, such as a product of a convex function
and a non-negative function, a product of two non-negative convex
functions, or a composition of a monotone function and a convex function.
Finally, the Jensen-like inequalities possess the desirable property
of tightening as the RV $X$ becomes increasingly concentrated around
its mean, akin to the conventional Jensen inequality.

\newpage{}

\section{Summary, Outlook and Open Issues}

In this monograph, we have provided an analytical toolbox for information-theoretic
analysis. We have described a generalization of the method of types,
which allows to address settings that go beyond the finite alphabet
case, including the prominent example of Gaussian sources and channels,
possibly with memory. This allows to evaluate the volumes of various
high-dimensional sets, and thus also their probability. We have also
described a generalization of this method to distributions from exponential
families. Further generalizing and refining such extensions to broader
classes of distributions is an interesting path for future research.
We have then described the saddle-point method for integration, which
not only allows to evaluate the pre-exponent of volumes or probabilities,
it is also necessary in the evaluation of redundancy rates, and may
provide solutions in settings for which the method of types fails. 

We then continued to present the TCEM, for evaluating the exponential
behavior of random codes. The method is principled, allows to analyze
optimal decoders, and is guaranteed to provide exponentially tight
results. It also provides the best known exponents in diverse problem
settings. Future research may further explore additional settings,
e.g., the error exponent of the typical random code in multi-user
configurations \cite{el2011network}. An additional important future
research direction is to consider structured random-ensembles. The
TCEM method rely on the assumption that the codewords in the ensemble
are drawn at random, IID (or some variant of such a random ensemble).
For practical decoding algorithms, codes must have some structure,
e.g., linear codes over finite fields, lattice codes for real/complex-input
channels \cite{zamir2014lattice}, convolutional codes or trellis-codes,
or even well-defined structure such as turbo-codes \cite{berrou1993near},
LDPC codes \cite{richardson2008modern}, polar codes \cite{arikan2009channel},
and so on. It is of interest to develop methods, akin to the TCEM,
to accurately analyze the error exponents of such codes. In addition,
it is also of interest to explore methods inspired by the TCEM in
derivation of converse results, in the finite-blocklength regime \cite{Polyanskiy10},
in the moderate-deviations regime \cite{altuug2014moderate,polyanskiy2010moderate}
and so on. We have briefly mentioned a few such initial results, which
hints the possibility of enriching this direction. Finally, it is
also of interest to further delve into the optimization problems involved
in the computation of exponents obtained by the TCEM, and develop
efficient, and perhaps ``general-purpose'', solvers, to solve them. 

We then considered the tight evaluation of expectations of non-linear
functions of RVs, including integral representations and a few variants
of Jensen's inequality. These techniques are highly useful in information
theory, as information measures typically involve such expectations.
For RJI, we have emphasized that it approaches the standard Jensen
inequality, when the RV of interest is tightly concentrated around
its mean value. It is thus of interest to relate the RJI we considered
to concentration-of-measure ideas \cite{boucheron2013concentration}.

\newpage{}

\appendix

\section{Computation of the Exponent \label{sec:Computation-of-the}}

In this appendix, we describe two possible approaches to efficiently
compute or bound the exponents obtained using the TCEM. This aspect
is an indispensable part of the TCEM, since it is possible for exact
exponents to take a rather intricate formula. Indeed, recall that
the exponents are given by Csisz{\'a}r--K{\"o}rner-style formulas,
and thus involve constrained optimization over joint types. Thus,
a direct optimization, using an exhaustive search or general purpose
global optimization over the probability simplex may be prohibitively
complex. 

The first approach we consider is based on Lagrange duality \cite{boyd2004convex},
in which the original exponent optimization problem is considered
to be the \emph{primal} optimization problem. When deriving instead
the \emph{dual} optimization problem of the exponent, the result is
a Gallager-style bound, which is rather easy to compute and plot for
an entire range of rates (rather than for a specific rate). This is
especially useful in multiuser problems, for which even problem instances
with binary alphabets lead to optimization problems in non-trivial
dimensions. In some of the problems, the number of optimization variables
for the Gallager-style bound does not increase with the alphabet size
of the source or channel. The downside is that some lower bounds may
be necessary for the derivation, and even if not, the dual exponent
may not be tight if the primal optimization problem of the exponent
is not convex. The second approach is based on utilization of convex
optimization solvers. While the optimization problem involved in the
computation of the exponent may not be convex as is, in many cases
it is possible to develop a procedure that allows to compute it by
only solving convex optimization problems.

Moreover, typically, the primal problem involves mostly \emph{minimization}
operators (over joint types), while the dual problem involves \emph{maximization}
operators (over scalar parameters). From this aspect, the dual exponent
is preferable, because even a sub-optimal choice of the dual variables
leads to a valid bound on the exponent. Thus, e.g., a coarse exhaustive
search on the dual variables may be performed and still lead to a
tight bound. By contrast, the minimization in the primal problem must
be performed accurately in order to obtain a valid numerical value
of the exponent. Nonetheless, it also possible for the primal problem
to include a maximization operator (possibly intertwined between minimization
operators), and the same holds for such maximization problems ---
any sub-optimal choice leads to a valid bound. In fact, in some cases,
an educated guess for the maximizing primal variable may be proposed,
and in some settings it is possible to show that this choice is actually
optimal. 

\subsection{Exponent Computation by Lagrange Duality}

Lagrange duality is based on the \emph{minimax theorem} \cite{Sion1958minimax},
stating the minimax value of a functional convex in the minimization
variable and concave in the maximization variable equals to the maximin
value. We will next exemplify this technique on the random-coding
error exponent $E_{\text{rc,\ensuremath{\alpha}}}(R,P_{X})$ from
\eqref{eq: random coding exponent expression}, and derive a Lagrange
dual lower bound on its value. As we have seen, if we consider the
MMI rule, then the random-coding error exponent is greatly simplified
to the standard random-coding error exponent in \eqref{eq: random coding exponent CKM},
which only contains a minimization over $Q_{Y|X}$ (with the minimization
over $\tilde{Q}_{Y|X}$ removed). In accordance, it is not very difficult
to obtain a dual Lagrange form of this exponent. In order to demonstrate
a few other techniques that are generally useful for the TCE-based
exponents, we will next let $\alpha(\cdot)$ be general, yet restricted
to be a linear function of $Q_{XY}$, given by $\alpha(Q_{XY})\triangleq\sum_{x\in{\cal X},y\in{\cal Y}}\alpha(x,y)\cdot Q(x,y)$
(this includes, e.g., the ML decoder). 

Let us start by writing the objective function of $E_{\text{rc,\ensuremath{\alpha}}}(R,P_{X})$
using a dual variable $\rho\in\reals$ as
\begin{align}
E_{\text{rc,\ensuremath{\alpha}}}(R,P_{X}) & =\min_{Q_{Y|X},\tilde{Q}_{Y|X}}D(Q_{Y|X}||W|P_{X})+\left[I(P_{X}\times\tilde{Q}_{Y|X})-R\right]_{+}\label{eq: random coding primal}\\
 & =\min_{Q_{Y|X},\tilde{Q}_{Y|X}}D(Q_{Y|X}||W|P_{X})+\max\left\{ I(P_{X}\times\tilde{Q}_{Y|X})-R,0\right\} \nonumber \\
 & =\min_{Q_{Y|X},\tilde{Q}_{Y|X}}D(Q_{Y|X}||W|P_{X})+\max_{\rho\in[0,1]}\rho\cdot\left[I(P_{X}\times\tilde{Q}_{Y|X})-R\right]\nonumber \\
 & =\min_{Q_{Y|X},\tilde{Q}_{Y|X}}\max_{\rho\in[0,1]}D(Q_{Y|X}||W|P_{X})+\rho\cdot\left[I(P_{X}\times\tilde{Q}_{Y|X})-R\right].
\end{align}
Now, the objective function is linear, and hence concave, in the maximizing
variable $\rho$, and the interval $[0,1]$ is convex. Moreover, $D(Q_{Y|X}||W|P_{X})$
is convex in $Q_{Y|X}$ and $\rho\cdot I(P_{X}\times\tilde{Q}_{Y|X})$
is convex in $\tilde{Q}_{Y|X}$ (for $\rho\geq0$), hence the objective
functional is jointly convex in $(Q_{Y|X},\tilde{Q}_{Y|X})$. The
constraint set for $(Q_{Y|X},\tilde{Q}_{Y|X})$, given by 
\begin{equation}
\left\{ Q_{Y|X},\;\tilde{Q}_{Y|X}\colon(P_{X}\times Q_{Y|X})_{Y}=(P_{X}\times\tilde{Q}_{Y|X})_{Y},\;\alpha(P_{X}\times\tilde{Q}_{Y|X})\geq\alpha(P_{X}\times Q_{Y|X})\right\} ,
\end{equation}
is the intersection of an hyperplane and a half space. We also note
the implicit constraint that $Q_{Y|X}$ and $\tilde{Q}_{Y|X}$ are
conditional probabilities, \emph{i.e.}, $\sum_{y\in{\cal Y}}Q_{Y|X}(y|x)=\sum_{y\in{\cal Y}}\tilde{Q}_{Y|X}(y|x)=1$
for all $x\in{\cal X}$ and $Q_{Y|X}(y|x),\tilde{Q}_{Y|X}(y|x)\geq0$
for all $x\in{\cal X},y\in{\cal Y}$. These are also convex constraints,
and since the intersection of convex sets is convex, the constraint
set for $(Q_{Y|X},\tilde{Q}_{Y|X})$ is convex. So, the minimax theorem
\cite{Sion1958minimax} implies that 
\begin{equation}
E_{\text{rc,\ensuremath{\alpha}}}(R,P_{X})=\max_{\rho\in[0,1]}\min_{Q_{Y|X},\tilde{Q}_{Y|X}}D(Q_{Y|X}||W|P_{X})+\rho\cdot\left[I(P_{X}\times\tilde{Q}_{Y|X})-R\right]
\end{equation}
over the constraint set. We next focus on the inner minimization for
a given $\rho\in[0,1]$. Following Lagrange duality \cite[Chapter 5]{boyd2004convex},
we introduce dual variables $\lambda\geq0$ and $\{\nu(y)\}_{y\in{\cal Y}}\subset\reals$.
The variable $\lambda$ is for the inequality constraint $\alpha(P_{X}\times\tilde{Q}_{Y|X})\geq\alpha(P_{X}\times Q_{Y|X})$,
whereas the variables $\{\nu(y)\}_{y\in{\cal Y}}$ are for the constraint
of equal output marginals, that is, the $|{\cal Y}|$ constraints
$(P_{X}\times Q_{Y|X})_{Y}=(P_{X}\times\tilde{Q}_{Y|X})_{Y}$. Note
that the constraint that $Q_{Y|X}$ and $\tilde{Q}_{Y|X}$ are conditional
probability distributions is kept implicit. Hence, the minimization
of interest is 
\begin{align}
 & \min_{Q_{Y|X},\tilde{Q}_{Y|X}}\max_{\lambda\geq0}\max_{\{\nu(y)\}_{y\in{\cal Y}}}D(Q_{Y|X}||W|P_{X})+\rho\cdot\left[I(P_{X}\times\tilde{Q}_{Y|X})-R\right]\nonumber \\
 & \hphantom{==}+\sum_{y\in{\cal Y}}\nu(y)\cdot\left[\sum_{x\in{\cal X}}P_{X}(x)\left(\tilde{Q}_{Y|X}(y|x)-Q_{Y|X}(y|x)\right)\right]\nonumber \\
 & \hphantom{==}+\lambda\cdot\left[\sum_{x\in{\cal X}}\sum_{y\in{\cal Y}}\alpha(x,y)\cdot P_{X}(x)\left(Q_{Y|X}(y|x)-\tilde{Q}_{Y|X}(y|x)\right)\right].
\end{align}
The minimax theorem now implies that we may interchange the minimization
and maximization order. We next focus on the minimization, and begin
by expressing the mutual information term via the \emph{golden formula}
using an arbitrary probability distribution $S_{Y}$ on ${\cal Y}$,
as 
\begin{align}
I(P_{X}\times\tilde{Q}_{Y|X}) & =D(\tilde{Q}_{Y|X}||\tilde{Q}_{Y}|P_{X})-D(\tilde{Q}_{Y}||S_{Y})\nonumber \\
 & =\min_{S_{Y}}D(\tilde{Q}_{Y|X}||S_{Y}|P_{X}).
\end{align}
Using this relation and slightly re-organizing the objective function,
we are thus remain to minimize over $(Q_{Y|X},\tilde{Q}_{Y|X})$ the
functional 

\begin{align}
 & \min_{S_{Y}}D(Q_{Y|X}||W|P_{X})+\sum_{x\in{\cal X}}\sum_{y\in{\cal Y}}P_{X}(x)Q_{Y|X}(y|x)\cdot\left[-\nu(y)+\lambda\cdot\alpha(x,y)\right]\nonumber \\
 & +\rho D(\tilde{Q}_{Y|X}||S_{Y}|P_{X})+\sum_{x\in{\cal X}}\sum_{y\in{\cal Y}}P_{X}(x)\tilde{Q}_{Y|X}(y|x)\cdot\left[\nu(y)-\lambda\cdot\alpha(x,y)\right].
\end{align}
It can be noticed that the minimization over $Q_{Y|X}$ is decoupled
from the minimization over $\tilde{Q}_{Y|X}$, and each of them can
be solved directly. Alternatively, we may use \emph{Donsker--Varadhan's}
variational formula \cite[Corollary 4.15]{boucheron2013concentration},
\cite{donsker1983asymptotic}, stating that for any two probability
measures $P_{1}$and $P_{2}$ on ${\cal Z}$ and a function $f\colon{\cal Z}\to\reals$
that does not depend on $P_{1}$ 
\begin{equation}
\min_{P_{2}}\left\{ D(P_{2}||P_{1})+\E_{P_{2}}\left[f(Z)\right]\right\} =-\ln\E_{P_{1}}\left[e^{-f(Z)}\right].\label{eq: Donsker Varadhan}
\end{equation}
Letting $W(\cdot|x)$ denote the conditional output of the channel
given $x\in{\cal X}$. By employing \eqref{eq: Donsker Varadhan}
separately for each $x\in{\cal X}$ we get
\begin{align}
 & \min_{Q_{Y|X}}D(Q_{Y|X}||W|P_{X})+\sum_{x\in{\cal X}}\sum_{y\in{\cal Y}}P_{X}(x)Q_{Y|X}(y|x)\cdot\left[-\nu(y)+\lambda\cdot\alpha(x,y)\right]\nonumber \\
 & =\sum_{x\in{\cal X}}P_{X}(x)\cdot\left\{ \min_{Q_{Y|X=x}}D(Q_{Y|X=x}||W(\cdot|x))+\sum_{y\in{\cal Y}}Q_{Y|X}(y|x)\cdot\left[-\nu(y)+\lambda\cdot\alpha(x,y)\right]\right\} \nonumber \\
 & =-\sum_{x\in{\cal X}}P_{X}(x)\cdot\ln\left(\sum_{y\in{\cal Y}}W(y|x)\cdot e^{\nu(y)-\lambda\cdot\alpha(x,y)}\right).
\end{align}
Similarly, the minimization over $\tilde{Q}_{Y|X}$ leads to the value
\begin{align}
 & \rho\cdot\sum_{x\in{\cal X}}P_{X}(x)\cdot\left\{ \min_{\tilde{Q}_{Y|X=x}}\rho D(\tilde{Q}_{Y|X=x}||S_{Y})+\sum_{y\in{\cal Y}}\tilde{Q}_{Y|X}(y|x)\cdot\left[\nu(y)-\lambda\cdot\alpha(x,y)\right]\right\} \nonumber \\
 & =\min_{S_{Y}}-\rho\sum_{x\in{\cal X}}P_{X}(x)\cdot\ln\left(\sum_{y\in{\cal Y}}S_{Y}(y)\cdot e^{-\left[\nu(y)+\lambda\cdot\alpha(x,y)\right]/\rho}\right)\nonumber \\
 & \geq\min_{S_{Y}}-\rho\ln\left(\sum_{x\in{\cal X}}\sum_{y\in{\cal Y}}P_{X}(x)S_{Y}(y)\cdot e^{-\left[\nu(y)+\lambda\cdot\alpha(x,y)\right]/\rho}\right),\label{eq: a lower bound used in the dual expression}
\end{align}
where the inequality follows from convexity and Jensen inequality,
yet is \emph{not} guaranteed to be tight. Since $\rho\in[0,1]$, minimizing
this last term over $S_{Y}$ corresponds to maximizing 
\begin{equation}
\sum_{y\in{\cal Y}}S_{Y}(y)\sum_{x\in{\cal X}}P_{X}(x)\cdot e^{-\left[\nu(y)+\lambda\cdot\alpha(x,y)\right]/\rho}
\end{equation}
which, due to Schwarz--Cauchy inequality, occurs when 
\begin{equation}
S_{Y}(y)=\frac{\sum_{x\in{\cal X}}P_{X}(x)\cdot e^{-\left[\nu(y)+\lambda\cdot\alpha(x,y)\right]/\rho}}{\sum_{y\in{\cal Y}}\sum_{x\in{\cal X}}P_{X}(x)\cdot e^{-\left[\nu(y)+\lambda\cdot\alpha(x,y)\right]/\rho}}.
\end{equation}
The minimal value over $S_{Y}$ is then 
\begin{align}
 & \min_{S_{Y}}-\rho\ln\left(\sum_{x\in{\cal X}}\sum_{y\in{\cal Y}}P_{X}(x)S_{Y}(y)\cdot e^{-\left[\nu(y)+\lambda\cdot\alpha(x,y)\right]/\rho}\right)\nonumber \\
 & =-\rho\ln\left(\frac{\sum_{y\in{\cal Y}}\left(\sum_{x\in{\cal X}}P_{X}(x)e^{-\left[\nu(y)+\lambda\cdot\alpha(x,y)\right]/\rho}\right)^{2}}{\sum_{y\in{\cal Y}}\sum_{x\in{\cal X}}P_{X}(x)\cdot e^{-\left[\nu(y)+\lambda\cdot\alpha(x,y)\right]/\rho}}\right).
\end{align}
We thus conclude the dual lower bound 
\begin{align}
 & E_{\text{rc},\alpha}(R,P_{X})\nonumber \\
 & \geq-\sum_{x\in{\cal X}}P_{X}(x)\cdot\ln\left(\sum_{y\in{\cal Y}}W(y|x)\cdot e^{\nu(y)-\lambda\cdot\alpha(x,y)}\right)\nonumber \\
 & \hphantom{===}-\rho\ln\left(\frac{\sum_{y\in{\cal Y}}\left(\sum_{x\in{\cal X}}P_{X}(x)e^{-\left[\nu(y)+\lambda\cdot\alpha(x,y)\right]/\rho}\right)^{2}}{\sum_{y\in{\cal Y}}\sum_{x\in{\cal X}}P_{X}(x)\cdot e^{-\left[\nu(y)+\lambda\cdot\alpha(x,y)\right]/\rho}}\right),\label{eq: random coding dual primal bound}
\end{align}
for any choice of $\rho\in[0,1]$, $\lambda\geq0$ and $\{\nu(y)\}_{y\in{\cal Y}}\subset\reals$. 

Let us compare the primal optimization in \eqref{eq: random coding primal},
with the dual lower bound \eqref{eq: random coding dual primal bound}.
The primal problem is a minimization problem of dimension $2|{\cal X}|(|{\cal Y}|-1)$
over a constrained set $(Q_{Y|X},\tilde{Q}_{Y|X})$ (the constraints
further reduce the dimension by $|{\cal Y}|+1$). For the exact exponent,
this minimization must be accurately solved. By comparison, the dual
exponent is a lower bound on the exact exponent (recall \eqref{eq: a lower bound used in the dual expression}),
and can be maximized over dimension $|{\cal Y}|+2$. Nonetheless,
this maximization can be performed in a crude manner, since any choice
of the dual parameters leads to a valid lower bound on the exponent. 

For additional derivations of dual Lagrange exponents formulations
and Gallager-style bounds, see \cite[Exercise 10.24]{csiszar2011information}
and, in the context of the TCEM, see \cite{averbuch2018exact,scarlett2014expurgated,merhav2019lagrange}. 

\subsection{Exponent Computation Procedures with Convex Optimization Solvers}

As we have seen, we may write 
\begin{equation}
E_{\text{rc,\ensuremath{\alpha}}}(R,P_{X})=\max_{\rho\in[0,1]}\min_{Q_{Y|X},\tilde{Q}_{Y|X}}D(Q_{Y|X}||W|P_{X})+\rho\cdot\left[I(P_{X}\times\tilde{Q}_{Y|X})-R\right]
\end{equation}
and when $\alpha(Q_{XY})$ is a linear function of $Q_{XY}$, then
the constraints set of $(Q_{Y|X},\tilde{Q}_{Y|X})$ is convex. Hence,
the inner minimization problem is a convex optimization problem that
can be efficiently solved. However, in principle, it should be solved
for the continuous set of values $\rho\in[0,1]$. We next describe
an alternative method to evaluate $E_{\text{rc,\ensuremath{\alpha}}}(R,P_{X})$. 

Let us write $E_{\text{rc,\ensuremath{\alpha}}}(R,P_{X})=\min\{E_{-}(R),E_{+}(R\}$
where\footnote{For brevity, we omit the explicit dependence on the score $\alpha$
and the input distribution $P_{X}$.}

\begin{equation}
E_{-}(R)=\min_{Q_{Y|X},\tilde{Q}_{Y|X}}D(Q_{Y|X}||W|P_{X}),
\end{equation}
where the minimization is over the set 
\begin{equation}
\left\{ Q_{Y|X},\tilde{Q}_{Y|X}\colon Q_{Y}=\tilde{Q}_{Y},\;\alpha(P_{X}\times\tilde{Q}_{Y|X})\geq\alpha(P_{X}\times Q_{Y|X}),\;I(P_{X}\times\tilde{Q}_{Y|X})\leq R\right\} 
\end{equation}
and where
\begin{equation}
E_{+}(R)=\min_{Q_{Y|X},\tilde{Q}_{Y|X}}D(Q_{Y|X}||W|P_{X})+I(P_{X}\times\tilde{Q}_{Y|X})-R,
\end{equation}
where the minimization over the set 
\begin{equation}
\left\{ Q_{Y|X},\tilde{Q}_{Y|X}\colon Q_{Y}=\tilde{Q}_{Y},\;\alpha(P_{X}\times\tilde{Q}_{Y|X})\geq\alpha(P_{X}\times Q_{Y|X}),\;I(P_{X}\times\tilde{Q}_{Y|X})\geq R\right\} .
\end{equation}
Note that the only difference between $E_{-}(R)$ and $E_{+}(R)$
is the constraint $I(P_{X}\times\tilde{Q}_{Y|X})\gtreqless R$, and
due to the continuity of the objective function, we have included
the points $\{I(P_{X}\times\tilde{Q}_{Y|X})=R\}$ in both problems.
Now, since the KL divergence is also a convex function of $Q_{Y|X}$
it can be seen that the objective function is jointly convex in $\{Q_{Y|X},\tilde{Q}_{Y|X}\}$
for both optimization problems. Since $\alpha(Q_{XY})$ is a linear
function of $Q_{XY}$, the set $\{Q_{Y}=\tilde{Q}_{Y},\;\alpha(P_{X}\times\tilde{Q}_{Y|X})\geq\alpha(P_{X}\times Q_{Y|X})\}$
is a convex set. Furthermore, the set $\{I(P_{X}\times\tilde{Q}_{Y|X})\leq R\}$
is also a convex set, and thus so is its intersection with the previous
set. Consequently, the minimization problem of $E_{-}(R)$ is a convex
optimization problem \cite{boyd2004convex} (of dimension $2|{\cal X}|\times(|{\cal Y}|-1)$),
which can be efficiently solved, e.g., using solvers such as \texttt{CVX}
\cite{grant2009cvx}. By contrast, the minimization problem of $E_{+}(R)$
involves the set $\{I(P_{X}\times\tilde{Q}_{Y|X})\geq R\}$, which
is \emph{not} a convex set. 

We thus proceed as follows. First, let us solve $E_{+}(R)$ for $R=0$.
In this case, the constraint $I(P_{X}\times Q_{Y|X})\geq R$ is idle,
and so 
\begin{equation}
E_{+}(0)=\min_{Q_{Y|X},\tilde{Q}_{Y|X}\colon\;\alpha(P_{X}\times\tilde{Q}_{Y|X})\geq\alpha(P_{X}\times Q_{Y|X})}D(Q_{Y|X}||W|P_{X})+I(P_{X}\times\tilde{Q}_{Y|X}).
\end{equation}
This is a convex optimization problem, which can be efficiently solved.
Let us denote the solution of this problem as $(Q_{Y|X}^{(0)},\tilde{Q}_{Y|X}^{(0)})$.
Now, as long as $R\leq R_{\text{cr}}\triangleq I(\tilde{Q}_{Y|X}^{(0)})$,
then the objective function in $E_{+}(R)$ is minimized by the unconstrained
solution $(Q_{Y|X}^{(0)},\tilde{Q}_{Y|X}^{(0)})$, even if the constraint
$I(P_{X}\times Q_{Y|X})\geq R$ is imposed. For these rates it thus
holds that $E_{+}(R)=E_{+}(0)-R$. Now, if $R\geq R_{\text{cr}}$
then the unconstrained solution $(Q_{Y|X}^{(0)},\tilde{Q}_{Y|X}^{(0)})$
does not solve $E_{+}(R)$, and so the solution must be obtained on
the boundary $\{I(P_{X}\times\tilde{Q}_{Y|X})=R\}$. However, for
such rates
\begin{align}
 & E_{+}(R)\nonumber \\
 & =\min_{Q_{Y|X},\tilde{Q}_{Y|X}\colon\;\alpha(P_{X}\times\tilde{Q}_{Y|X})\geq\alpha(P_{X}\times Q_{Y|X}),\;I(P_{X}\times\tilde{Q}_{Y|X})=R}D(Q_{Y|X}||W|P_{X})+I(P_{X}\times\tilde{Q}_{Y|X})-R\nonumber \\
 & =\min_{Q_{Y|X},\tilde{Q}_{Y|X}\colon\;\alpha(P_{X}\times\tilde{Q}_{Y|X})\geq\alpha(P_{X}\times Q_{Y|X}),\;I(P_{X}\times\tilde{Q}_{Y|X})=R}D(Q_{Y|X}||W|P_{X})\nonumber \\
 & \geq\min_{Q_{Y|X},\tilde{Q}_{Y|X}\colon\;\alpha(P_{X}\times\tilde{Q}_{Y|X})\geq\alpha(P_{X}\times Q_{Y|X}),\;I(P_{X}\times\tilde{Q}_{Y|X})\leq R}D(Q_{Y|X}||W|P_{X})\nonumber \\
 & =E_{-}(R),
\end{align}
where the inequality holds since the feasible set is larger for $E_{-}(R)$.
Consequently, for rates $R\geq R_{\text{cr}}$, the exponent is given
by $\min\{E_{-}(R),E_{+}(R)\}=E_{-}(R)$. 

To conclude, despite the fact that the minimization problem of $E_{+}(R)$
is not a convex optimization problem, the exponent can be computed
for all rates by only solving convex optimization problems. To summarize,
this is done by the following procedure: (1) Solve the optimization
problem for $E_{+}(0)$, and compute the critical rate $R_{\text{cr}}$.
(2) Solve the optimization problem $E_{-}(R)$ for any $R>R_{\text{cr}}$.
The exponent is 
\begin{equation}
\begin{cases}
E_{+}(0)-R, & 0\leq R\leq R_{\text{cr}}\\
E_{-}(R), & R>R_{\text{cr}}
\end{cases}.
\end{equation}
Note that this method requires solving two convex optimization problems
at most for each rate, and the first one for finding $E_{+}(0)$ one
is common to all rates. 

For additional computational algorithms, see, for example, \cite[Section V]{etkin2009error}
for the computation of the exponent of the interference channel, \cite[Appendix A]{weinberger2014codeword}
for the exponents of joint detection and decoding, and \cite[Section VI]{weinberger2019reliability}
for exponents of distributed hypothesis testing.

\section{The Derivation of the Expurgated Exponent \label{sec:The-Expurgated-Exponent}}

In this appendix we outline the expurgation argument that follows
the TCEM method. The proof follows \cite[Appendix]{merhav2014list}.
Let us focus on a specific codeword index $m$. We showed in Section
\ref{subsec:Probabilistic-Properties-of} that, effectively, $\overline{N}_{m}(Q_{X\tilde{X}})\sim\text{Binomial}(e^{nR},e^{-nI(Q_{X\tilde{X}})})$.
Thus, we separate between \emph{typically populated} joint types ($I(Q_{X\tilde{X}})\leq R$)
and \emph{typically empty} joint types ($I(Q_{X\tilde{X}})>R$). First,
for the populated types, for any $\epsilon>0$, it holds by \eqref{eq: exponent of upper tail}
that
\begin{equation}
\Pr\left[\overline{N}_{m}(Q_{X\tilde{X}})\geq e^{n(R-I(Q_{X\tilde{X}})+\epsilon)}\right]\doteq e^{-n\cdot\infty}.
\end{equation}
By the union bound over exponentially number of codewords $e^{nR}$
and polynomial number of joint types, the event
\begin{equation}
{\cal F}\triangleq\left\{ \bigcup_{m=1}^{e^{nR}}\bigcup_{Q_{X\tilde{X}}\colon Q_{X}=Q_{\tilde{X}}=P_{X},\;I(Q_{X\tilde{X}})\geq R}\left\{ \overline{N}_{m}(Q_{X\tilde{X}})\geq e^{n(R-I(Q_{X\tilde{X}})+\epsilon)}\right\} \right\} 
\end{equation}
satisfies $\Pr[{\cal F}]=e^{-n\cdot\infty}$. Since by \eqref{eq: left tail populated type}
the lower tail also similarly decays double-exponentially, for the
sake of exponent analysis, the TCE are \emph{effectively} deterministic,
for all codewords in the codebook and all joint types with $I(Q_{X\tilde{X}})\leq R$,
and is given by 
\begin{equation}
\overline{N}_{m}(Q_{X\tilde{X}})\doteq e^{n[R-I(Q_{X\tilde{X}})]}.
\end{equation}
Second, for the empty types for which $I(Q_{X\tilde{X}})>R$, it holds
by \eqref{eq: exponent of upper tail} that 

\begin{equation}
\Pr\left[\overline{N}_{m}(Q_{X\tilde{X}})\geq1\right]\doteq e^{-n[I(Q_{X\tilde{X}})-R]},
\end{equation}
which is exponentially small. Thus, we do not expect to observe other
codewords $\tilde{m}\neq m$ which have joint type $Q_{X\tilde{X}}$
with $\boldsymbol{X}_{m}$. Indeed, the event
\begin{equation}
{\cal E}_{m}\triangleq\left\{ \bigcup_{Q_{X\tilde{X}}\colon Q_{X}=Q_{\tilde{X}}=P_{X},\;I(Q_{X\tilde{X}})>R}\left\{ \overline{N}_{m}(Q_{X\tilde{X}})\geq1\right\} \right\} 
\end{equation}
is the event that the $m$th codeword is a-typical neighboring codewords,
in the sense that there exists a $Q_{X\tilde{X}}$ with $I(Q_{X\tilde{X}})>R$
and at least one neighboring codeword $\boldsymbol{X}_{\tilde{m}}$
so that $\hat{Q}_{\boldsymbol{X}_{m}\boldsymbol{X}_{\tilde{m}}}=Q_{X\tilde{X}}$.
By the union bound, since the number of joint types increases polynomially
with $n$, $p_{n}\triangleq\Pr[{\cal E}_{m}]\doteq e^{-n(I(Q_{X\tilde{X}})-R)}$.
Thus, on the average, we expect that $p_{n}e^{nR}$ codewords will
have such a-typical neighboring codewords. So, the event 
\begin{equation}
{\cal E}^{*}\triangleq\left\{ \frac{1}{e^{nR}}\sum_{m=1}^{e^{nR}}\I\{{\cal E}_{m}\}\geq2p_{n}\right\} ,
\end{equation}
in which more than $2p_{n}e^{nR}$ have such a-typical neighboring
codeword has low probability. Indeed, by Markov's inequality, which
does not require independence of the events $\{{\cal E}_{m}\}$, implies
that $\Pr[{\cal E}^{*}]\leq\frac{1}{2}$. Hence, with probability
larger than $1/2-\Pr[{\cal F}]\geq1/2-e^{-n\infty}$, both ${\cal F}^{c}$
and $[{\cal E}^{*}]^{c}$ hold. We thus may choose a codebook ${\cal C}_{n}$
that belongs to the event ${\cal F}^{c}\cap[{\cal E}^{*}]^{c}$. The
number of codewords in this codebook for which $\I\{{\cal E}_{m}\}=1$
is less than $2p_{n}e^{nR}$. Thus, we can \emph{expurgate} those
codewords from the codebook, and obtain a new codebook ${\cal C}_{n}^{*}$
which satisfies: (1) Its size is larger than $|{\cal C}_{n}^{*}|\geq e^{nR}(1-2p_{n})\doteq e^{nR}$.
(2) Its TCEs $\overline{N}_{m}^{*}(Q_{X\tilde{X}})$ are only smaller
than those of the original codebook, and specifically, $\overline{N}_{m}^{*}(Q_{X\tilde{X}})=0$
for all $Q_{X\tilde{X}}$ with $I(Q_{X\tilde{X}})>R$. (3) $\overline{N}_{m}^{*}(Q_{X\tilde{X}})\leq e^{n(R-I(Q_{X\tilde{X}})+\epsilon)}$
for all $Q_{X\tilde{X}}$ with $I(Q_{X\tilde{X}})\leq R$. 

For such a codebook, and after taking $\epsilon\downarrow0$, the
error probability bound in \eqref{eq: expurgated exponent first}
is given by
\begin{equation}
P_{\mathsf{e}}\leq\exp\left[-n\cdot E_{\text{ex}}(R,P_{X})\right],
\end{equation}
where $E_{\text{ex}}(R,P_{X})$ is as defined in \eqref{eq: expurgated exponent}.

Compared to the TCEM, the properties of codebook ${\cal C}_{n}^{*}$
traditionally follow from the \emph{packing lemma} \cite[Exercise 10.2]{csiszar2011information},
\cite{csiszar1977new} (which is somewhat similar) or from a \emph{graph
decomposition lemma} \cite[Corollary to Lemma 2]{csiszar1981graph}.
In the latter case, equipped with the existence of such a codebook,
\cite{csiszar1981graph} derived a bound for decoders with general
score $\alpha(\cdot)$, and when $\alpha(\cdot)$ is set to be the
ML decoder, then this exponent is shown to improve both the random-coding
error exponent and the expurgated exponent.

\section{Proofs for Section \ref{subsec:Probabilistic-Properties-of} \label{sec:Proofs-for-TCE-properties}}

Before proving Theorems \ref{thm: tail probabilities of N}, \ref{thm: The moments of N}
and \ref{thm: Intersection of tail events}, we recall the following
Chernoff tail bounds of a binomial RV $X\sim\text{Binomial}(m,p)$.
If $r>p$ then $rm>\E[X]=pm$ and so the probability of the upper
tail is 
\begin{equation}
e^{-m\cdot D(r||p)-o(m)}\leq\Pr\left[X>rm\right]\leq e^{-m\cdot D(r||p)},
\end{equation}
where $D(r||p)\triangleq r\ln\frac{r}{p}+(1-r)\ln\frac{(1-r)}{(1-p)}$
is the binary KL divergence. If $r<p$ then this probability $\Pr[X>rm]\geq\Pr[X>\lfloor\E[X]\rfloor]\geq1/2$,
and the so the exponent is zero. Similarly, if $r<p$ then the probability
of the lower tail is 
\begin{equation}
e^{-m\cdot D(r\|p)-o(m)}\leq\Pr\left[X<rm\right]\leq e^{-m\cdot D(r\|p)},
\end{equation}
 and if $r>p$ then the exponent is zero. 

We will also need the following simple lemma regarding the KL divergence. 
\begin{lem}
\label{lem: KL divergence expansion}Let $\{a_{n},b_{n}\}$ be sequences
in $(0,1)$ such that $a_{n}=o(1)$ and $b_{n}=o(1)$. Then, 
\begin{equation}
D(a_{n}||b_{n})\sim\begin{cases}
b_{n} & \frac{a_{n}}{b_{n}}=o(1)\\
a_{n}\ln\frac{a_{n}}{b_{n}}, & \frac{a_{n}}{b_{n}}=\omega(1)
\end{cases},\label{eq: asymptotic expansion binary KL}
\end{equation}
where for a sequence $\{c_{n}\}$, the notation $c_{n}=o(1)$ means
that $\lim_{n\to\infty}c_{n}=0$ and the notation $c_{n}=\omega(1)$
means that $\lim_{n\to\infty}c_{n}=\infty$. 
\end{lem}
\begin{proof}
We use the expansion $\ln(1+x)=x+\Theta(x^{2})$ throughout. If $\frac{a_{n}}{b_{n}}=o(1)$
then it holds that
\begin{align}
(1-a_{n})\ln\left[\frac{1-a_{n}}{1-b_{n}}\right] & =(1-a_{n})\ln(1-a_{n})-(1-a_{n})\ln(1-b_{n})\nonumber \\
 & =-a_{n}(1-a_{n})+\Theta(a_{n}^{2})+b_{n}(1-a_{n})+\Theta(b_{n}^{2})\nonumber \\
 & =(b_{n}-a_{n})(1-a_{n})+\Theta(b_{n}^{2})\nonumber \\
 & =b_{n}\cdot\left[\left(1-\frac{a_{n}}{b_{n}}\right)-a_{n}(1-a_{n})+\Theta(b_{n}^{2})\right]\nonumber \\
 & \sim b_{n},
\end{align}
and so for all $n$ large enough
\begin{equation}
\left|a_{n}\ln\frac{a_{n}}{b_{n}}\right|=a_{n}\ln\frac{b_{n}}{a_{n}}=-b_{n}\cdot\frac{a_{n}}{b_{n}}\ln\frac{a_{n}}{b_{n}}=-o(b_{n})
\end{equation}
since $\lim_{t\downarrow0}t\ln t=0$. This is negligible compared
to the first term. 

If $\frac{a_{n}}{b_{n}}=\omega(1)$ then
\begin{align}
\left|(1-a_{n})\ln\left(\frac{1-a_{n}}{1-b_{n}}\right)\right| & =\left|(1-a_{n})\ln(1-a_{n})-(1-a_{n})\ln(1-b_{n})\right|\nonumber \\
 & =\left|(1-a_{n})\left[-a_{n}+\Theta(a_{n}^{2})+b_{n}+\Theta(b_{n}^{2})\right]\right|\nonumber \\
 & =\Theta(a_{n}),
\end{align}
which is negligible compared to $a_{n}\ln\frac{a_{n}}{b_{n}}=\omega(a_{n})$.
\end{proof}
We are now ready to prove Theorem \ref{thm: tail probabilities of N}. 
\begin{proof}[Proof of Theorem \ref{thm: tail probabilities of N}]
 In the case of a TCE, we are dealing with both an exponential number
of trials and an exponentially decaying success probability, and thus
consider the events $\{N>e^{n\lambda}\}$ and $\{N<e^{n\lambda}\}$
for some $\lambda\in\mathbb{R}$. Throughout, we will use the asymptotic
expansion of the binary KL divergence in Lemma \ref{lem: KL divergence expansion}. 

We distinguish between two cases: 
\begin{enumerate}
\item If $A>B$ then the mean value $\E[N]=e^{n(A-B)}$ is exponentially
large. For the upper tail, we assume $\lambda>A-B$, for which 
\begin{equation}
\Pr\left[N>e^{n\lambda}\right]\leq\exp\left[-e^{nA}\cdot D(e^{-n(A-\lambda)}||e^{-nB})\right].
\end{equation}
Since $A-B<\lambda$ then $\nicefrac{e^{-n(A-\lambda)}}{e^{-nB}}=\omega(1)$
and the exponent is 
\begin{align}
e^{nA}\cdot D(e^{-n(A-\lambda)}||e^{-nB}) & \sim e^{nA}e^{-n(A-\lambda)}\ln\frac{e^{-n(A-\lambda)}}{e^{-nB}}\nonumber \\
 & =n(\lambda-(A-B))e^{n\lambda}.
\end{align}
Thus, the right-tail probability decays double-exponentially. Similarly,
for the lower tail, we assume $\lambda<A-B$, for which 
\begin{equation}
\Pr\left[N<e^{n\lambda}\right]\leq\exp\left[-e^{nA}\cdot D(e^{-n(A-\lambda)}||e^{-nB})\right].
\end{equation}
Since $A-B>\lambda$ then $\nicefrac{e^{-n(A-\lambda)}}{e^{-nB}}=o(1)$
and the exponent is 
\begin{align}
e^{nA}\cdot D(e^{-n(A-\lambda)}||e^{-nB}) & \sim e^{n(A-B)}.
\end{align}
Thus, the lower-tail probability also decays double-exponentially. 
\item If $B>A$ then the mean value $\E[N]=e^{-n(B-A)}\leq1$ is exponentially
small. For the upper tail, we set $\lambda>0>A-B$ and obtain a double-exponentially
decay, exactly as in the previous case. Next, as $N$ is integer,
for $\lambda\leq0$, Markov's inequality implies that 
\begin{equation}
\Pr\left[N>e^{n\lambda}\right]=\Pr\left[N\ge1\right]\leq\E[N]=\exp\left[-n(B-A)\right].
\end{equation}
On the other hand,
\begin{align}
\Pr\left[N>e^{n\lambda}\right] & \geq\Pr\left[N=1\right]={e^{nA} \choose 1}\cdot e^{-nB}\cdot(1-e^{-nB})^{e^{nA}-1}\nonumber \\
 & =e^{-n(B-A)}\cdot(1-e^{-nB})^{e^{nA}-1}\nonumber \\
 & \sim\exp\left[-n(B-A)\right],
\end{align}
which shows that Markov's inequality is exponentially tight in this
case, and hence $\Pr[N>e^{n\lambda}]\doteq e^{-n(B-A)}$. The variable
$N$ has no lower tail since the above implies that $\Pr[N=0]\geq1-e^{-n(B-A)}$
. 
\end{enumerate}
Combining the two cases leads to the claimed result.
\end{proof}
We next prove Theorem \ref{thm: The moments of N}.
\begin{proof}[Proof of Theorem \ref{thm: The moments of N}]
 We separate again between two cases, depending on the sign of $A-B$. 
\begin{enumerate}
\item If $A>B$ then we know that any exponential deviation from the mean
leads to a double-exponentially decay. Hence, for any $\lambda>A-B$
\begin{align}
\E\left[N^{s}\right] & =\Pr[N\leq e^{n\lambda}]\cdot\E\left[N^{s}|N\leq e^{n\lambda}\right]+\Pr[N>e^{n\lambda}]\cdot\E\left[N^{s}|N\geq e^{n\lambda}\right]\nonumber \\
 & \dot{\leq}e^{n\lambda s}+e^{-n\cdot\infty}\cdot e^{nsA}\nonumber \\
 & \doteq e^{n\lambda s},
\end{align}
where we have used the fact that $N\leq e^{nA}$ with probability
$1$, and write $e^{-n\cdot\infty}$ for a probability that decays
super-exponentially. Taking the limit $\lambda\downarrow A-B$ shows
that
\begin{equation}
\E\left[N^{s}\right]\dot{\leq}e^{n(A-B)s}.
\end{equation}
A matching lower bound can be derived in an analogous way: For any
$\lambda<A-B$
\begin{align}
\E\left[N^{s}\right] & =\Pr[N\geq e^{n\lambda}]\cdot\E\left[N^{s}|N\geq e^{n\lambda}\right]+\Pr[N<e^{n\lambda}]\cdot\E\left[N^{s}|N<e^{n\lambda}\right]\nonumber \\
 & \geq\left[1-\Pr[N<e^{n\lambda}]\right]\cdot e^{n\lambda s}\nonumber \\
 & =\left[1-e^{-n\cdot\infty}\right]\cdot e^{n\lambda s},
\end{align}
after taking the limit $\lambda\uparrow A-B$. Hence, 
\begin{equation}
\E\left[N^{s}\right]\doteq e^{n(A-B)s}.
\end{equation}
\item If $A<B$ then we take $\lambda>0$ to obtain 
\begin{align}
\E\left[N^{s}\right] & =\Pr[1\leq N\leq e^{n\lambda}]\cdot\E\left[N^{s}|1\leq N\leq e^{n\lambda}\right]+\Pr[N>e^{n\lambda}]\cdot\E\left[N^{s}|N\geq e^{n\lambda}\right]\nonumber \\
 & \leq\Pr[N\geq1]\cdot e^{n\lambda}+e^{-n\cdot\infty}\cdot e^{nsA}\nonumber \\
 & \dot{\leq}e^{-n(B-A)}\cdot e^{n\lambda}.
\end{align}
Taking the limit $\lambda\downarrow0$ shows that
\begin{equation}
\E\left[N^{s}\right]\dot{\leq}e^{-n(B-A)}.
\end{equation}
A lower bound is obtained by 
\begin{equation}
\E\left[N^{s}\right]\geq\Pr[N=1]\cdot1^{s}\geq[1+o(1)]\cdot e^{-n(B-A)},
\end{equation}
which shows that the upper bound is tight. 
\end{enumerate}
Combining the two cases leads to the claimed result.
\end{proof}
We finally prove Theorem \ref{thm: Intersection of tail events}. 
\begin{proof}[Proof of Theorem \ref{thm: Intersection of tail events}]
If there is a $j^{*}\in[k_{n}]$ so that $B_{j^{*}}<A_{j^{*}}$ and
$\lambda<A_{j^{*}}-B_{j^{*}}$ then $\Pr[N_{j^{*}}<e^{n\lambda}]\doteq e^{-n\cdot\infty}$.
So, 
\begin{equation}
\Pr\left[\bigcap_{j=1}^{k_{n}}\left\{ N_{j}<e^{n\lambda}\right\} \right]\leq\max_{1\leq j\leq k_{n}}\Pr\left[N_{j}<e^{n\lambda}\right]\doteq e^{-n\cdot\infty}.\label{eq: probability of intersection upper tails first case}
\end{equation}
Otherwise, if all $j=1,\ldots,k_{n}$ it holds that either $B_{j}>A_{j}$
or $\lambda>A_{j}-B_{j}$ then \eqref{eq: exponent of upper tail}
implies that $\Pr[N_{j}>e^{n\lambda}]<e^{-n\infty}$ for all $j=1,\ldots,k_{n}$.
Thus, from the union bound, as $n\to\infty$ 
\begin{align}
\Pr\left[\bigcap_{j=1}^{k_{n}}\left\{ N_{j}\leq e^{n\lambda}\right\} \right] & =1-\Pr\left[\bigcup_{j=1}^{k_{n}}\left\{ N_{j}>e^{n\lambda}\right\} \right]\nonumber \\
 & \geq1-\sum_{j=1}^{k_{n}}\Pr\left[N_{j}>e^{n\lambda}\right]\nonumber \\
 & \geq1-k_{n}\cdot\max_{1\leq j\leq k_{n}}\Pr\left[N_{j}>e^{n\lambda}\right]\nonumber \\
 & \geq1-k_{n}\cdot e^{-\min_{1\leq j\leq k_{n}}E_{j}}\nonumber \\
 & \to1.\label{eq: probability of intersection upper tails second case}
\end{align}
Combining \eqref{eq: probability of intersection upper tails first case}
and \eqref{eq: probability of intersection upper tails second case}
leads to the stated claim.
\end{proof}
\newpage{}

\bibliographystyle{IEEEtran}
\bibliography{IT_tools}

\end{document}